\documentclass[12pt]{article}
\usepackage{amsmath}
\usepackage{graphicx,psfrag,epsf}
\usepackage{enumerate}
\usepackage{natbib}
\usepackage{url} 

\newcommand{\blind}{0}

\usepackage{color, colortbl}
\usepackage{amsmath, amssymb, amsthm, ascmac, authblk, bm, booktabs, caption, comment, enumerate, latexsym, lscape, longtable, mathabx, mathrsfs, mathtools, multirow, natbib, newtxtext, psfrag, setspace, subcaption, thmtools}
\usepackage[stable]{footmisc}
\definecolor{mycolor}{cmyk}{0.85, 0.21, 0, 0.06}
\definecolor{Gray}{gray}{0.9}
\usepackage[colorlinks=true, linkcolor=mycolor, filecolor=mycolor, urlcolor=mycolor, citecolor=mycolor, backref = page, driverfallback=dvipdfmx]{hyperref}
\usepackage[margin = 2.3cm]{geometry}

\usepackage{clipboard}

\DeclareMathOperator*{\argmin}{argmin}

\DeclareMathOperator*{\bE}{\mathbb{E}}
\DeclareMathOperator*{\bP}{\mathbb{P}}
\DeclareMathOperator*{\Cov}{\mathrm{Cov}}

\DeclareMathOperator*{\Var}{\mathrm{Var}}

\DeclareMathOperator*{\CATT}{\mathrm{CATT}}
\DeclareMathOperator*{\IPW}{\mathrm{IPW}}
\DeclareMathOperator*{\OR}{\mathrm{OR}}
\DeclareMathOperator*{\DR}{\mathrm{DR}}

\DeclareMathOperator*{\nev}{\mathrm{nev}}
\DeclareMathOperator*{\ny}{\mathrm{ny}}

\DeclareMathOperator*{\sub}{\mathrm{sub}}

\DeclareMathOperator*{\SE}{\mathrm{SE}}

\DeclareMathOperator*{\US}{\mathrm{US}}
\DeclareMathOperator*{\es}{\mathrm{es}}
\DeclareMathOperator*{\sel}{\mathrm{sel}}
\DeclareMathOperator*{\cc}{\mathrm{c}}
\DeclareMathOperator*{\W}{\mathrm{W}}
\DeclareMathOperator*{\OO}{\mathrm{O}}
\DeclareMathOperator*{\OW}{\mathrm{OW}}

\DeclareMathOperator*{\diag}{\mathrm{diag}}
\DeclareMathOperator*{\en}{\textbf{--}}

\renewcommand{\hat}{\widehat}
\renewcommand{\tilde}{\widetilde}

\theoremstyle{definition}
\newtheorem{theorem}{Theorem}
\newtheorem{assumption}{Assumption}

\newtheorem{lemma}{Lemma}

\declaretheorem[style=definition]{remark}

\numberwithin{equation}{section}

\begin{document}

\def\spacingset#1{\renewcommand{\baselinestretch}%
{#1}\small\normalsize} \spacingset{1}


\if0\blind
{
  \title{\bf Doubly Robust Uniform Confidence Bands for Group-Time Conditional Average Treatment Effects in Difference-in-Differences}
  \author{Shunsuke Imai\vspace{-0.35cm}\thanks{Yoshida Honmachi, Sakyo, Kyoto, 606-8501, Japan. Email: \href{mailto:imai.shunsuke.57n@st.kyoto-u.ac.jp}{imai.shunsuke.57n@st.kyoto-u.ac.jp}}\\
    Graduate School of Economics, Kyoto University\\
    and \\
    Lei Qin\thanks{471 Ebigase, Higashi-ku, Niigata, 950-8680, Japan. Email: \href{mailto:qinlei6d@unii.ac.jp}{qinlei6d@unii.ac.jp}}\\
    Faculty of International Economic Studies, University of Niigata Prefecture\\
    and \\
    Takahide Yanagi\thanks{Corresponding author: Yoshida Honmachi, Sakyo, Kyoto, 606-8501, Japan. Email: \href{mailto:yanagi@econ.kyoto-u.ac.jp}{yanagi@econ.kyoto-u.ac.jp}}\\
    Graduate School of Economics, Kyoto University\\
    }
  \maketitle
} \fi

\if1\blind
{
  \bigskip
  \bigskip
  \bigskip
  \begin{center}
    {\LARGE\bf Doubly Robust Uniform Confidence Bands for Group-Time Conditional Average Treatment Effects in Difference-in-Differences}
\end{center}
  \medskip
} \fi

\bigskip 
\begin{abstract}
    We consider a panel data analysis to examine the heterogeneity in treatment effects with respect to groups, periods, and a pre-treatment covariate of interest in the staggered difference-in-differences setting of \cite{callaway2021difference}.
    Under standard identification conditions, a doubly robust estimand conditional on the covariate identifies the group-time conditional average treatment effect given the covariate.
    Focusing on the case of a continuous covariate, we propose a three-step estimation procedure based on nonparametric local polynomial regressions and parametric estimation methods.
    Using uniformly valid distributional approximation results for empirical processes and weighted/multiplier bootstrapping, we develop doubly robust inference methods to construct uniform confidence bands for the group-time conditional average treatment effect function and a variety of useful summary parameters.
    The accompanying R package {\ttfamily \href{https://tkhdyanagi.github.io/didhetero/}{didhetero}} allows for easy implementation of our methods.
\end{abstract}

\noindent%
\vfill

\newpage
\spacingset{1.8} 


\section{Introduction} \label{sec:introduction}

Difference-in-differences (DiD) is a powerful quasi-experimental approach to estimate meaningful treatment parameters.
The recent DiD literature predominantly contributes to the development of identification and estimation methods in the \textit{staggered adoption} case, where each unit continues to receive a binary treatment after the initial treatment receipt.
\cite{callaway2023difference}, \cite{de2023survey}, \cite{roth2023s}, and \cite{sun2022linear} review recent contributions.

\phantomsection\label{page:AE-2-1}\Copy{AE-2-1}{
    In empirical research using the DiD method, it is essential to understand the heterogeneity in treatment effects with respect to covariates, as well as ``groups'' and periods.
    As a concrete empirical example, suppose that we are interested in assessing whether and how minimum wage increases reduce poverty using county-level panel data.
    In this scenario, it is important to understand how the instantaneous and dynamic effects of minimum wage increases on the unemployment rate depend on the ``pre-treatment'' poverty rate.
    For example, if minimum wage increases result in significant lasting job losses but no substantial wage gains in high-poverty counties, then minimum wage increases would not effectively reduce poverty.
    In this case, policymakers should explore alternative policies for reducing poverty instead of relying heavily on the minimum wage policy.
}

\phantomsection\label{page:AE-3-1}\Copy{AE-3-1}{
    In this paper, we develop identification, estimation, and uniform inference methods to examine the treatment effect heterogeneity with respect to covariate values and other key variables (i.e., groups, periods, and treatment exposure time)  in the staggered DiD setting.
}
We build on the setup of \cite{callaway2021difference} and consider two types of target parameters: (i) the group-time conditional average treatment (CATT) function given a continuous pre-treatment covariate of interest and (ii) a variety of summary parameters that aggregate CATTs with certain estimable weights.
We begin by showing that, under essentially the same identification conditions as in \cite{callaway2021difference}, CATT is identifiable from a conditional version of the doubly robust (DR) estimand in \cite{callaway2021difference}.
Then, we propose three-step procedures for estimating CATT and the summary parameters:
the first stage is the same as the parametric estimation procedures for the outcome regression (OR) function and the generalized propensity score (GPS) in \cite{callaway2021difference};
the second and third stages comprise nonparametric local polynomial regressions (LPR) for estimating certain nuisance parameters and the conditional DR estimand. 
Lastly, to construct uniform confidence bands for the target parameters, we develop two uniform inference methods based on an analytical distributional approximation result and weighted/multiplier bootstrapping.

We investigate two statistical properties of our methods under the asymptotic framework where the number of cross-sectional units is large and the length of the time series is small and fixed.
First, we derive asymptotic linear representations and asymptotic mean squared errors (MSEs) of our estimators, which are used for constructing standard errors and for choosing appropriate bandwidths.
This part of the asymptotic investigations builds on the theory of the LPR estimation (\citealp{fan1996local}).
Second, we prove uniformly valid distributional approximation results for studentized statistics and their bootstrap counterparts, which play an essential role in constructing asymptotically valid critical values for the uniform confidence bands.
This result extends \cite*{lee2017doubly} and \cite*{fan2022estimation}, who study the uniform confidence bands for the conditional average treatment effect (CATE) function in the so-called unconfoundedness setup, to the staggered DiD setting.
More precisely, similar to these prior studies, 
\phantomsection\label{page:AE-17-2}\Copy{AE-17-2}{
    we use approximation theorems for suprema of empirical processes and asymptotic theory for potentially non-Donsker empirical processes by \cite*{chernozhukov2014anti,chernozhukov2014gaussian} to prove a uniformly valid analytical distributional approximation result and the uniform validity of weighted/multiplier bootstrap inference.
}

Our key assumptions are, in line with \cite{callaway2021difference}, the staggered treatment adoption and the conditional parallel trends assumption for panel data, in contrast to the unconfoundedness assumption for cross-sectional data.
\phantomsection\label{page:AE-3-2}\Copy{AE-3-2}{
    Unlike the unconfoundedness approach proposed by \cite{lee2017doubly} and \cite{fan2022estimation}, our methods allow researchers to learn about the heterogeneity of treatment effects with respect to key variables specific to the staggered DiD, such as groups, calendar time, and elapsed treatment time, as well as covariate values.
}
This attractiveness of our proposal is achieved with identification, estimation, and uniform inference methods tailored to the staggered DiD, whose statistical properties are not trivial from the existing results.
\phantomsection\label{page:AE-1a-1}\Copy{AE-1a-1}{
    In particular, since our CATT is a causal parameter that captures the conditional average treatment effect on the treated, we need to estimate nonparametric nuisance parameters in the second-stage estimation, and the construction of our uniform confidence bands requires careful consideration of its impact on conditional DR estimates.
}
\phantomsection\label{page:AE-3-3}\Copy{AE-3-3}{
    In addition, we highlight the importance of selecting appropriate critical values and bandwidths to ensure the uniform validity not only over covariate values but also over other key variables.
}

\paragraph{Related Literature.}

This paper focuses on CATT and the conditional aggregated parameter in the staggered DiD setup as the target parameters, building directly on the focus on ATT and the unconditional aggregated parameter by \citet{callaway2021difference}.
In doing so, we develop uniformly valid inference for treatment effect heterogeneity with respect to covariate values and other key variables while taking advantage of \citet{callaway2021difference}, that is, multiple treatment timing, treatment effect heterogeneity across units and time, the useful aggregation, and the attractive DR property.
While all of these are empirically desirable, the DR property should be particularly important when performing uniform inference for treatment effect heterogeneity with respect to covariate values, as highlighted by \citet{lee2017doubly} and \citet{fan2022estimation} in the unconfoundedness setup.
This is the main reason why we build directly on \citet{callaway2021difference} rather than the other useful DiD methods (e.g., \citealp{sun2021estimating,wooldridge2021two,borusyak2024revisiting}).

In terms of developing uniform inference for treatment effect heterogeneity with respect to covariate values, we build on directly \citet{lee2017doubly} and \citet{fan2022estimation} in the unconfoundedness setup.
Compared to their cross-sectional data analyses, our panel data analysis allows for understanding both the static and dynamic nature of treatment effect heterogeneity, which should be important from an empirical point of view.
From a theoretical perspective, this advantage of our proposal stems from our DiD approach with (i) (possibly long) time differences of outcomes rather than their levels and (ii) novel DR estimators constructed with parametric and nonparametric nuisance function estimators and weighting schemes different from theirs.
In particular, estimating the nonparametric nuisance functions in our second-stage estimation has non-negligible (first-order) effects on the asymptotic properties of our DR estimator, as shown in Theorems \ref{thm:bias_variance} and \ref{thm:mb_linearize}, which are new insights in the literature.
\phantomsection\label{page:AE-1a-2}\Copy{AE-1a-2}{
    The issue of estimating nonparametric nuisance functions arises in our case because our CATT is a type of the conditional average treatment effect on the treated, as opposed to the focus on CATE in \citet{lee2017doubly} and \citet{fan2022estimation}.
}
After dealing with these considerations by building on the theory of the LPR estimation, our uniformly valid approximation results in Theorems \ref{thm:approx_us_distribution} and \ref{thm:mb_valid} are obtained as applications of empirical process techniques in the same manner as in these two previous studies.

This paper clearly builds on previous work in the DiD literature that developed methods to understand treatment effect heterogeneity arising from covariates.
For example, \cite{abadie2005semiparametric} proposed pointwise inference for the conditional average treatment effect on the treated given a covariate based on inverse probability weighting (IPW) and series approximations in the canonical two-periods and two-groups DiD setting.
His proposal is even applicable to the staggered adoption case by focusing on a subset of the original dataset consisting only of a treated group with a specific treatment timing and a comparison group.
Compared to his proposal, our methods have novelties in terms of the empirically desirable DR property, the kernel smoothing technique that facilitates tuning parameter selection, and the uniform validity over covariate values and other key variables proven by empirical process theory.

\paragraph{Paper Organization.}
The rest of the paper is organized as follows. 
Section \ref{sec:setup} introduces the setup and provides a non-technical roadmap for implementing our methods. 
Section \ref{sec:application} illustrates our methods in the context of the minimum wage.
Sections \ref{sec:CATTgt} and \ref{sec:Summary} discuss the identification, estimation, and uniform inference methods for CATT and the summary parameters, respectively.
The supplementary results are presented in the online appendix.
The accompanying R package {\ttfamily \href{https://tkhdyanagi.github.io/didhetero/}{didhetero}} is available from the authors' websites.

\section{Setup and Roadmap} \label{sec:setup}

Whenever possible, we use the same notation as in \cite{callaway2021difference}.
For each unit $i \in \{ 1, \dots, n \}$ and time period $t \in \{ 1, \dots, \mathcal{T} \}$, we observe a binary treatment $D_{i,t} \in \{ 0, 1 \}$, an outcome variable $Y_{i,t} \in \mathcal{Y} \subseteq \mathbb{R}$, and a vector of the pre-treatment covariates $X_i \in \mathcal{X} \subseteq \mathbb{R}^k$.
For notational simplicity, we often suppress the subscript $i$.

We consider the staggered adoption design, which includes the canonical two-periods and two-groups setting as a special case, under the random sampling scheme for balanced panel data.

\begin{assumption}[Staggered Treatment Adoption] \label{as:staggered}
	$D_1 = 0$ almost surely (a.s.).
	For any $t = 3, \dots, \mathcal{T}$, $D_{t-1} = 1$ implies that $D_t = 1$ a.s.
\end{assumption}

\begin{assumption}[Random Sampling] \label{as:iid}
	The panel data $\{ (Y_{i,t}, X_i, D_{i,t}): i = 1, \dots, n, t = 1, \dots, \mathcal{T} \}$ are independent and identically distributed (IID) across $i$.
\end{assumption}

Denote the time period when the unit becomes treated for the first time as $G \coloneqq \min\{ t : D_t = 1 \}$.
We set $G = \infty$ if the unit has never been treated.
We often refer to $G$ as the ``group'' to which the unit belongs.
In particular, we call the set of units with $G = g$ for $g \in \{ 2, \dots, \mathcal{T} \}$ as the ``not-yet-treated'' group in pre-treatment periods $t < g$ and that with $G = \infty$ as the ``never-treated'' group.
Assuming that $\bar g \coloneqq \max_{1 \le i \le n} G_i$ is known a priori, we write the set of realized treatment timings before $\bar g$ as $\mathcal{G} \coloneqq \mathrm{supp}(G) \setminus \{ \bar g \}$.
With an abuse of notation, we let $\bar g - 1 = \mathcal{T}$ if $\bar g = \infty$.

Under Assumption \ref{as:staggered}, the potential outcome given $G$ is well-defined.
Specifically, we write $Y_t(g)$ as the potential outcome in period $t$ given that the unit becomes treated at period $g \in \{ 2, \dots, \mathcal{T} \}$.
Meanwhile, we denote $Y_t(0)$ as the potential outcome in period $t$ when the unit belongs to the never-treated group (i.e., when $G = \infty$).
By construction, $Y_t = Y_t(0) + \sum_{g = 2}^{\mathcal{T}} [ Y_t(g) - Y_t(0) ] \cdot G_g$, where $G_g \coloneqq \bm{1}\{ G = g \}$.
Note that $Y_t(g) - Y_t(0)$ is the effect of receiving the treatment for the first time in period $g$ on the outcome in period $t$.

We aim to examine the extent to which the average treatment effect varies with groups, periods, and a single continuous covariate.
To be specific, suppose that $X$ can be decomposed into $X = (Z, X_{\sub}^\top)^\top$ with a scalar continuous covariate $Z$ and the other elements $X_{\sub}$.
The presence of $X_{\sub}$ should be important in typical DiD applications where the parallel trends assumption is more likely to hold only after conditioning on a number of covariates.
For some pre-specified real numbers $a$ and $b$ such that $a < b$, let $\mathcal{I} = [a, b]$ denote a proper closed subset of the support of $Z$.
As the first target parameter, we consider the group-time conditional average treatment effect (CATT) given $Z = z$ for $z \in \mathcal{I}$:
\begin{align} \label{eq:CATT}
    {\CATT}_{g,t}(z)
    \coloneqq \bE [ Y_t(g) - Y_t(0) \mid G_g = 1, Z = z ].
\end{align}
Estimating ${\CATT}_{g,t}(z)$ over $(g,t,z)$ is helpful in understanding the treatment effect heterogeneity with respect to group $g$, calendar time $t$, and covariate value $z$.

In Section \ref{sec:CATTgt}, we develop identification, estimation, and uniform inference methods for CATT.
We begin by introducing a conditional DR estimand $\DR_{g,t}(z)$, which is a conditional counterpart of the DR estimand in \citet{callaway2021difference}, by using the not-yet-treated group as the comparison group.
We then show that $\CATT_{g,t}(z)$ is identified by $\DR_{g,t}(z)$ for each $(g, t, z) \in \mathcal{A}$, where
\begin{align} \label{eq:mathcalA}
    \mathcal{A} 
    \coloneqq \{ (g, t, z) : g \in \mathcal{G}, t \in \{ 2, \dots, \mathcal{T} \}, g \le t < \bar g, z \in \mathcal{I} \}.
\end{align}
Given the identification result, we propose to construct a $(1 - \alpha)$ uniform confidence band for $\CATT_{g,t}(z)$ over $(g,t,z) \in \mathcal{A}$ by a family of intervals, denoted as $\mathcal{C} \coloneqq \{ \mathcal{C}_{g,t}(z): (g,t,z) \in \mathcal{A} \}$ with
\begin{align} \label{eq:UCB_CATT}
    \mathcal{C}_{g,t}(z)
    \coloneqq 
    \left[ 
    \hat{\DR}_{g,t}(z) - c(1 - \alpha) \cdot \hat{\SE}_{g,t}(z),
    \qquad 
    \hat{\DR}_{g,t}(z) + c(1 - \alpha) \cdot \hat{\SE}_{g,t}(z)
    \right],
\end{align}
where $\hat{\DR}_{g,t}(z)$ is a three-step estimator computed with certain parametric estimation procedures and nonparametric LPR estimation, $\hat{\SE}_{g,t}(z)$ is a pointwise standard error, and $c(1 - \alpha)$ is a uniform critical value obtained from an analytical method or weighted bootstrapping.
\phantomsection\label{page:AE-3-4}\Copy{AE-3-4}{
    Importantly, to ensure that the uniform confidence band $\mathcal{C}$ is uniformly valid over $(g, t, z) \in \mathcal{A}$, the critical value $c(1 - \alpha)$ must not depend on $(g, t, z)$ and is larger than the standard Wald-type pointwise critical value (i.e., the $(1 - \alpha/2)$ quantile of the standard normal distribution).
    As will be discussed later, the bandwidth used for the LPR estimation is crucial for constructing $c(1 - \alpha)$, and we recommend using a bandwidth that does not depend on $(g, t, z)$ for our uniform inference.
    See Remark \ref{remark:uniform}.
}

We can also consider a variety of useful summary parameters by aggregating $\CATT_{g, t}(z)$'s.
Specifically, building on the aggregation scheme of \cite{callaway2021difference}, we set the second target parameter to the aggregated parameter of the following form:
\begin{align} \label{eq:summary}
    \theta(z) 
    \coloneqq \sum_{g \in \mathcal{G}} \sum_{t=2}^{\mathcal{T}} w_{g,t}(z) \cdot {\CATT}_{g,t}(z),
\end{align}
where $w_{g,t}(z)$ is a known or estimable weighting function that determines the causal interpretation of $\theta(z)$.
For example, letting $e = t - g \ge 0$ denote elapsed treatment time, we can consider the ``event-study-type'' conditional average treatment effect:
\begin{align*}
    \theta_{\es}(e, z)
    & \coloneqq \bE[ Y_{i,G+e}(G) - Y_{i,G+e}(0) \mid G + e < \bar g, Z = z ] \\
    & = \sum_{g \in \mathcal{G}} \bm{1}\{ g + e < \bar g \} \cdot \Pr(G = g \mid G + e < \bar g, Z = z) \cdot {\CATT}_{g,g+e}(z).
\end{align*}
This is the conditional counterpart of the event-study-type summary parameter in \cite{callaway2021difference} and useful for understanding the treatment effect heterogeneity with respect to treatment exposure time $e$ and covariate value $z$.
Another useful example is the simple weighted conditional average treatment effect, which aggregates $\CATT_{g,t}(z)$'s into an overall effect as follows:
\begin{align*}
    \theta_{\W}^{\OO}(z)
    & \coloneqq \frac{1}{\kappa(z)} \sum_{t=2}^{\mathcal{T}} \bE[ Y_t(G) - Y_t(0) \mid G < \bar g, Z = z ] \\
    & = \frac{1}{\kappa(z)} \sum_{g \in \mathcal{G}} \sum_{t=2}^{\mathcal{T}} \bm{1}\{ g \le t < \bar g \} \cdot \Pr(G = g \mid G < \bar g, Z = z) \cdot {\CATT}_{g,t}(z),
\end{align*}
where $\kappa(z) \coloneqq \sum_{g \in \mathcal{G}} \sum_{t=2}^{\mathcal{T}} \bm{1}\{ g \le t < \bar g \} \cdot \Pr(G = g \mid G < \bar g, Z = z)$.
We can also consider other useful summary parameters by appropriately choosing different weights.
See Appendix \ref{sec:supp-summary}.

In Section \ref{sec:Summary} and Appendix \ref{sec:supp-summary}, we study how to perform the uniform inference for these summary parameters.
The proposed uniform confidence band for the aggregated parameter $\theta(z)$ has the same form as the uniform confidence band for CATT, namely $\mathcal{C}_{\theta} \coloneqq \{ \mathcal{C}_{\theta}(z) \}$ with
\begin{align} \label{eq:UCB_summary}
    \mathcal{C}_{\theta}(z)
    \coloneqq \left[ \hat{\theta}(z) - c_{\theta}(1 - \alpha) \cdot \hat{\SE}_{\theta}(z), \quad \hat{\theta}(z) + c_{\theta}(1 - \alpha) \cdot \hat{\SE}_{\theta}(z) \right],
\end{align}
where $\hat{\theta}(z)$ is an estimator obtained as an empirical analogue of \eqref{eq:summary}, $\hat{\SE}_{\theta}(z)$ is a pointwise standard error, and $c_{\theta}(1 - \alpha)$ is a uniform critical value via an analytical method or multiplier bootstrapping.
\phantomsection\label{page:AE-3-5}\Copy{AE-3-5}{
    Similar to the case of CATT, the uniform critical value $c_{\theta}(1 - \alpha)$ and the bandwidth should not depend on $z$ and the variable specific to the summary parameter of interest (e.g., treatment exposure time $e$).
}

\phantomsection\label{page:AE-16-1}\Copy{AE-16-1}{
    To justify the uniform confidence bands \eqref{eq:UCB_CATT} and \eqref{eq:UCB_summary}, we need to make bias arising from kernel smoothing asymptotically negligible.
    To this end, we propose an undersmoothing approach based on the insight of the simple robust bias-corrected (RBC) inference.
    Specifically, we consider estimating the target parameters by local quadratic regressions (LQR) based on integrated mean squared error (IMSE) optimal bandwidths for local linear regressions (LLR).
    See Section \ref{subsec:bandwidth}.
}

\phantomsection\label{page:AE-7-1}\Copy{AE-7-1}{
    In the same spirit of the focus on ``pre-trends'' in previous studies, it would be beneficial to assess the credibility of the identifying assumptions using our uniform inference.
    To this end, we focus on the following testable implications in the pre-treatment periods: $\CATT_{g,t}(z) = \theta_{\es}(e, z) = 0$ for all $g \in \mathcal{G}$, $t \ge 2$ such that $t \le g - 2$, $e \le -2$, and $z \in \mathcal{I}$.
    Note that we exclude $t = g - 1$ and $e = -1$ as base periods.
    If we find estimation and uniform inference results that are inconsistent with these testable implications, it suggests violations of the identifying assumptions.
    We discuss this type of simple diagnosis based on pre-trends in Appendix \ref{sec:pretrends}.
}

Throughout the paper, we focus on the treatment effect heterogeneity with respect to a single continuous covariate $Z$, rather than the full covariate $X$.
This is because focusing on a single continuous covariate of interest allows us to easily visualize and interpret the heterogeneity in instantaneous and dynamic treatment effects with respect to its values, as illustrated in Figure \ref{fig:boot} in the next section.

\section{Empirical Illustration} \label{sec:application}

In this section, we use our proposal to assess the heterogeneity in the effects of the minimum wage change on youth employment.
In doing so, we illustrate the empirical relevance of our methods using a real dataset before proceeding to the technical discussions in the following sections.

We use the same dataset as in \cite{callaway2021difference}, which includes county level minimum wages, county level teen employment, and other county characteristics for 2,284 U.S. counties in 2001-2007. 
The outcome variable $Y_{i,t}$ is the logarithm of teen employment in county $i$ at year $t$.
We define the group $G_i$ by considering 100, 223, and 584 counties that increased their minimum wages in 2004, 2006, and 2007, respectively, as the treated units.
This implies that the remaining 2,184, 1,961, and 1,377 counties in 2004, 2006, and 2007, respectively, are the not-yet-treated units in each year.
In line with the specification in the empirical analysis of \cite{callaway2021difference}, the pre-treatment covariates $X_i$ consist of county characteristics before 2000, including the poverty rate (i.e., the share of the population below the poverty line), the share of the white population, the share of the population of high school graduates, the regional dummy, the median income, the total population, and the squares of the median income and the total population.
To save space, we relegate more information about the dataset, summary statistics, and pre-trends to Appendix \ref{sec:empirical2}.

\phantomsection\label{page:AE-2-2}\Copy{AE-2-2}{
    Among the covariates, we focus on examining the treatment effect heterogeneity with respect to the poverty rate.
    Reducing poverty should be one of the main purposes of the minimum wage policy, but it is unclear a priori whether and how minimum wage increases reduce poverty.
    This is because the extent to which minimum wage increases reduce poverty depends on the structural relationship between wage gains and job losses at the bottom of the income distribution, as well as other factors that determine income, as discussed in \citet{dube2019minimum}.
    From this perspective, understanding how the impact of minimum wage increases on teen employment depends on the poverty rate should be useful for assessing whether the minimum wage policy alleviates poverty.
    For example, if we find that minimum wage increases significantly decrease teen employment in low-poverty counties, but have no significant effect in high-poverty counties, policymakers should emphasize the importance of the minimum wage policy for poverty reduction at least in high-poverty counties.
}

\phantomsection\label{page:AE-6-1}\Copy{AE-6-1}{
    In Figure \ref{fig:boot}, panel (a) shows the estimation and uniform inference results for ${\CATT}_{g,t}(z)$ for $(g, t) \in \{2004, 2006, 2007\}^2$.
    We restrict our focus to this set of $(g, t)$ for presentation purposes.
    Panels (b) and (c) depict the event-study-type conditional average treatment effect $\theta_{\es}(e, z)$ for $e \in \{ 0, 1, 2, 3 \}$ and the simple weighted conditional average treatment effect $\theta_{\W}^{\OO}(z)$, respectively.
    Panels (b) and (c) use data from all available groups and post-treatment periods, not just data from $(g, t) \in \{2004, 2006, 2007\}^2$.
}
In each panel, the horizontal axis corresponds to the interval $\mathcal{I}$ set as the interquartile range of the poverty rate, the solid line indicates the LQR estimates based on the IMSE-optimal bandwidth for the LLR estimation, and the gray area corresponds to the 95\% uniform confidence band via weighted/multiplier bootstrapping using \citeauthor{mammen1993bootstrap}'s (\citeyear{mammen1993bootstrap}) weights.
\phantomsection\label{page:AE-3-6}\Copy{AE-3-6}{
    Because the bandwidth used for the LQR estimation should not depend on the variables of interest (e.g., $(g, t, z)$ for $\CATT_{g,t}(z)$) for our uniform inference, we take the minimum of the integrated (over $z \in \mathcal{I}$) MSE-optimal bandwidths across the variables (e.g., groups $g$ and post-treatment periods $t$).
}
\phantomsection\label{page:AE-9-1}\Copy{AE-9-1}{
    We also found that the LLR-based inference methods (both analytical and bootstrap) and the LQR-based analytical method lead to almost the same empirical results as those presented here, but we suppress them to save space.
}

The main empirical findings can be summarized as follows.
First, the estimated CATT functions are nearly flat around zero for the 2004, 2006, and 2007 groups in 2004, 2006, and 2007, respectively, and the corresponding uniform confidence bands are not as wide.
This means that, with satisfactory precision in terms of uniform inference, minimum wage increases have on average almost no instantaneous effect on teen employment.
Second, we find negative but small CATT estimates with a small amount of treatment effect heterogeneity for the 2004 and 2006 groups in 2006-2007 and 2007, respectively, and the corresponding uniform confidence bands are wider than those for the instantaneous effects.
This result may suggest that there are small but not substantial dynamic effects of minimum wage increases on teen employment, but this may be due to the lack of precision of uniform inference for dynamic effects.
Third, the uniform inference results for the summary parameters also imply that there are no substantial effects of minimum wage increases on teen employment with modest treatment effect heterogeneity.
\phantomsection\label{page:AE-2-3}\Copy{AE-2-3}{
    In particular, the result for the simple weighted conditional average treatment effect indicates that there is almost no effect, especially at high poverty rates.
    Overall, our empirical results suggest that minimum wage increases do not substantially decrease teen employment and thus may be effective in reducing poverty, particularly in high-poverty counties.
}

\begin{figure}[ht]
    \centering
    \begin{minipage}[t]{\textwidth}
        \centering
        \includegraphics[width=\linewidth, bb=0 0 2160 1080]{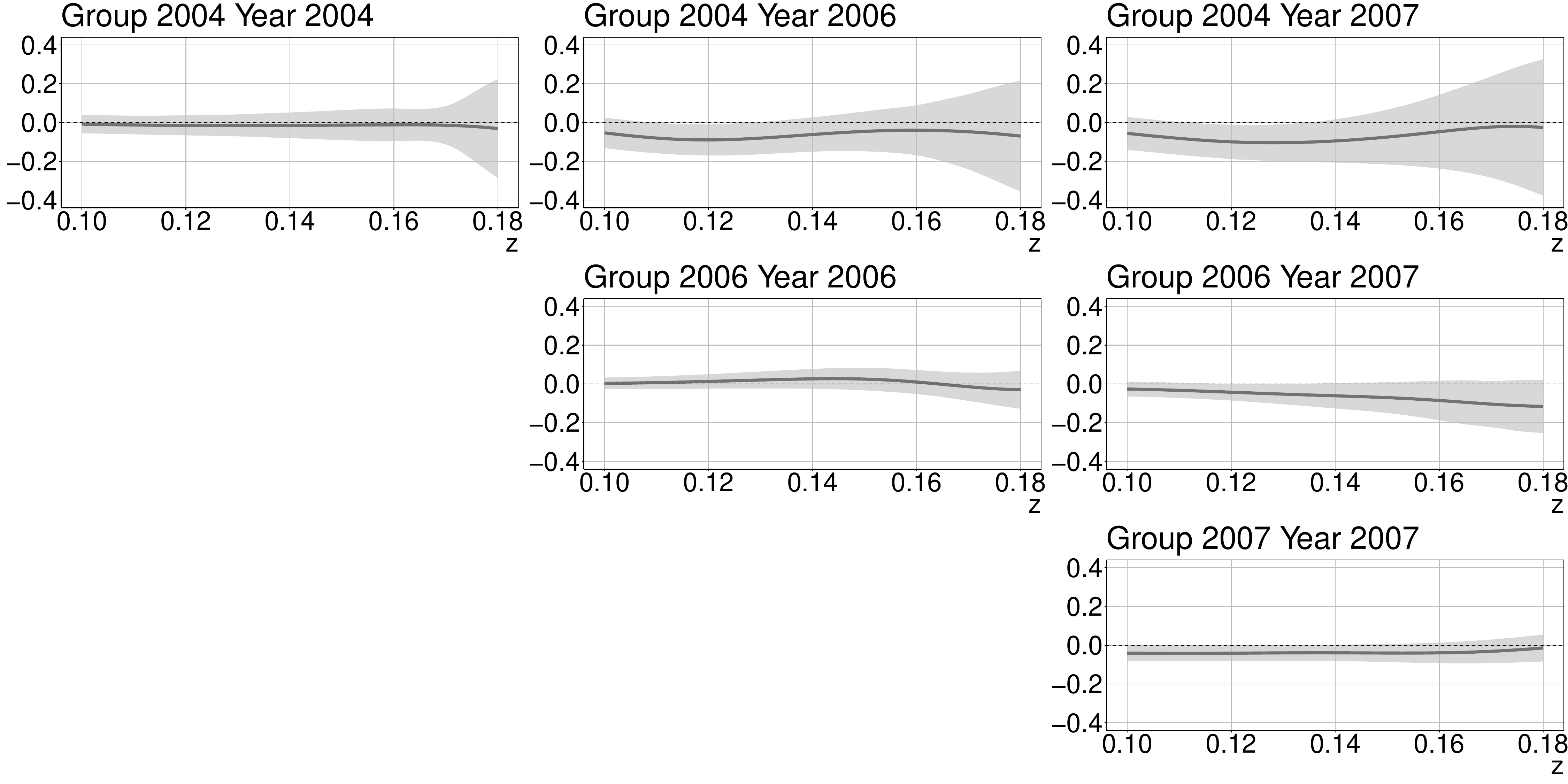}
        \subcaption{CATT}
    \end{minipage}
    
    \vspace{-5em} 
    
    \begin{minipage}[t]{\textwidth}
        \centering
        \includegraphics[width=\linewidth, bb=0 0 2160 1080]{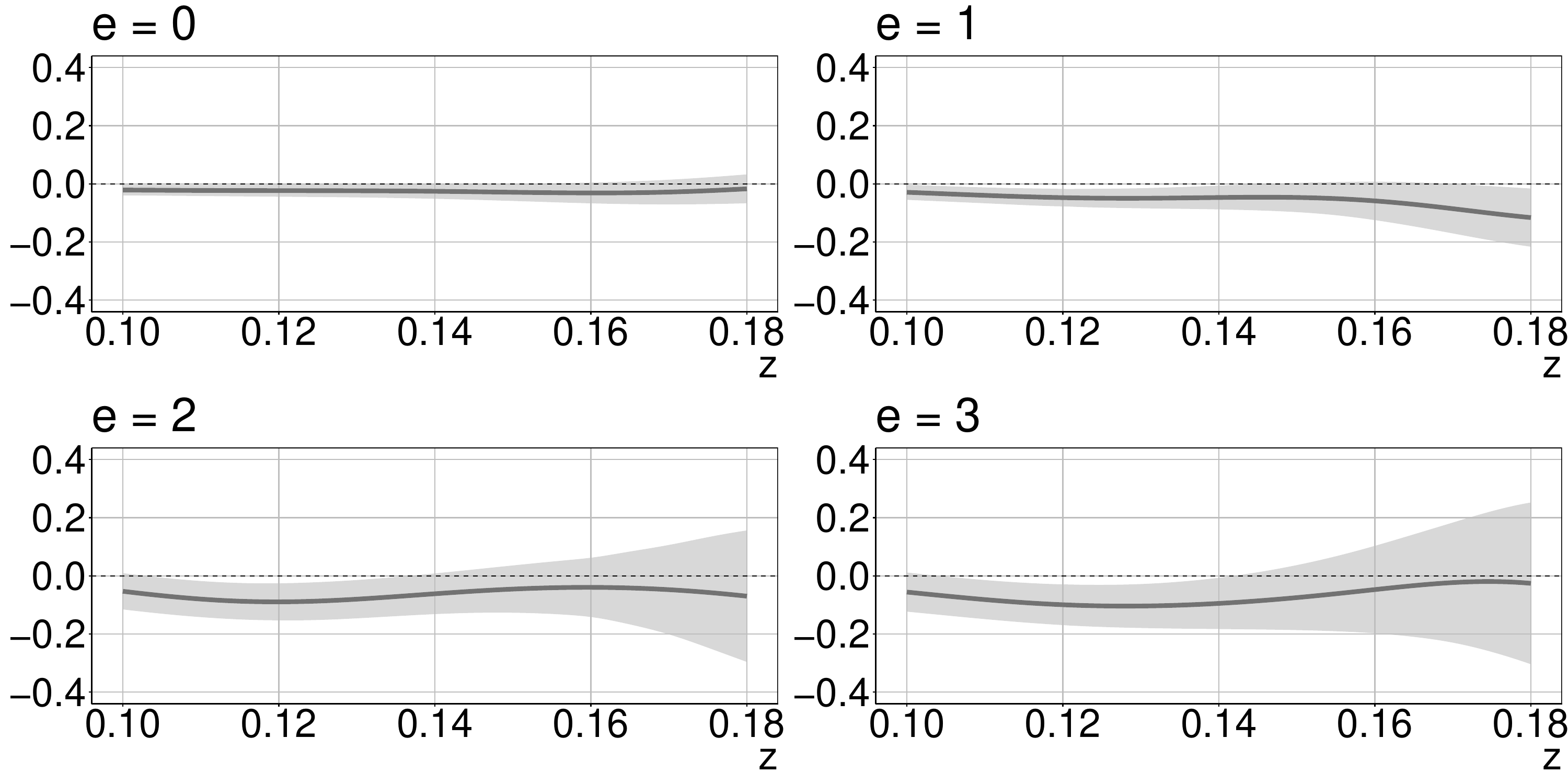}
        \subcaption{The event-study-type conditional average treatment effect}
    \end{minipage}
    
    \vspace{-12em} 
    
    \begin{minipage}[t]{\textwidth}
        \centering
        \includegraphics[width=\linewidth, bb=0 0 2160 1080]{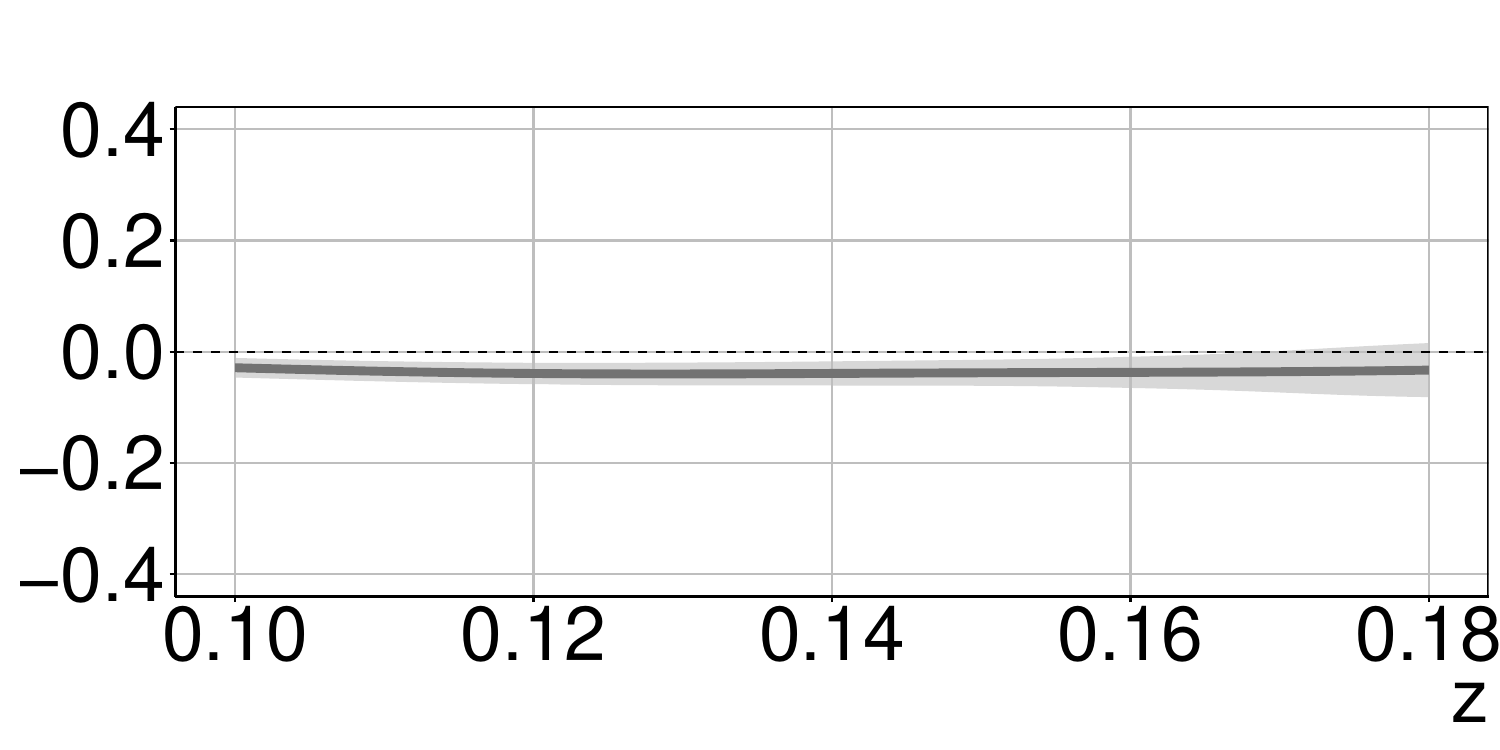}
        \subcaption{The simple weighted conditional average treatment effect}
    \end{minipage}
    \caption{The LQR estimates and 95\% uniform confidence bands constructed with weighted/multiplier bootstrapping in the post-treatment periods.}
    \label{fig:boot}
\end{figure}

\section{Inference for CATT} \label{sec:CATTgt}

In this section, we develop identification, estimation, and uniform inference methods for CATT defined in \eqref{eq:CATT}.
For this purpose, there are two options for the comparison group: the not-yet-treated group and the never-treated group.
For presentation purposes, the main body of the paper presents only the analysis using the not-yet-treated group.
The analysis using the never-treated group is relegated to Appendix \ref{sec:never-treated}.
Moreover, we can consider three types of estimands: OR, IPW, and DR estimands.
\phantomsection\label{page:AE-10}\Copy{AE-10}{
Throughout the paper, we focus on the DR estimand as it is more robust against model misspecifications.}

The following quantities play important roles in the analysis using the not-yet-treated group and the DR estimand.
For each $g$ and $t$, we define the generalized propensity score (GPS) and the OR function respectively by
\begin{align} \label{eq:GPS-OR}
\begin{split}
    p_{g,t}(X) 
    & \coloneqq \bP[G_g = 1 \mid X, G_g + (1 - D_t)(1 - G_g) = 1], \\
    m_{g,t}(X) 
	& \coloneqq \bE[ Y_t - Y_{g-1} \mid X, D_{t}=0, G_g=0 ]. 
\end{split}
\end{align}
Let $p_{g,t}(X; \pi_{g,t})$ and $m_{g,t}(X; \beta_{g,t})$ be parametric specifications for these quantities, where each function is known up to the corresponding finite-dimensional parameter.
Denote the corresponding parameter spaces as $\Pi_{g,t}$ and $\mathscr{B}_{g,t}$.

\subsection{Identification} \label{sec:identificaton}

We impose the following identification conditions, which are essentially the same as Assumptions 3, 4, 6, and 7(iii) in \cite{callaway2021difference}.

\begin{assumption}[No Treatment Anticipation] \label{as:anticipation-ny}
	It holds that 
	\begin{align*}
		\bE [ Y_t(g) \mid X, G_g = 1 ]
		= \bE [ Y_t(0) \mid X, G_g = 1 ] 
		\quad
		\text{a.s. for all $g \in \mathcal{G}$ and $t \in \{ 1, \dots, \mathcal{T} \}$ such that $t < g$}. 
	\end{align*}
\end{assumption}

\begin{assumption}[Conditional Parallel Trends Based on the ``Not-Yet-Treated'' Group] \label{as:parallel-ny}
	For each $g \in \mathcal{G}$ and each $(s, t) \in \{ 2, \dots, \mathcal{T} \} \times \{ 2, \dots, \mathcal{T} \}$ such that $t \ge g$ and $t  \le s < \bar g$,
	\begin{align*}
		\bE[ Y_t(0) - Y_{t-1}(0) \mid X, G_g = 1 ]
		= \bE[ Y_t(0) - Y_{t-1}(0) \mid X, D_s = 0, G_g = 0 ]
		\quad 
		\text{a.s.}
	\end{align*}
\end{assumption}

\begin{assumption}[Overlap] \label{as:overlap}
	For each $g \in \mathcal{G}$ and $t \in \{ 2, \dots, \mathcal{T} \}$, there exists $\varepsilon > 0$ such that $\bP( G_g = 1 \mid Z ) > \varepsilon$ and $p_{g,t}(X) < 1 - \varepsilon$ a.s.
\end{assumption}

\begin{assumption}[Parametric Models for the ``Not-Yet-Treated'' Group] \label{as:parametric-ny}
	For each $g \in \mathcal{G}$ and $t \in \{ 2, \dots, \mathcal{T}\}$ such that $t \ge g$, either condition is satisfied:
	\begin{enumerate}[(i)]
		\item There exists a unique $\beta_{g,t}^{*} \in \mathscr{B}_{g,t}$ such that $m_{g,t}(X) = m_{g,t}(X; \beta_{g,t}^{*})$ a.s.
		\item There exists a unique $\pi_{g,t}^* \in \Pi_{g,t}$ such that $p_{g,t}(X) = p_{g,t}(X; \pi_{g,t}^*)$ a.s.
	\end{enumerate}
\end{assumption}

The conditional DR estimand based on the not-yet-treated group is defined by
\begin{align} \label{eq:DRestimand}
\begin{split} 
    & {\DR}_{g,t}(Z; \beta_{g,t}, \pi_{g,t}) \\
    & \coloneqq \bE \left[ \left( \frac{G_g}{\bE[ G_g \mid Z]} - \frac{ R_{g,t}(W; \pi_{g,t})  }{\bE \left[ R_{g,t}(W; \pi_{g,t}) \; \middle| \; Z \right]} \right) \left( Y_t - Y_{g-1} - m_{g,t}(X; \beta_{g,t}) \right) \; \middle| \; Z \right],
\end{split}
\end{align}
where $W \coloneqq (Y_1, \dots, Y_\mathcal{T}, X^\top, D_1, \dots, D_\mathcal{T})^\top$ and
\begin{align} \label{eq:Rdef}
    R_{g,t}(W; \pi_{g,t}) 
    \coloneqq \frac{p_{g,t}(X; \pi_{g,t}) (1-D_t)(1-G_g)}{1 - p_{g,t}(X; \pi_{g,t})}.
\end{align}

As a building block of our uniform inference methods, the next lemma shows that our DR estimand identifies CATT if at least one of the GPS and OR function is specified correctly.
This result follows from almost the same arguments as in Theorem 1 of \cite{sant2020doubly} and Theorem 1 of \cite{callaway2021difference}.
In fact, the only difference is that our estimand is conditioned on a single covariate $Z$, while their estimands integrate over all covariates.

\begin{lemma} \label{lem:DR}
	Suppose that Assumptions \ref{as:staggered}--\ref{as:parametric-ny} hold.
	Fix arbitrary $(g, t, z) \in \mathcal{A}$, where $\mathcal{A}$ is defined in \eqref{eq:mathcalA}.
	\begin{enumerate}[(i)]
		\item Under Assumption \ref{as:parametric-ny}(i), ${\CATT}_{g,t}(z) = {\DR}_{g,t}(z; \beta_{g,t}^{*}, \pi_{g,t})$ for all $\pi_{g,t} \in \Pi_{g,t}$.
		\item Under Assumption \ref{as:parametric-ny}(ii), ${\CATT}_{g,t}(z) = {\DR}_{g,t}(z; \beta_{g,t}, \pi_{g,t}^*)$ for all $\beta_{g,t} \in \mathscr{B}_{g,t}$.
	\end{enumerate}
\end{lemma}

\begin{remark} \label{remark:anticipation}
    \cite{callaway2021difference} consider a more general anticipation assumption than Assumption \ref{as:anticipation-ny}, called limited treatment anticipation, which allows units to anticipate treatment by a known amount of time.
    While our analysis can be extended in this direction, for exposition purposes, we impose Assumption \ref{as:anticipation-ny} throughout the main text.
    See Appendix \ref{sec:not-yet-treated} for further discussion.
\end{remark}

\begin{remark}
\phantomsection\label{page:AE-11}\Copy{AE-11}{
    To be consistent with \cite{callaway2021difference}, we focus on the parametric approach to the GPS and OR function by assuming that there is a small number of covariates.
    This parametric first-stage estimation greatly simplifies our theory because we can treat the GPS and OR function as if we know them, given that the convergence rates of their parametric estimators are faster than those of the nonparametric estimators.
    However, if we would like to avoid parametric specifications or if there are many covariates, it might be desirable to rely on nonparametric and/or (double/de-biased) machine learning methods even in the first stage.
    Such a strategy would be possible but challenging to quantify the estimation effects arising from nonparametric and/or machine learning methods.
}
\end{remark}

\subsection{Estimation and Uniform Inference} \label{sec:estimation}

We develop estimation and uniform inference methods for CATT based on the identification result in Lemma \ref{lem:DR}.
With an abuse of notation, we write $m_{g,t} \coloneqq m_{g,t}(X; \beta_{g,t}^{*})$ and $R_{g,t} \coloneqq R_{g,t}(W; \pi_{g,t}^*)$.
Let
\begin{align} \label{eq:Adef}
    A_{g,t} \coloneqq \left( \frac{G_g}{\bE[ G_g \mid Z = z ]} - \frac{ R_{g,t} }{\bE \left[ R_{g,t} \mid Z = z \right]} \right) \left( Y_t - Y_{g-1} - m_{g,t} \right).
\end{align}
Note that $A_{g,t}$ depends on the covariate value $z$, but we suppress its dependence to simplify the exposition.
The goal is to construct the uniform confidence band for $\CATT_{g,t}(z)$, identified by ${\DR}_{g,t}(z) = \bE[ A_{g,t} \mid Z = z ]$.

For the subsequent discussion, it is convenient to introduce the following notation related to the LQR estimation.
For a generic variable $Q$ and a generic integer $\nu \ge 0$, let $\mu_Q^{(\nu)}(z) \coloneqq \bE^{(\nu)}[Q \mid Z = z]$ denote the $\nu$-th derivative with respect to $z$ of the conditional mean of $Q$ given $Z = z$.
As usual, we write $\mu_Q(z) = \mu_Q^{(0)}(z) = \bE[Q \mid Z = z]$.
The LQR estimator of $\mu_Q^{(\nu)}(z)$ for $\nu \in \{ 0, 1, 2 \}$ is defined by
\begin{align*}
	\hat{\mu}_Q^{(\nu)}(z) \coloneqq \nu! \bm{e}_\nu^\top \hat{\bm{\beta}}_Q (z)
	\quad \text{with} \quad
	\hat{\bm{\beta}}_Q (z) \coloneqq \argmin_{\bm{b} \in \mathbb{R}^{3}} \sum_{i=1}^n \left( Q_i - \bm{r}_2(Z_i - z)^\top \bm{b} \right)^2 K_Q \left( \frac{Z_i - z}{h_Q} \right),
\end{align*}
where $\bm{e}_\nu$ is the $3 \times 1$ vector in which the $(\nu + 1)$-th element is $1$ and the rest are $0$, $\bm{r}_2(u) \coloneqq (1, u, u^2)^\top$ is the $3 \times 1$ vector of second order polynomials, $K_Q$ is a kernel function, and $h_Q > 0$ is a bandwidth.

\subsubsection{Procedure} \label{subsec:procedure}

We explain how to obtain the conditional DR estimator $\hat{\DR}_{g,t}(z)$ that appeared in the proposed $(1 - \alpha)$ uniform confidence band for $\CATT_{g,t}(z)$ in \eqref{eq:UCB_CATT}.
Several more technical issues, including the formulas for the standard error and the critical value, are discussed in the following subsections.

We consider a three-step estimation procedure.
First, we estimate $\beta_{g,t}^{*}$ and $\pi_{g,t}^*$ via some parametric methods, such as the least squares method and the maximum likelihood estimation.
Using the resulting first-stage estimators $\hat \beta_{g,t}$ and $\hat \pi_{g,t}$, we compute $\hat R_{i,g,t} \coloneqq R_{g,t}(W_i; \hat \pi_{g,t})$ and $\hat m_{i,g,t} \coloneqq m_{g,t}(X_i; \hat \beta_{g,t})$ for each $i$.
Second, for each $i$, we compute 
\begin{align} \label{eq:A_igt}
    \hat A_{i,g,t}
    & \coloneqq \left( \frac{ G_{i,g} }{ \hat \mu_G(z) } - \frac{ \hat R_{i,g,t} }{ \hat \mu_{\hat R}(z) } \right) \left( Y_{i,t} - Y_{i,g-1} - \hat m_{i,g,t} \right),
\end{align}
where $\hat \mu_G(z)$ is the LQR estimator of $\mu_G(z) = \bE[ G_g \mid Z = z]$ using a bandwidth $h_G$ and a kernel function $K_G$, and $\hat \mu_{\hat R}(z)$ is the LQR estimator of $\mu_R(z) = \bE[ R_{g,t} \mid Z = z]$ using a bandwidth $h_R$ and a kernel function $K_R$.
Specifically, $\hat \mu_{\hat R}(z)$ is defined by
\begin{align*}
    \hat{\mu}_{\hat R}(z)
    \coloneqq \bm{e}_0^\top \hat{\bm{\beta}}_{\hat R}(z),
    \qquad
    \hat{\bm{\beta}}_{\hat R}(z) \coloneqq \argmin_{\bm{b} \in \mathbb{R}^{3}} \sum_{i=1}^n \left( \hat R_{i,g,t} - \bm{r}_{2}(Z_i - z)^\top \bm{b} \right)^2 K_R \left( \frac{Z_i - z}{h_R} \right),
\end{align*}
and the definition of $\hat \mu_G(z)$ is analogous.
Finally, we obtain the conditional DR estimator $\hat{\DR}_{g,t}(z)$ from the following LQR using a bandwidth $h_A$ and a kernel function $K_A$:
\begin{align} \label{eq:DRest-ny}
    \begin{split}
        \hat{\DR}_{g,t}(z) & \coloneqq \hat \mu_{\hat A}(z) \coloneqq \bm{e}_0^\top \hat{\bm{\beta}}_{\hat A}(z), \\ 
        \hat{\bm{\beta}}_{\hat A}(z) & \coloneqq \argmin_{\bm{b} \in \mathbb{R}^{3}} \sum_{i=1}^n \left( \hat A_{i,g,t} - \bm{r}_{2}(Z_i - z)^\top \bm{b} \right)^2 K_A \left( \frac{Z_i - z}{h_A} \right).
    \end{split}
\end{align}

We impose the next assumption on the bandwidths and the kernel functions.

\begin{assumption} \label{as:common}
    In the second- and third-stage estimation,
    \phantomsection\label{page:AE-3-7}\Copy{AE-3-7}{
        (i) $h_G = h_R = h_A = h$ for a common bandwidth $h > 0$ such that $h \to 0$ as $n \to \infty$ and $h$ does not depend on $(g, t, z) \in \mathcal{A}$}, and (ii) $K_G = K_R = K_A = K$ for a common kernel function $K$.
\end{assumption}

By condition (i), we require a common bandwidth $h$ over the estimation of $\mu_G(z)$, $\mu_R(z)$, and $\mu_A(z)$, which plays an essential role in constructing the critical value. 
The common bandwidth $h$ should not depend on $(g, t, z)$ for our uniform inference.
See Section \ref{subsec:bandwidth} and Remark \ref{remark:uniform} for further discussion of bandwidth selection.

By condition (ii), we require a common kernel function $K$ over the nonparametric regressions, which is not necessary but greatly facilitates both exposition and theoretical analysis.

We focus on the LQR estimation for the following two reasons.
First, the LQR estimation is a standard recommendation for estimating the nonparametric regression function in the kernel smoothing literature due to the boundary adaptive property (\citealp{fan1996local}).
Second, it is well known in the literature that, in combination with an appropriate choice of bandwidth, the inference based on the LQR estimation (without analytical bias correction) is numerically identical to the RBC inference based on the bias-corrected LLR estimation.
More precisely, the LQR estimator is numerically equivalent to the bias-corrected LLR estimator when the regression function estimation and the bias estimation are carried out with the same appropriate bandwidth (e.g., the IMSE-optimal bandwidth for the LLR estimation), and moreover, the asymptotic variances of the two estimators are identical.
This type of RBC inference is \textit{simple} in the sense that it does not require analytical bias correction, nor does it require adjustment of the standard error due to bias correction, unlike more sophisticated RBC inference.
In the literature on RBC inference in kernel smoothing estimation, to the best of our knowledge, only \cite{cattaneo2022boundary} consider uniform inference based on this simple RBC approach, while other previous studies focus on pointwise inference.
Building on their proposal for simple RBC inference, we propose to use the LQR estimation based on the IMSE-optimal bandwidth for the LLR estimation.
\phantomsection\label{page:AE-12-1}\Copy{AE-12-1}{
    This RBC approach could be generalized to other polynomial orders as in Section 3 of \cite{cattaneo2022boundary}, but we focus on the LQR-based inference throughout the paper.
}

\begin{remark} \label{rem:RBC}
	Another promising approach is to perform more sophisticated RBC inference based on the bias-corrected LPR estimator with analytically correcting for bias. 
	To do this, it is essential to derive the formula for the fixed-$n$ conditional variance of the bias-corrected LPR estimator (cf. \citealp{calonico2018effect}).
	However, the dependent variable in our situation, $\hat{A}_{i,g,t}$, depends on the entire sample $\{ W_i \}_{i=1}^n$, which substantially complicates the analysis for the fixed-$n$ conditional variance.
	\citet{cattaneo2022boundary} discuss an analogous problem in a different context and propose uniformly valid simple RBC inference as an alternative strategy.
\end{remark}

\begin{remark} \label{rem:secondthirdbandwidth}
    Assumption \ref{as:common} facilitates our analysis in both theory and implementation, but it may come at the expense of coverage accuracy and/or confidence interval length.
    For example, if the nuisance functions in the second-stage estimation are nearly flat over $z$, but CATT is more wavy in $z$, then it will be more desirable to use a smaller bandwidth to estimate CATT, rather than using the same bandwidth to estimate both curves.
    We do not pursue this direction further in view of the priority of theoretical and practical tractability.
\end{remark}

\subsubsection{Overview of asymptotic properties} \label{subsec:biascorrection}

We present an overview of several statistical properties of our estimator, which serve as the bases for the standard error and the critical value.
The formal results are relegated to Section \ref{subsec:asymptotic}.

We will show that the leading term of our estimator is given by 
\begin{align*}
	\hat{\DR}_{g,t}(z) - {\DR}_{g,t}(z)
	& \approx \frac{1}{f_Z(z)} \frac{1}{nh} \sum_{i=1}^n \Psi_{i,h} (B_{i,g,t} - \mu_B(Z_i)) K\left( \frac{Z_i - z}{h} \right) + \mathrm{Bias}\left[ \hat{\DR}_{g,t}(z) \; \middle| \; \bm{Z} \right],
\end{align*}
where $\bm{Z} \coloneqq (Z_1, \dots, Z_n)^\top$, $f_Z$ is the density of $Z$, $\Psi_{i,h} \coloneqq  ( I_{4,K} - u_{i,h}^2I_{2,K} ) / (I_{4,K} - I_{2,K}^2 ) $ with $u_{i,h} \coloneq (Z_i - z) / h$ and $I_{l,L} \coloneq \int u^l L(u)du$ for a non-negative integer $l$ and a function $L$, and we denote
\begin{align} \label{eq:B_E_F}
    \begin{split}
        B_{i,g,t}
        & \coloneqq A_{i,g,t} + \frac{ \mu_E(z) }{ \mu_R^2(z) } R_{i,g,t} - \frac{ \mu_F(z) }{ \mu_G^2(z) } G_{i,g}, \\
        E_{i,g,t}
        & \coloneqq R_{i,g,t} ( Y_{i,t} - Y_{i,g-1} - m_{i,g,t} ), \\
        F_{i,g,t} 
        & \coloneqq G_{i,g} ( Y_{i,t} - Y_{i,g-1} - m_{i,g,t} ).
    \end{split}
\end{align}
The second and third terms in $B_{i,g,t}$ originate from the fact that we estimate the nonparametric nuisance parameters $\mu_R(z)$ and $\mu_G(z)$ in the second-stage estimation.
Note that the effect of the first-stage estimation does not appear in this (first-order) asymptotic representation because the convergence rates of the first-stage parametric estimators are faster than the nonparametric rate.

Using this asymptotic linear representation, we can derive the asymptotic bias and variance of our estimator.
Specifically, we will show that 
\begin{align} \label{eq:bias_var_nev_2}
	\mathrm{Bias}\left[ \hat{\DR}_{g,t}(z) \;\middle|\; \bm{Z} \right] \approx h^4 \mathcal{B}_{g,t}(z),
	\qquad
	\mathrm{Var}\left[ \hat{\DR}_{g,t}(z) \;\middle|\; \bm{Z} \right] \approx \frac{1}{nh} \mathcal{V}_{g,t}(z),
\end{align}
where
\begin{align} \label{eq:bias_nev_2}
	\mathcal{B}_{g,t}(z)
	\coloneqq \frac{1}{24 f_Z(z)}\left(2\mu_B^{(3)}(z)f_Z^{(1)}(z) + \mu_B^{(4)}(z)f_Z(z)\right)  \left( \frac{I_{4,K}^2 - I_{2,K}I_{6,K}}{I_{4,K} - I_{2,K}^2} \right)
\end{align}
and
\begin{align} \label{eq:var_nev_2}
	\mathcal{V}_{g,t}(z)
	\coloneqq \frac{\sigma_B^2(z)}{f_Z(z)} \left( \frac{ I_{4,K}^2 I_{0,K^2} - 2I_{2,K}I_{4,K}I_{2,K^2} + I_{2,K}^2I_{4,K^2} }{ ( I_{4,K} - I_{2,K}^2 )^2 } \right)
\end{align}
with denoting $\sigma_B^2(z) \coloneqq \Var[B_{i,g,t} \mid Z_i = z]$ and $\mu_B^{(\nu)}(z) \coloneqq \mu_A^{(\nu)}(z) + [ \mu_E(z) / \mu_R^2(z) ] \mu_R^{(\nu)}(z) - [ \mu_F(z) / \mu_G^2(z) ] \mu_G^{(\nu)}(z)$.

\subsubsection{Standard error} \label{subsec:se}

We compute the standard error of our estimator as follows.
We start by estimating the density $f_Z(z)$ by some nonparametric method, and let $\hat f_Z(z)$ denote the resulting estimator.
Next, we compute the following variables:
\begin{align*}
    \hat B_{i,g,t}
    & \coloneqq \hat A_{i,g,t} + \frac{ \hat \mu_{\hat E}(z) }{ \hat \mu_{\hat R}^2(z) } \hat R_{i,g,t} - \frac{ \hat \mu_{\hat F}(z) }{ \hat \mu_G^2(z) } G_{i,g},\\
    \hat E_{i,g,t} 
    & \coloneqq \hat R_{i,g,t} ( Y_{i,t} - Y_{i,g-1} - \hat m_{i,g,t} ), \\ 
    \hat F_{i,g,t} 
    & \coloneqq G_{i,g} ( Y_{i,t} - Y_{i,g-1} - \hat m_{i,g,t} ),
\end{align*}
where $\hat \mu_{\hat E}(z)$ and $\hat \mu_{\hat F}(z)$ denote the LLR estimators of $\mu_E(z)$ and $\mu_F(z)$, respectively.
Using these variables, we estimate the conditional variance $\sigma_B^2(z)$ by the following LLR:
\begin{align*}
	\hat \sigma_{\hat B}^2(z)
	\coloneqq \bm{e}_0^\top \hat{\bm{\beta}}_{\hat U^2}(z),
	\qquad 
	\hat{\bm{\beta}}_{\hat U^2}(z)
	\coloneqq \argmin_{\bm{b} \in \mathbb{R}^{2}} \sum_{i=1}^n \left( \hat U_{i,g,t}^2 - \bm{r}_1(Z_i - z)^\top \bm{b} \right)^2 K_{\sigma} \left( \frac{Z_i - z}{h_{\sigma}} \right),
\end{align*}
where $\hat U_{i,g,t} \coloneqq \hat B_{i,g,t} - \hat \mu_{\hat B}(Z_i)$ and $\bm{r}_1(u) \coloneqq (1,u)^\top$.
Here, the bandwidth and the kernel function can be different from those used for the second- and third-stage estimation.
Then, we compute
\begin{align*}
    \hat{\mathcal{V}}_{g,t}(z)
    \coloneqq \frac{ \hat \sigma_{\hat B}^2(z) }{ \hat f_Z(z) } \left( \frac{ I_{4,K}^2 I_{0,K^2} - 2I_{2,K}I_{4,K}I_{2,K^2} + I_{2,K}^2I_{4,K^2} }{ ( I_{4,K} - I_{2,K}^2 )^2 } \right).
\end{align*}
The asymptotic variance in \eqref{eq:bias_var_nev_2} can be estimated by $\hat{\mathcal{V}}_{g,t}(z) / (nh)$, which leads to the following standard error of $\hat{\DR}_{g,t}(z)$: $\hat{{\SE}}_{g,t}(z) \coloneqq \sqrt{ \hat{\mathcal{V}}_{g,t}(z) / (nh) }$.

\subsubsection{Critical value} \label{subsec:UCB}

We consider two methods for constructing the critical value: 
(i) an analytical method and (ii) weighted bootstrapping.

\paragraph{Analytical method.}

We will show in Section \ref{subsec:asymptotic} that the same type of distributional approximation result as in \cite{lee2017doubly} holds even in our situation, which is based on the approximation of suprema of empirical processes by suprema of Gaussian processes \citep{chernozhukov2014gaussian} and several approximation results for suprema of Gaussian processes (\citealp{piterbarg1996asymptotic}; \citealp{ghosal2000testing}).
This in turn implies that we can compute the critical value by the same analytical method as in \cite{lee2017doubly}.
\phantomsection\label{page:AE-3-8}\Copy{AE-3-8}{
    Specifically, we consider the following critical value for the two-sided symmetric uniform confidence band
    \begin{align} \label{eq:analytical-critical}
	   \hat c(1 - \alpha)
	   \coloneqq \left( a_n^2 - 2 \log \left( \log \left( \frac{1}{\sqrt{1 - \alpha}} \right) \right) \right)^{1/2},
    \end{align}
    where
    \begin{align} \label{eq:lambda}
	   a_n^2
	   \coloneqq 2 \log \left( \frac{ b - a }{ h } \right) + 2 \log \frac{ \sqrt{\lambda} }{ 2 \pi },
	   \qquad
	   \lambda
	   \coloneqq - \frac{ \int K(u) K^{(2)}(u) du }{ \int K^2(u) du }, 
    \end{align}
    with $b - a$ corresponding to the length of the interval $\mathcal{I} = [a, b]$.
    Note that the common bandwidth condition in Assumption \ref{as:common}(i) ensures that $\hat c(1 - \alpha)$ and $a_n^2$ do not depend on $(g, t, z)$.
    The proposed $(1 - \alpha)$ uniform confidence band over $(g,t,z) \in \mathcal{A}$ is $\hat{\mathcal{C}} \coloneqq \{ \hat{\mathcal{C}}_{g,t}(z): (g,t,z) \in \mathcal{A} \}$, where
    \begin{align} \label{eq:analytical_UCB}
	   \hat{\mathcal{C}}_{g,t}(z)
	   \coloneqq 
	   \left[ 
	   \hat{\DR}_{g,t}(z) - \hat c(1 - \alpha) \cdot \hat{{\SE}}_{g,t}(z),
	   \qquad 
	   \hat{\DR}_{g,t}(z) + \hat c(1 - \alpha) \cdot \hat{{\SE}}_{g,t}(z)
	   \right].
    \end{align}
}

\paragraph{Weighted bootstrapping.}

\phantomsection\label{page:AE-17-1}\Copy{AE-17-1}{
    As an alternative to the analytical method, we can consider weighted bootstrap inference. 
    Building on \cite{ma2005robust}, \cite{chen2009efficient}, and \cite{fan2022estimation}, we propose the following algorithm.
}
For each $b = 1, \dots, B$, we generate a set of IID bootstrap weights $\{V_i^{\star,b}\}_{i=1}^n$ independently of $\{ W_i \}_{i=1}^n$, such that $\bE[V_i^{\star,b}] = 1$, $\Var[V_i^{\star,b}] = 1$, and its distribution has sub-exponential tails.
Common choices include a normal random variable with unit mean and unit variance and
\phantomsection\label{page:AE-8-1}\Copy{AE-8-1}{
    \citeauthor{mammen1993bootstrap}'s (\citeyear{mammen1993bootstrap}) wild bootstrap weights such that $\mathbb{P}(V_i^{\star,b} = 2 - c_v) = c_v / \sqrt{5}$ and $\mathbb{P}(V_i^{\star,b} = 1 + c_v) = 1 - c_v / \sqrt{5}$ with $c_v = (\sqrt{5} + 1) / 2$.
}
\phantomsection\label{page:AE-3-9}\Copy{AE-3-9}{
    In each bootstrap repetition, we compute the bootstrapped LQR estimator
    \begin{align*}
	   \hat{\DR}^{\star,b}_{g,t}(z) \coloneqq \hat \mu_{\hat A}^{\star,b}(z) \coloneqq \bm{e}_0^\top \hat{\bm{\beta}}^{\star,b}_{\hat A}(z),
    	\quad 
	   \hat{\bm{\beta}}_{\hat A}^{\star,b}(z) & \coloneqq \argmin_{\bm{b} \in \mathbb{R}^{3}} \sum_{i = 1}^n V_i^{\star,b} \left( \hat A_{i,g,t} - \bm{r}_2(Z_i - z)^\top \bm{b} \right)^2 K \left( \frac{Z_i - z}{h} \right),
    \end{align*}
    and the supremum of the bootstrap counterpart of the studentized statistic
    \begin{align} \label{eq:Mb}
	   M^{\star,b}
	   \coloneqq \sup_{(g,t,z) \in \mathcal{A}} \frac{ \left| \hat{\DR}^{\star,b}_{g,t}(z) - \hat{\DR}_{g,t}(z)\right|}{\hat{\SE}_{g,t}(z)}.
    \end{align}
    Here, we should use the same $\hat A_{i,g,t}$, the same bandwidth $h$, and the same kernel function $K$ as the original estimator $\hat{\DR}_{g,t}(z)$.
    Let $ \tilde c(1 - \alpha)$ be \text{the empirical $(1 - \alpha)$ quantile of $\{ M^{\star,b} \}_{b = 1}^B$.}
    Note that $\tilde c(1 - \alpha)$ and $M^{\star,b}$ do not depend on $(g, t, z)$ due to the supremum taken in the definition of $M^{\star,b}$.
    The $(1 - \alpha)$ uniform confidence band over $(g,t,z) \in \mathcal{A}$ is $\tilde{\mathcal{C}} \coloneqq \{ \tilde{\mathcal{C}}_{g,t}(z): (g,t,z) \in \mathcal{A} \}$, where
    \begin{align}\label{eq:mb_UCB}
	   \tilde{\mathcal{C}}_{g,t}(z) 
	   \coloneqq 
	   \left[ \hat{\DR}_{g,t} (z) - \tilde{c}(1 - \alpha) \cdot \hat{{\SE}}_{g,t}(z), 
	   \qquad 
	   \hat{\DR}_{g,t} (z) + \tilde{c}(1 - \alpha) \cdot \hat{{\SE}}_{g,t}(z) \right].
    \end{align}
}

\begin{remark} \label{remark:pointwise}
    If our goal is to construct a pointwise confidence interval of $\CATT_{g,t}(z)$ for a given $(g, t, z)$, we can use standard Wald-type inference.
    For the analytical method, the pointwise confidence interval can be obtained by replacing $\hat c(1 - \alpha)$ in \eqref{eq:analytical_UCB} with the $(1 - \alpha/2)$ quantile of the standard normal distribution, which is theoretically justified by the asymptotic linearity in Section \ref{subsec:asymptotic} and the Lyapunov central limit theorem.
    For weighted bootstrapping, we can get the pointwise critical value simply by not taking the supremum in \eqref{eq:Mb}, and the validity of the resulting pointwise confidence interval can be shown in the same way as in Theorem \ref{thm:mb_valid}.
\end{remark}

\begin{remark} \label{remark:refinement} 
    We compare the analytical method with weighted bootstrapping by Monte Carlo experiments and find that 
    \phantomsection\label{page:AE-18-1}\Copy{AE-18-1}{
        bootstrapping generally leads to better uniform coverage, suggesting that bootstrap inference may have some asymptotic refinements.
    }
    See Appendix \ref{sec:simulation} for the simulation experiments.
\end{remark}

\subsubsection{Theoretical justifications} \label{subsec:asymptotic}

This subsection presents theoretical justifications for the proposed methods.

We impose the following set of mild regularity conditions.
Hereafter, to simplify the exposition, we often write ``for $(g, t)$'' to refer to a generic $(g, t)$ such that $(g, t, z) \in \mathcal{A}$ for $z \in \mathcal{I}$.

\begin{assumption} \label{as:asymptotic1}
	\hfill 
	\begin{enumerate}[(i)]
		\item The distribution of $Z$ has a Lebesgue density $f_Z$ that is five-times continuously differentiable and bounded above and away from zero on $\mathcal{I}$.
		\item The kernel function $K$ is compactly supported, symmetric around zero, and six-times differentiable.
		\item The bandwidth $h$ satisfies $C n^{- 1/2 + \varepsilon} \le h \le C n^{ - 1/9 - \varepsilon }$ for some positive constants $C$ and $\varepsilon$.
		\item Let $Q$ be a generic notation for $A_{g,t}$, $E_{g,t}$, $F_{g,t}$, $G_g$, or $R_{g,t}$.
		Denote $\mu_Q(z) = \bE[ Q \mid Z = z ]$.
		\begin{enumerate}[(a)]
			\item For all $(g, t)$, $\mu_Q$ is five-times continuously differentiable on $\mathcal{I}$.
			\item $\displaystyle \sup_{(g,t,z) \in \mathcal{A}} \bE[ Q^4 \mid Z = z ] < \infty$.
		\end{enumerate}
		\item The first-stage parametric estimators are $1 / \sqrt{n}$-consistent so that, for all $(g, t)$, 
		\begin{align*}
			& \max_{1 \le i \le n} | \hat R_{i,g,t} - R_{i,g,t} | = O_{\bP} \left( \frac{1}{\sqrt{n}} \right),
			& & \max_{1 \le i \le n} | \hat m_{i,g,t} - m_{i,g,t} | = O_{\bP} \left( \frac{1}{\sqrt{n}} \right), \\
			& \max_{1 \le i \le n} | \hat E_{i,g,t} - E_{i,g,t} | = O_{\bP} \left( \frac{1}{\sqrt{n}} \right),
			& & \max_{1 \le i \le n} | \hat F_{i,g,t} - F_{i,g,t} | = O_{\bP} \left( \frac{1}{\sqrt{n}} \right) .
		\end{align*}
		\item For all $(g, t)$, the conditional variance $\sigma_B^2$ is continuously differentiable on $\mathcal{I}$.
	\end{enumerate}
\end{assumption}

Among these assumptions, the undersmoothing condition on the common bandwidth $h$ in Assumption \ref{as:asymptotic1}(iii) is particularly important for our analysis.
This assumption ensures that the asymptotic bias $h^4 \mathcal{B}_{g,t}(z)$ is asymptotically negligible when constructing the uniform confidence band.
As such, the assumption rules out, for example, computing the LQR estimator $\hat{\DR}_{g,t}(z)$ using the IMSE-optimal bandwidth for the LQR estimation, which is of order $O(n^{-1 /9})$.
To fulfill the assumption, we propose using the IMSE-optimal bandwidth for the LLR estimation, not for the LQR estimation, based on the insight of the simple RBC.
See Section \ref{subsec:bandwidth} for details.

\phantomsection\label{page:AE-21-2}\Copy{AE-21-2}{
    Note that our undersmoothing condition accommodates any undersmoothing bandwidth, not limited to our proposal based on the simple RBC approach, as long as the rate required in Assumption \ref{as:asymptotic1}(iii) is satisfied.
    For example, as in \cite{lee2017doubly} and \cite{fan2022estimation}, we can consider a rule-of-thumb adjustment that achieves undersmoothing by shrinking the IMSE-optimal bandwidth for the LQR estimation obtained from the plug-in or cross-validation method by $n^{-\varepsilon}$ for some appropriate $\varepsilon > 0$.
    However, we prefer the simple RBC approach to the rule-of-thumb adjustment to follow the recent literature that demonstrates the desirable performance of the RBC approach.
}

To facilitate our theoretical investigations, Assumption \ref{as:asymptotic1} contains several high-level conditions that can actually be replaced by less restrictive but more complicated conditions.
For example, the compact support condition in Assumption \ref{as:asymptotic1}(ii) can be replaced by other conditions that guarantee the existence of technical moments related to the kernel function at the expense of complicated proofs.
This in turn implies that commonly used kernel functions (e.g., the Gaussian kernel) should be permissible for use with our analysis.

The next theorem formalizes the asymptotic linear representation and the asymptotic bias and variance formulas described in Section \ref{subsec:biascorrection}.
Hereafter, with an abuse of notation, we often write $o_{\bP}( \sqrt{ (\log n) / (nh) } )$ to indicate $O_{\bP}( \sqrt{ (\log n) / (n^{1 + \varepsilon} h) } )$ for some $\varepsilon > 0$ to simplify notation.
In addition, 
\phantomsection\label{page:AE-21-3}\Copy{AE-21-3}{
    the theorems presented below treat the bandwidth $h$ as a deterministic sequence, as with many prior studies in the kernel smoothing literature.
    While investigating the effects of using a data-driven stochastic bandwidth would be interesting, it is beyond the scope of this paper to develop the theory to handle stochastic bandwidths. 
}

\begin{theorem} \label{thm:bias_variance}
    Suppose that Assumptions \ref{as:staggered}--\ref{as:asymptotic1} hold.
    When $n \to \infty$, we have
    \begin{align*}
        & \hat{\DR}_{g,t}(z) - {\DR}_{g,t}(z) \\
	& = \frac{1}{f_Z(z)} \frac{1}{nh} \sum_{i=1}^n  \Psi_{i,h} (B_{i,g,t} - \mu_B(Z_i)) K\left( \frac{Z_i - z}{h} \right) + \mathrm{Bias}\left[ \hat{\DR}_{g,t}(z) \; \middle| \; \bm{Z} \right] + o_{\bP}\left( h^4 \right) + o_{\bP}\left( \sqrt{\frac{\log n}{nh}} \right),
    \end{align*}
    and
    \begin{align*}
        \mathrm{Bias}\left[ \hat{\DR}_{g,t}(z) \; \middle| \; \bm{Z} \right]
	= h^4 \mathcal{B}_{g,t}(z) + o_{\bP} \left( h^4\right) 
	= o_{\bP} \left( \frac{1}{nh} \right),
        \quad 
        \Var\left[ \hat{\DR}_{g,t}(z) \; \middle| \; \bm{Z} \right]
	= \frac{1}{nh} \mathcal{V}_{g,t}(z) + o_{\bP} \left( \frac{1}{nh} \right),
    \end{align*}
    where $\mathcal{B}_{g,t}(z)$ and $\mathcal{V}_{g,t}(z)$ are defined in \eqref{eq:bias_nev_2} and \eqref{eq:var_nev_2} and the convergence rates of the remainder terms hold uniformly in $(g, t, z) \in \mathcal{A}$. 
\end{theorem}

\smallskip 

Next, we present a theoretical justification for the uniform confidence band constructed with the analytical method in Section \ref{subsec:UCB}.
For this purpose, we consider a uniformly valid distributional approximation result for the studentized statistic.
To proceed, we rewrite the standard error as $\hat{\SE}_{g,t}(z) = \hat{\mathcal{S}}_{g,t}(z) / \sqrt{nh}$, where we denote $\hat{\mathcal{S}}_{g,t}(z) \coloneqq \sqrt{\hat{\mathcal{V}}_{g,t}(z)}$ as the estimator of $\mathcal{S}_{g,t} \coloneq \sqrt{\mathcal{V}_{g,t}(z)}$.

We add the following regularity conditions, which is essentially the same as the conditions in Assumption 1 of \cite{lee2017doubly}.
Letting $U_{i,g,t} \coloneqq B_{i,g,t} - \mu_B(Z_i)$ denote the population counterpart of $\hat U_{i,g,t}$, we write the standard deviation of the $\sqrt{nh}$ times leading term in Theorem \ref{thm:bias_variance} as $\tilde{\mathcal{S}}_{g,t}(z) \coloneqq h^{-1/2} f_Z^{-1}(z) \sqrt{ \bE [ \Psi_{i,h}^2 U_{i,g,t}^2 K^2 ( (Z_i - z) / h ) ] }.$
\phantomsection\label{page:AE-20-1}\Copy{AE-20-1}{
    This quantity depends on $h$ and also on $n$ through $h$, but we suppress the dependence to simplify notation.
}

\begin{assumption} \label{as:asymptotic2}
	\hfill 
	\begin{enumerate}[(i)]
            \item $\inf_{n \in \mathbb{N}} \inf_{(g,t,z) \in \mathcal{A}} \tilde{\mathcal{S}}_{g,t}(z) > 0$ and $\tilde{\mathcal{S}}_{g,t}(z)$ is continuous in $z \in \mathcal{I}$ for any $(g, t)$ and $n \in \mathbb{N}$.
		\item $\sup_{(g,t,z) \in \mathcal{A}} \left| \hat{\mathcal{S}}_{g,t}(z) - \tilde{\mathcal{S}}_{g,t}(z) \right| = O_{\bP}(n^{-c})$ for some positive constant $c$.
		\item $\sup_{(g,t,z) \in \mathcal{A}} \bE[U_{i,g,t}^4 \mid Z_i=z] < \infty$.
		\item For all $(g, t)$, $\bE[U_{i,g,t}^2 \mid Z_i=z]f(z)$ is Lipschitz continuous in $z\in\mathcal{I}$.
	\end{enumerate}
\end{assumption}

The next theorem gives a uniformly valid distributional approximation result.

\begin{theorem} \label{thm:approx_us_distribution}
    Suppose that Assumptions \ref{as:staggered}--\ref{as:asymptotic2} hold.
    When $n \to \infty$, there exists $\kappa > 0$ such that, uniformly in $s$, on any finite interval:
    \begin{align*}
        \mathbb{P} \left( a_n \left[ \sup_{(g,t,z) \in \mathcal{A}} \left| \frac{ \hat{ \DR }_{ g,t}(z)  - { \DR }_{ g,t }(z) }{ \hat{{\SE}}_{ g,t }(z) } \right| - a_n \right] < s \right) = \exp\left( -2e^{-s-s^2/(2a_n^2)} \right) + O(n^{-\kappa}),
    \end{align*}
    where $a_n$ and $\lambda$ are defined in \eqref{eq:lambda}, and $a_n$ is the largest solution of $(b-a) (2\pi h)^{-1} \lambda^{1/2} \exp ( - a_n^2/2 ) = 1.$
\end{theorem}

This theorem in turn justifies the use of the analytical critical value $\hat c(1 - \alpha)$ defined in \eqref{eq:analytical-critical}.
To see this, denoting $s_n \coloneqq a_n + s / a_n$, observe that
\begin{align*}
    \mathbb{P} \left( \sup_{ (g,t,z) \in \mathcal{A} } \left| \frac{ \hat{ \DR }_{ g,t }(z)  - { \DR }_{ g,t }(z) }{ \hat{{\SE}}_{ g,t }(z) } \right| < s_n \right)
    \approx \exp\left( -2e^{(a_n^2 - s_n^2)/2} \right).
\end{align*}
In words, we can approximate the distribution function of the supremum of the studentized statistic by the right-hand side.
Then, it is easy to see that the critical value $\hat c(1 - \alpha)$ defined in \eqref{eq:analytical-critical} is the approximated $1 - \alpha$ quantile such that $1 - \alpha = \exp \left( - 2e^{(a_n^2 - \hat c(1 - \alpha)^2)/2} \right)$.
As a result, the $(1 - \alpha)$ uniform confidence band $\hat{\mathcal{C}}$ defined in \eqref{eq:analytical_UCB} has the desired coverage:
\begin{align*}
    \bP \left( {\CATT}_{g,t}(z) \in \hat{\mathcal{C}}_{g,t}(z) \; \text{for all $(g,t,z) \in \mathcal{A}$} \right)
    = \bP \left( \sup_{ (g,t,z) \in \mathcal{A} } \left| \frac{ \hat{ \DR }_{ g,t }(z)  - { \DR }_{ g,t }(z) }{ \hat{{\SE}}_{ g,t }(z) } \right| \le \hat c(1 - \alpha) \right)
    \approx 1 - \alpha.
\end{align*}

\smallskip 

Lastly, we study several theoretical properties for the uniform confidence band constructed with weighted bootstrapping in Section \ref{subsec:UCB}.
For this purpose, we make the following assumption on the bootstrap weight.
Hereafter, to simplify notation, we suppress the superscript $b$ indicating the bootstrap repetition.

\begin{assumption} \label{as:mb} 
	The random variable $\{V_i^\star\}_{i=1}^n$ is independent of $\{W_i\}_{i=1}^n$ and satisfies $\bE[V_i^\star] = 1$, $\Var[V_i^\star] = 1$, and $\bP(|V_i^\star| > x) \le C_1 \exp(-C_2x)$ for every $x$ and some constants $C_1$ and $C_2$.
\end{assumption}

The following theorem gives an asymptotic linear representation for the bootstrap estimator $\hat{\DR}_{g,t}^{\star}(z)$.
\phantomsection\label{page:AE-17-3}\Copy{AE-17-3}{
    It implies that our weighted bootstrap procedure for CATT is (first-order) asymptotically equivalent to the multiplier bootstrap procedure considered in previous studies (e.g., \citealp{callaway2021difference}).
}
The proof is almost the same as that of Theorem \ref{thm:bias_variance}, and is thus omitted.

\begin{theorem} \label{thm:mb_linearize}
	 Suppose that Assumptions \ref{as:staggered}--\ref{as:asymptotic1} and \ref{as:mb} hold.
	When $n \to \infty$,  we have
	\begin{align*}
		\hat{\DR}_{g,t}^{\star}(z) - \hat{\DR}_{g,t}(z)
		= \frac{1}{f_Z(z)} \frac{1}{nh} \sum_{i=1}^n (V_i^{\star} - 1) \Psi_{i,h} \big( B_{i,g,t} - \mu_B(Z_i) \big) K\left( \frac{Z_i - z}{h} \right) +  o_{\bP}\left( \sqrt{\frac{\log n}{nh}} \right),
	\end{align*}
	where the convergence rate of the remainder term holds uniformly in $(g, t, z) \in \mathcal{A}$.
\end{theorem}

Using this result and Corollary 3.1 of \cite{chernozhukov2014anti}, the next theorem proves the validity of the uniform confidence band constructed with the weighted bootstrap procedure in \eqref{eq:mb_UCB}.

\begin{theorem} \label{thm:mb_valid}
    Suppose that Assumptions \ref{as:staggered}--\ref{as:mb} hold.
    When $n \to \infty$, for some positive constants $c$ and $C$, we have $\bP ( {\CATT}_{g,t}(z) \in \tilde{\mathcal{C}}_{g,t}(z) \; \text{for all $(g,t,z) \in \mathcal{A}$} ) \ge (1 - \alpha) - C n^{-c}$.
\end{theorem}

\begin{remark} \label{remark:uniform}
\phantomsection\label{page:AE-21-1}\Copy{AE-21-1}{
    It is the key for Theorems \ref{thm:approx_us_distribution} and \ref{thm:mb_valid} that the critical values $\hat c(1 - \alpha)$ and $\tilde c(1 - \alpha)$ do not depend on the values of $(g, t, z)$ in $\mathcal{A}$.
    In other words, if we consider critical values that change depending on the values of $(g, t, z)$ or $(g, t)$, there is no guarantee that the uniform confidence bands obtained with those critical values are uniformly valid over $(g, t, z)$ (although they can achieve pointwise validity or uniform validity only over $z$).
    Since $\hat c(1 - \alpha)$ for the analytical method depends on $h$ through $a_n$, the validity of the analytical method relies on the use of a bandwidth that is independent of $(g, t, z)$.
    In contrast, $\tilde c(1 - \alpha)$ for the bootstrap inference does not depend on the values of $(g, t, z)$ even if we estimate $\CATT_{g,t}(z)$ by using bandwidths that vary with the values of $(g, t, z)$, which should be clear from the definition of $M^{\star,b}$ in \eqref{eq:Mb}.
    Nevertheless, we recommend using a common bandwidth that is independent of the values of $(g, t, z)$ for both the analytical and bootstrap methods.
    This is to prevent that, for each $(g, t, z) \in \mathcal{A}$, the point estimate of $\CATT_{g,t}(z)$ used for the analytical method differs from that used for the bootstrap inference, which would complicate the interpretation of the point estimates between the analytical and bootstrap methods.
}
\end{remark}

\begin{remark}
    Theorems \ref{thm:approx_us_distribution} and \ref{thm:mb_valid} show that the approximation errors for the analytical method and weighted bootstrapping decrease polynomially in the cross-sectional size $n$.
    This approximation accuracy is more desirable than the classical uniform confidence band constructed with the following critical value based on the Gumbel distribution: $\check c(1 - \alpha)
        \coloneqq a_n -a_n^{-1}\log\{ \log[(1 - \alpha)^{-1/2}] \}  $,
    where the definition of $a_n$ is the same as in \eqref{eq:lambda}.
    This critical value leads to the uniform confidence band with the approximation error of a logarithmic rate.
    See \cite{chernozhukov2014anti} and \cite{lee2017doubly} for further discussion.
\end{remark}

\begin{remark} \label{remark:chernozhukov2014anti} \phantomsection\label{page:AE-22}\Copy{AE-22}{
    Similar to Theorem 4.2 of \citet{fan2022estimation} in the unconfoundedness setup, Theorem \ref{thm:mb_valid} in this paper is a direct application of Corollary 3.1 and Theorem 3.2 of \citet{chernozhukov2014anti}.
    While \cite{chernozhukov2014anti} present their results through nonparametric density estimation, their results are applicable to general nonparametric problems, as demonstrated in the context of nonparametric regressions in \citet{fan2022estimation} and this paper.
    It would be worthwhile to expand this line of research to infer treatment effect heterogeneity in other research designs.
}
\end{remark}

\subsubsection{Bandwidth selection} \label{subsec:bandwidth}

Recall that the construction of our uniform confidence bands relies on the common bandwidth choice and the undersmoothing condition in Assumptions \ref{as:common}(i) and \ref{as:asymptotic1}(iii).
Here, we propose to use an optimal bandwidth in terms of IMSE for the LLR estimation, instead of that for the LQR estimation.
This proposal is based on the insight that the LQR estimation with the IMSE-optimal bandwidth for the LLR estimation can be interpreted as the simple RBC inference, as discussed in Section \ref{subsec:procedure}.

In line with Assumption \ref{as:common}(i), we consider the common bandwidth $h_\mathtt{LL} \coloneqq \min_{(g, t)} h_\mathtt{LL}(g, t)$, where $h_\mathtt{LL}(g, t)$ is the IMSE-optimal bandwidth for the LLR estimation for each $(g, t)$, defined as follows.
\phantomsection\label{page:AE-23-1}\Copy{AE-23-1}{
    Let $\hat{\DR}_{g,t}^{\mathtt{LL}}(z)$ denote the LLR estimator of the DR estimand, whose definition can be found in Appendix \ref{sec:notation}.
}
In Appendix \ref{subsec:LLR}, we show that its asymptotic bias and variance are
\begin{align} \label{eq:bias_var_nev_1}
	\mathrm{Bias}\left[ \hat{\DR}_{g,t}^{\mathtt{LL}}(z) \;\middle|\; \bm{Z} \right] \approx h^2 \frac{I_{2,K}}{2} \mu_B^{(2)}(z),
	\qquad
	\mathrm{Var}\left[ \hat{\DR}_{g,t}^{\mathtt{LL}}(z) \;\middle|\; \bm{Z} \right] \approx \frac{I_{0,K^2}}{nh} \frac{ \sigma_B^2(z) }{ f_Z(z) }.
\end{align}
Thus, 
it is easy to see that the (infeasible) IMSE-optimal bandwidth for the LLR estimator is given by
\begin{align} \label{eq:IMSE_bw}
	h_{\mathtt{LL}}(g,t)
	\coloneqq \left( \frac{ I_{0,K^2}\int_{z \in \mathcal{I}} f_Z^{-1}(z)\sigma
 _B^2(z)dz }{I_{2,K}^2 \int_{z \in \mathcal{I}} [\mu_B^{(2)}(z)]^2 dz } \right)^{1/5} n^{-1/5}.
\end{align}
Notice that $h_{\mathtt{LL}}(g,t)$ and thus $h_{\mathtt{LL}}$ are of order $n^{-1/5}$, which satisfies Assumption \ref{as:asymptotic1}(iii).

To construct feasible counterparts of $h_{\mathtt{LL}}$ and $h_{\mathtt{LL}}(g,t)$, we need to estimate the unknown quantities.
The estimators for the density $f_Z(z)$ and the conditional variance $\sigma_B^2(z)$ are already available as in Section \ref{subsec:se}.
The second order derivative $\mu_B^{(2)}(z)$ can be estimated by the $p_B$-th order LPR estimation with some $p_B \ge 2$, where the dependent variable is $\hat B_{i,g,t}$ and the regressors are the $p_B$-th order polynomials $\bm{r}_{p_B}(Z_i - z)$.
Let $\hat \mu_{\hat B}^{(2)}(z)$ denote this LPR estimator.
Then, we obtain the feasible bandwidths, say $\hat h_\mathtt{LL} \coloneqq \min_{(g, t)} \hat h_\mathtt{LL}(g, t)$ and $\hat h_{\mathtt{LL}}(g,t)$, as the estimated counterparts.
Note that using $\hat h_\mathtt{LL}$ implicitly assumes that all nonparametric functions to be estimated have the same smoothness over all $(g, t, z) \in \mathcal{A}$, which may be restrictive for the same reason as in Remark \ref{rem:secondthirdbandwidth}.

Since our undersmoothing approach is based on the insight of simple RBC, it should be preferable to the traditional rule-of-thumb undersmoothing strategy and the conventional MSE-optimal bandwidth-based inference in terms of achieving \textit{both} good coverage and a short confidence interval length.
In the literature on uniform inference for treatment effect heterogeneity, \citet{lee2017doubly} and \citet{fan2022estimation} propose the rule-of-thumb undersmoothing strategy for the LLR estimation, say $\hat h_{\US} \coloneqq \hat h_{\mathrm{LL}} \cdot n^{1/5} \cdot n^{-2/7}$ in our context.
In contrast, several previous studies in the kernel smoothing literature show in theory and numerical analysis that (simple) RBC inference generally leads to better coverage and shorter confidence interval length (e.g., \citealp{calonico2018effect} for pointwise inference in a general setting of kernel smoothing estimation).
Thus, we can expect our undersmoothing approach to have the same nice properties.
Since our proposal is not based on formal theory but on analogy with previous studies, we thoroughly investigate its performance through a series of Monte Carlo experiments in Appendix \ref{sec:simulation}, the results of which corroborate our discussion here.
We leave more sophisticated RBC inference and bandwidth choices for uniform inference as future topics.

\begin{remark} \label{remark:snooping}
\phantomsection\label{page:AE-24-1}\Copy{AE-24-1}{
    As shown in Section \ref{subsec:asymptotic}, our proposal is uniformly valid over $(g, t, z) \in \mathcal{A}$, but there is no guarantee of uniform validity across a set of bandwidths or data generating processes.
    One important potential consequence of this is the issue of \textit{bandwidth snooping} \citep{armstrong2018simple}: Our uniform confidence bands may have undesirable coverage if we perform sensitivity analyses or robustness checks involving multiple bandwidths, including non-data-driven ones.
    Developing inference procedures that are robust to bandwidth snooping and have uniform validity across data generating processes is an important area of future research.
}
\end{remark}

\section{Inference for Summary Parameters} \label{sec:Summary}

We turn to statistical inference for the aggregated parameter $\theta(z)$ defined in \eqref{eq:summary}.
Since the weighting function $w_{g,t}(z)$ is known or estimable, we can compute the aggregated estimator as follows:
\begin{align*}
    \hat{\theta}(z) 
    \coloneqq \sum_{g\in\mathcal{G}} \sum_{t=2}^{\mathcal{T}} \hat{w}_{g,t}(z) \cdot \hat{\DR}_{g,t}(z), 
\end{align*}
where $\hat{w}_{g,t}(z) = w_{g,t}(z)$ if $w_{g,t}(z)$ is known, otherwise $\hat{w}_{g,t}(z)$ is a nonparametric estimator constructed with certain LQR estimation depending on the form of $w_{g,t}(z)$.

To perform the uniform inference for the aggregated parameter $\theta(z)$, we can construct the standard error and the uniform critical value in the same manner as in the case of CATT.
To be specific, focusing on the case where $w_{g,t}(z)$ is unknown but estimable, suppose that its estimator satisfies
\begin{align*}
    & \hat{w}_{g,t}(z) - w_{g,t}(z) \\
    & = \frac{1}{f_Z(z)} \frac{1}{nh} \sum_{i=1}^n \Psi_{i,h} (\xi_{i,g,t} - \mu_{\xi}(Z_i)) K \left( \frac{Z_i - z}{h} \right) + \mathrm{Bias}\left[ \hat{w}_{g,t}(z) \; \middle| \; \bm{Z} \right] + o_{\bP}(h^4) +o_{\bP}\left( \sqrt{\frac{\log n}{nh}}\right),
\end{align*}
where $\xi_{i,g,t}$ is an estimable variable whose definition depends on the choice of $w_{g,t}(z)$.
In Appendix \ref{sec:supp-summary}, we show that this asymptotic linearity holds for a variety of weighting functions of interest.
Then, $\hat{\theta}(z)$ exhibits the same form of asymptotic linearity as in the case of CATT in Theorem \ref{thm:bias_variance}:
\begin{align} \label{eq:linear_summary}
\begin{split}
    & \hat{\theta}(z) - \theta(z) \\
    & = \frac{1}{f_Z(z)} \frac{1}{nh} \sum_{i=1}^n \Psi_{i,h} (J_i - \mu_J(Z_i)) K\left( \frac{Z_i - z}{h} \right) + \mathrm{Bias}\left[ \hat{\theta}(z) \; \middle| \; \bm{Z} \right] + o_{\bP}(h^4) + o_{\bP}\left( \sqrt{\frac{\log n}{nh}}\right),
\end{split}
\end{align}
where $J_i \coloneqq \sum_{g\in\mathcal{G}} \sum_{t=2}^{\mathcal{T}} [ w_{g,t}(z) \cdot B_{i,g,t} + {\DR}_{g,t}(z) \cdot \xi_{i,g,t} ]$.
Using this result, we can derive the asymptotic bias and variance formulas for the aggregated estimator $\hat{\theta}(z)$ in the same way as in Theorem \ref{thm:bias_variance}.
Moreover, the construction of the uniform critical value and the bandwidth selection are essentially the same as in the previous section.

To save space in the main text, we relegate the proof of the asymptotic linear representation in \eqref{eq:linear_summary} to Appendix \ref{sec:supp-summary}.
The appendix also presents the formulas for the standard error and the uniform critical values based on the analytical method and multiplier bootstrapping, and discusses how to implement the uniform inference for concrete summary parameters.

\if0\blind
{
	\begin{center}
		{\large\bf ACKNOWLEDGMENTS}
	\end{center}
    The authors appreciate the valuable comments of the coeditor (Ivan Canay), the associate editor, and two anonymous referees, which greatly improved the paper.
    The authors also thank Ryo Okui and Toshiki Tsuda for their helpful discussions.
	This work was supported by JSPS KAKENHI Grant Number 20K13469 and 24KJ1472.
} \fi
\if1\blind
{
} \fi 
\begin{center}
	{\large\bf DISCLOSURE STATEMENT}
\end{center}
The authors report there are no competing interests to declare.
\begin{center}
{\large\bf SUPPLEMENTARY MATERIAL}
\end{center}
\begin{description}
\item[Supplement:] The proofs of all technical results and the other supplementary results (PDF file).
\item[Replication:] The R codes to reproduce the numerical results (ZIP file).
\end{description}

\clearpage

\bibliographystyle{tandfx}
\bibliography{ref}

\clearpage

\appendix 

\renewcommand{\thepage}{S\arabic{page}}
\renewcommand{\thetable}{S\arabic{table}}
\renewcommand{\thefigure}{S\arabic{figure}}
\setcounter{page}{1}
\setcounter{table}{0}
\setcounter{figure}{0}

\if0\blind
{
	\noindent \title{\Large \bf Supplementary Appendix for \\
		``Doubly Robust Uniform Confidence Bands for Group-Time Conditional Average Treatment Effects in Difference-in-Differences''}
	\maketitle
} \fi

\if1\blind
{
	\bigskip
	\bigskip
	\bigskip
	\begin{center}
		{\LARGE\bf Supplementary Appendix for ``Doubly Robust Uniform Confidence Bands for Group-Time Conditional Average Treatment Effects in Difference-in-Differences''}
	\end{center}
	\medskip
} \fi

\bigskip
\begin{abstract}
	Appendix \ref{sec:notation} presents a list of the key symbols introduced in the main text.
	Appendices \ref{sec:proof} and \ref{sec:lemma} contain the proofs of the main theorems and lemmas.
	In Appendices \ref{sec:not-yet-treated} and \ref{sec:never-treated}, we consider the analysis under limited treatment anticipation and the analysis using the never-treated group as the comparison group, respectively. 
	Appendix \ref{sec:supp-summary} contains additional discussions of the inference methods for the summary parameters.
	In Appendix \ref{sec:pretrends}, we discuss the usefulness of our uniform inference methods for examining pre-trends.
	Appendix \ref{sec:simulation} presents the results of Monte Carlo experiments.
	Appendix \ref{sec:empirical2} contains additional details on the empirical analysis.
\end{abstract}

\newpage
\spacingset{1.8} 



\section{Notation} \label{sec:notation}

For easy reference, Table \ref{tab:notation} provides a list of the key symbols we introduce in the main text.

\begin{table}[h]
	\centering
	\caption{List of the key symbols introduced in the main text.}
	\label{tab:notation}
	\begin{tabular}{lll}
		\toprule
		Symbol & Description & Equation \\
		\midrule 
		$Y_{i,t}$ & The scalar outcome for unit $i$ in period $t$. & \\
		$D_{i,t}$ & The binary treatment. & \\
		$X_i = (Z_i, X_{i,\sub})$ & The vector of the pre-treatment covariates. & \\
		$W_i = \{ Y_{i,t}, D_{i,t}, X_i \}_{t=1}^{\mathcal{T}}$ & The vector of the observed variables for unit $i$. & \\
		$G_i = \min \{ t: D_{i,t} = 1 \}$ & The group. & \\
		$G_{i,g} = \bm{1}\{ G_i = g \}$ & The indicator for whether unit $i$ belongs to group $g$. & \\
		$Y_{i,t}(g)$ and $Y_{i,t}(0)$ & The potential outcomes. & \\
		$\mathcal{G} = \mathrm{supp}(G) \setminus \{ \bar g \}$ & The set of realized treatment timings before $\bar g = \max_{1 \le i \le n} G_i$. & \\
		$\mathcal{I} = [a, b]$ & A proper closed subset of the support of $Z$. & \\
		$\CATT_{g,t}(z)$ & The group-time conditional average treatment effect (CATT). & \eqref{eq:CATT} \\
		$\mathcal{A}$ & The set of $(g, t, z)$ on which CATT is identifiable. & \eqref{eq:mathcalA} \\
		$\mathcal{C}_{g,t} = \{ \mathcal{C}_{g,t}(z) \}$ & The uniform confidence band for CATT. & \eqref{eq:UCB_CATT} \\
		$\theta(z)$ and $w_{g,t}(z)$ & The aggregated parameter and corresponding weight. & \eqref{eq:summary} \\
		$\mathcal{C}_{\theta} = \{ \mathcal{C}_{\theta}(z) \}$ & The uniform confidence band for the aggregated parameter. & \eqref{eq:UCB_summary} \\
		$p_{g,t}(X_i; \pi_{g,t}^*)$ & The generalized propensity score (GPS). & \eqref{eq:GPS-OR} \\
		$m_{i,g,t} = m_{g,t}(X_i; \beta_{g,t}^*)$ & The outcome regression (OR) function. & \eqref{eq:GPS-OR} \\
		$\DR_{g,t}(z)$ & The conditional doubly robust (DR) estimand. & \eqref{eq:DRestimand} \\
		$R_{i,g,t} = R_{g,t}(W_i; \pi_{g,t}^*)$ & The technical variable. & \eqref{eq:Rdef} \\
		$A_{i,g,t}$ & The technical variable. & \eqref{eq:Adef} \\
		$B_{i,g,t}$, $E_{i,g,t}$, and $F_{i,g,t}$ & The technical variables. & \eqref{eq:B_E_F} \\
		$\mathcal{B}_{g,t}(z)$ and $\mathcal{V}_{g,t}(z)$ & The asymptotic bias and variance formulas for CATT. & \eqref{eq:bias_nev_2}, \eqref{eq:var_nev_2} \\
		$a_n$ and $\lambda$ & The technical variables. & \eqref{eq:lambda} \\
		\bottomrule
	\end{tabular}
\end{table}

In the proofs of Theorems \ref{thm:bias_variance} and \ref{thm:approx_us_distribution}, we use the following notation.
For a generic variable $Q$ and a generic integer $\nu \ge 0$, let $\mu_Q^{(\nu)}(z) \coloneqq {\bE}^{(\nu)}[ Q \mid Z = z ]$ denote the $\nu$-th derivative with respect to $z$ of the conditional mean of $Q$ given $Z = z$.
For example, we write $\mu_G^{(\nu)}(z) = {\bE}^{(\nu)}[ G_g \mid Z = z ]$, $\mu_{R}^{(\nu)}(z) = {\bE}^{(\nu)}[ R_{g,t} \mid Z = z ]$, and $\mu_{A}^{(\nu)}(z) = {\bE}^{(\nu)}[ A_{g,t} \mid Z = z ]$.
The $p$-th order LPR estimator of $\mu_Q^{(\nu)}(z)$ is defined as
\begin{align*}
	& \hat{\mu}_Q^{(\nu)}(z;h,K,p) \coloneqq \nu! \bm{e}_\nu^\top \hat{\bm{\beta}}_Q (z;h,K,p), \\
	& \hat{\bm{\beta}}_Q (z;h,K,p) \coloneqq \argmin_{\bm{b} \in \mathbb{R}^{p+1}} \sum_{i=1}^n \left( Q_i - \bm{r}_p(Z_i-z)^\top \bm{b} \right)^2 K \left(\frac{Z_i - z}{h}\right),
\end{align*}
where $p$ is a positive integer such that $p \ge \nu$, $\bm{e}_\nu$ is the $p + 1$ dimensional vector whose $(\nu + 1)$-th element is $1$ and the rest are 0, and $\bm{r}_p(u) \coloneqq ( 1, u, u^2, \dots, u^p )^\top$.
For notational simplicity, we often suppress the dependence on $h$, $K$, and $p$, unless there is confusion.
For example, we often write $\hat{\mu}_Q^{(\nu)}(z) = \hat{\mu}_Q^{(\nu)}(z; h, K, p)$ and $\hat{\bm{\beta}}_Q(z) = \hat{\bm{\beta}}_Q (z; h, K, p)$.

To simplify notation, we write
\begin{align*}
	u_{i,h} \coloneqq \frac{Z_i - z}{h},
	\qquad 
	K_{i,h} \coloneqq K( u_{i,h} ),
	\qquad 
	\mathcal{K}_h \coloneqq \mathrm{diag}[ K_{1,h}, \dots, K_{n,h} ].
\end{align*}
Define the following $n \times (p + 1)$ matrices:
\begin{align*}
	\mathcal{Z}
	\coloneqq \Big[ \bm{r}_p(Z_1 - z), \dots, \bm{r}_p(Z_n - z) \Big]^\top,
	\qquad 
	\check{\mathcal{Z}}
	\coloneqq \Big[ \bm{r}_p( u_{1,h} ), \dots, \bm{r}_p( u_{n,h}) \Big]^\top.
\end{align*}
Note that $\check{\mathcal{Z}} = \mathcal{Z} \bm{H}_p^{-1}$, where $\bm{H}_p \coloneqq \diag[ 1, h, h^2, \dots, h^p ]$.
Let
\begin{align*}
	\tilde{\bm{\Gamma}}_{(h,K,p)} 
	\coloneqq \frac{1}{nh} \check{\mathcal{Z}}^\top \mathcal{K}_h \check{\mathcal{Z}}
	= \frac{1}{nh} \sum_{i=1}^n K_{i,h} \bm{r}_p( u_{i,h} ) \bm{r}_p( u_{i,h} )^\top,
\end{align*}
and 
\begin{align*}
	\tilde{\bm{\Omega}}_{(h,K,p)} 
	\coloneqq \frac{1}{h} \check{\mathcal{Z}}^\top \mathcal{K}_h
	= \frac{1}{h} \Big[ K_{1,h} \bm{r}_p( u_{1,h} ), \dots, K_{n,h} \bm{r}_p( u_{n,h} ) \Big].
\end{align*}
Notice that the dimensions of $\tilde{\bm{\Gamma}}_{(h,K,p)}$ and $\tilde{\bm{\Omega}}_{(h,K,p)}$ are $(p + 1) \times (p + 1)$ and $(p + 1) \times n$, respectively.
Additionally, define the following $(p + 1) \times (p + 1)$ matrix:
\begin{align*}
	\bm{\Gamma}_{(K,p)} 
	\coloneqq f_Z(z) \left( \int K(u) \bm{r}_p(u) \bm{r}_p(u)^\top du \right).
\end{align*}

It is easy to see that $\hat{\bm{\beta}}_Q(z)$ obtained from the $p$-th order LPR estimation can be rewritten as follows.
Denoting $\bm{Q} = [ Q_1, \dots, Q_n ]^\top$, we have
\begin{align*}
	\bm{\hat{\beta}}_Q(z)
	& = ( \mathcal{Z}^\top \mathcal{K}_h \mathcal{Z} )^{-1} \mathcal{Z}^\top \mathcal{K}_h \bm{Q} \\
	& = \left( \left[ \mathcal{Z} \bm{H}_p^{-1} \bm{H}_p \right]^\top \mathcal{K}_h \left[ \mathcal{Z} \bm{H}_p^{-1} \bm{H}_p \right] \right)^{-1} \left[ \mathcal{Z} \bm{H}_p^{-1} \bm{H}_p \right]^\top \mathcal{K}_h \bm{Q} \\ 
	& = \bm{H}_p^{-1} ( \check{\mathcal{Z}}^\top \mathcal{K}_h \check{\mathcal{Z}} )^{-1} \check{\mathcal{Z}}^\top \mathcal{K}_h \bm{Q} \\
	& = \frac{1}{n} \bm{H}_p^{-1} \tilde{\bm{\Gamma}}_{(h,K,p)}^{-1} \tilde{\bm{\Omega}}_{(h,K,p)} \bm{Q}.
\end{align*}

\phantomsection\label{page:AE-23-supp}\Copy{AE-23-supp}{
	Because our bandwidth selection proposed in Section \ref{subsec:bandwidth} relies on the LLR estimator of the DR estimand, we clarify its definition here.
	The first-stage estimation is the same as in the main text: we estimate $\beta_{g,t}^{*}$ and $\pi_{g,t}^*$ via some parametric methods, which leads to $\hat R_{i,g,t} = R_{g,t}(W_i; \hat \pi_{g,t})$ and $\hat m_{i,g,t} = m_{g,t}(X_i; \hat \beta_{g,t})$ for each $i$.
	Next, we compute 
	\begin{align*}
		\hat A_{i,g,t}^{\mathtt{LL},\ny}
		& \coloneqq \left( \frac{ G_{i,g} }{ \hat \mu_G^{\mathtt{LL}}(z) } - \frac{ \hat R_{i,g,t} }{ \hat \mu_{\hat R}^{\mathtt{LL}}(z) } \right) \left( Y_{i,t} - Y_{i,g-1} - \hat m_{i,g,t} \right),
	\end{align*}
	where $\hat \mu_G^{\mathtt{LL}}(z)$ and $\hat \mu_{\hat R}^{\mathtt{LL}}(z)$ are the LLR estimators of $\mu_G(z)$ and $\mu_R(z)$, respectively.
	Specifically, $\hat \mu_{\hat R}^{\mathtt{LL}}(z)$ is defined by
	\begin{align*}
		\hat{\mu}_{\hat R}^{\mathtt{LL}}(z)
		\coloneqq \bm{e}_0^\top \hat{\bm{\beta}}_{\hat R}^{\mathtt{LL}}(z),
		\qquad
		\hat{\bm{\beta}}_{\hat R}^{\mathtt{LL}}(z) \coloneqq \argmin_{\bm{b} \in \mathbb{R}^{2}} \sum_{i=1}^n \left( \hat R_{i,g,t} - \bm{r}_{1}(Z_i - z)^\top \bm{b} \right)^2 K \left( \frac{Z_i - z}{h} \right),
	\end{align*}
	and the definition of $\hat \mu_G^{\mathtt{LL}}(z)$ is analogous.
	Finally, we obtain the LLR estimator $\hat{\DR}_{g,t}^{\mathtt{LL},\ny}(z)$ of the DR estimand as follows:
	\begin{align} \label{eq:DR-LL-ny}
		\begin{split}
			\hat{\DR}_{g,t}^{\mathtt{LL},\ny}(z) 
			& \coloneqq \hat \mu_{\hat A}^{\mathtt{LL}}(z) \coloneqq \bm{e}_0^\top \hat{\bm{\beta}}_{\hat A}^{\mathtt{LL}}(z), \\ 
			\hat{\bm{\beta}}_{\hat A}^{\mathtt{LL}}(z) & \coloneqq \argmin_{\bm{b} \in \mathbb{R}^{2}} \sum_{i=1}^n \left( \hat A_{i,g,t}^{\mathtt{LL},\ny} - \bm{r}_{1}(Z_i - z)^\top \bm{b} \right)^2 K \left( \frac{Z_i - z}{h} \right).
		\end{split}
	\end{align}
	Hereafter, for notational simplicity, we often suppress the superscripts ``${\mathtt{LL}}$'' and ``$\ny$'' unless there is confusion.
}

In the following analysis, unless otherwise noted, the convergence rates of the remainder terms hold uniformly in $(g,t,z) \in \mathcal{A}$.
We write $a \lesssim b$ if there exists a constant $C > 0$ such that $a \le C b$.
In addition, with an abuse of notation, we often write $o_{\bP}( \sqrt{ (\log n) / (nh) } )$ to indicate $O_{\bP}( \sqrt{ (\log n) / (n^{1 + \varepsilon} h) } )$ for some $\varepsilon > 0$ to simplify notation.

In Appendices \ref{sec:proof}, \ref{sec:lemma}, \ref{sec:not-yet-treated}, and \ref{sec:never-treated}, we add the subscripts or superscripts ``$\ny$'' and ``$\nev$'' to the quantities specific to the not-yet-treated group and the never-treated group, respectively.
This notation is useful in making clear the difference between the analyses that use the not-yet-treated group and the never-treated group as the comparison groups.
For example, in those appendices, we write the conditional DR estimand defined in \eqref{eq:DRestimand} as $\DR_{g,t}^{\ny}(Z)$ to emphasize that this estimand uses the not-yet-treated group as the comparison group.
In contrast, we suppress such subscripts and superscripts in Appendices \ref{sec:supp-summary}, \ref{sec:pretrends}, \ref{sec:simulation}, and \ref{sec:empirical2} for ease of exposition.

\section{Proofs} \label{sec:proof}

\subsection{Proof of Lemma \ref{lem:DR}}

We provide the proof of result (i) only, as result (ii) can be proved analogously.
Given Lemma \ref{lem:nonparametric}, the proof is completed if we show
\begin{align*}
	\bE \left[ \frac{ R_{g,t}(W; \pi_{g,t}) }{\bE [ R_{g,t}(W; \pi_{g,t}) \mid Z ]} \left( Y_t - Y_{g-1} - m_{g,t}^{\ny}(X; \beta_{g,t}^{\ny*}) \right) \; \middle| \; Z \right]
	= 0
	\quad 
	\text{for arbitrary $\pi_{g,t} \in \Pi_{g,t}$}.
\end{align*}
Recall that
\begin{align*}
	& R_{g,t}(W; \pi_{g,t})
	\coloneqq \frac{p_{g,t}(X; \pi_{g,t}) (1 - D_{t}) (1 - G_g)}{1 - p_{g,t}(X; \pi_{g,t})}, \\
	& m_{g,t}^{\ny}(X) 
	\coloneqq \bE[ Y_t - Y_{g-1} \mid X, D_{t} = 0, G_g = 0 ].
\end{align*}

From the law of iterated expectations, it suffices to show $\bE [ (1 - D_{t}) (1 - G_g) ( Y_t - Y_{g-1} - m_{g,t}^{\ny}(X; \beta_{g,t}^{\ny*}) ) \mid X ] = 0$ a.s.
Using Assumption \ref{as:parametric-ny}(i), we can observe that
\begin{align*}
	& \bE[ (1 - D_{t}) (1 - G_g) \left( Y_t - Y_{g-1} - m_{g,t}^{\ny}(X; \beta_{g,t}^{\ny*}) \right) \mid X ] \\
	& = \bE[ \bE[ (1 - D_{t}) (1 - G_g) \left( Y_t - Y_{g-1} - m_{g,t}^{\ny}(X; \beta_{g,t}^{\ny*}) \right) \mid X, D_t, G_g] \mid X ] \\
	& = \bE[(1 - D_{t}) (1 - G_g) \mid X ] \cdot m_{g,t}^{\ny}(X; \beta_{g,t}^{\ny*}) - \bE[(1 - D_{t}) (1 - G_g) \mid X ] \cdot m_{g,t}^{\ny}(X; \beta_{g,t}^{\ny*}) \\
	& = 0
	\quad \text{a.s.}
\end{align*}
\qed 

\subsection{Proof of Theorem \ref{thm:bias_variance}}

Observe that
\begin{align} \label{eq:linear_srbc1}
	\begin{split}
		\hat{\DR}_{g,t}^{\ny}(z) - {\DR}_{g,t}^{\ny}(z) 
		& = \hat{\mu}_{\hat{A}}(z) - \mu_A(z) \\
		& = \Big( \hat \mu_A(z) - \hat \mu_{\bE[A\mid Z]}(z) \Big) + \Big( \hat \mu_{\bE[A\mid Z]}(z) - \mu_A(z) \Big) + \Big( \hat \mu_{\hat A}(z) - \hat \mu_A(z) \Big) \\
		& \eqqcolon (\mathbf{LQR.I}) + (\mathbf{LQR.II}) + (\mathbf{LQR.III}),
	\end{split}
\end{align}
where each $\hat \mu$ in the right-hand side denotes the corresponding LQR estimator.
For example, $\hat \mu_{\bE[A \mid Z]}(z)$ denotes the LQR estimator for the conditional mean of $\bE[ A_{g,t}^{\ny} \mid Z ]$ at $Z = z$, that is,
\begin{align*}
	\hat{\mu}_{\bE[A \mid Z]}(z) \coloneqq \bm{e}_0^\top \hat{\bm{\beta}}_{\bE[A \mid Z]} (z)
	\quad \text{with} \quad
	\hat{\bm{\beta}}_{\bE[A \mid Z]} (z) \coloneqq \argmin_{\bm{b} \in \mathbb{R}^3} \sum_{i=1}^n \left( \bE[A_{i,g,t}^{\ny} \mid Z_i] - \bm{r}_2(Z_i - z)^\top \bm{b} \right)^2 K_{i,h}.
\end{align*}
Here, $(\mathbf{LQR.I})$ and $(\mathbf{LQR.II})$ correspond to the variance and bias terms, respectively, for the LQR estimation where the dependent variable is the ``true'' $A_{i,g,t}^{\ny}$, and $(\mathbf{LQR.III})$ accounts for the estimation error caused by the first- and second-stage estimation.

\subsubsection{Proof of the asymptotic linear representation} \label{subsubsec:proof-linear}

To proceed, we provide the following auxiliary results that we use repeatedly.
First, by Lemma \ref{lem:Gamma}, we can see that
\begin{align*}
	\tilde{\bm{\Gamma}}_{(h,K,2)} 
	& = \bm{\Gamma}_{(K,2)} + o_{\bP}(1) \\
	& =  f_Z(z) \left( \int K(u) \bm{r}_2(u) \bm{r}_2(u)^\top du \right) + o_{\bP}(1) \\
	& = f_Z(z)
	\begin{bmatrix}
		1 & 0 & I_{2,K} \\
		0 & I_{2,K} & 0 \\
		I_{2,K} & 0 & I_{4,K} 
	\end{bmatrix} + o_{\bP}(1).
\end{align*}
In conjunction with the continuous mapping theorem, this implies that
\begin{align*}
	\bm{e}_0^\top\tilde{\bm{\Gamma}}_{(h,K,2)}^{-1} 
	& = \bm{e}_0^\top \bm{\Gamma}_{(K,2)}^{-1} + o_{\bP}(1) \\
	& = \frac{1}{f_Z(z)} \frac{1}{I_{4,K} - I_{2,K}^2} \left( I_{4,K}, \; 0, \; I_{2,K} \right) + o_{\bP}(1).
\end{align*}
Thus, we have
\begin{align} \label{eq:e0_srbc}
	\bm{e}_0^\top \bm{\Gamma}_{(K,2)}^{-1} \bm{r}_2(u_{i,h})
	= \frac{1}{f_Z(z)} \left( \frac{ I_{4,K} - u_{i,h}^2 I_{2,K} }{ I_{4,K} - I_{2,K}^2 } \right)
	= \frac{1}{f_Z(z)} \Psi_{i,h}.
\end{align}

\paragraph{Evaluation on $(\mathbf{LQR.I})$.}
To evaluate the first term in \eqref{eq:linear_srbc1}, denoting $\bm{A}_{g,t}^{\ny} = (A_{1,g,t}^{\ny}, \dots, A_{n,g,t}^{\ny})^\top$, observe that
\begin{align*}
	\hat \mu_A(z) - \hat \mu_{\bE[A \mid Z]}(z) 
	&= \frac{1}{n} \bm{e}_0^\top \tilde{\bm{\Gamma}}_{(h,K,2)}^{-1} \tilde{\bm{\Omega}}_{(h,K,2)} \left( \bm{A}_{g,t}^{\ny} - \bE\left[ \bm{A}_{g,t}^{\ny} \; \middle| \; \bm{Z} \right] \right) \\
	&= \bm{e}_0^\top \tilde{\bm{\Gamma}}_{(h,K,2)}^{-1} \frac{1}{nh} \sum_{i=1}^n K_{i,h} \bm{r}_2(u_{i,h}) \left(  A_{i,g,t}^{\ny} - \mu_A(Z_i) \right) \\
	& = \left( \bm{e}_0^\top\bm{\Gamma}_{(K,2)}^{-1} + o_{\bP}(1) \right) \frac{1}{nh} \sum_{i=1}^n K_{i,h} \bm{r}_2(u_{i,h}) \left(  A_{i,g,t}^{\ny} - \mu_A(Z_i) \right) \\
	&  =  \frac{1}{nh} \frac{1}{f_Z(z)} \sum_{i=1}^n K_{i,h} \Psi_{i,h} (A_{i,g,t}^{\ny} - \mu_A(Z_i)) + o_{\bP}\left( \frac{1}{nh} \sum_{i=1}^n K_{i,h} \bm{r}_2(u_{i,h}) (A_{i,g,t}^{\ny} - \mu_A(Z_i)) \right),
\end{align*}
where the last equality follows from (\ref{eq:e0_srbc}).
From Lemma \ref{lem:eval_ka-mu}, which states that for any non-negative integer $q\in\mathbb{Z}_+$,
\begin{align*}
	\frac{1}{nh} \sum_{i=1}^n K_{i,h} (Z_i-z)^q \left(  A_{i,g,t}^{\ny} - \mu_A(Z_i) \right) = O_{\bP}\left( h^q \sqrt{\frac{\log n}{nh}}\right) 
\end{align*}
holds uniformly in $(g,t,z) \in \mathcal{A}^{\ny}$, we can see that
\begin{align*}
	\hat \mu_A(z) - \hat \mu_{\bE[A\mid Z]}(z) 
	= \frac{1}{nh} \frac{1}{f_Z(z)} \sum_{i=1}^n K_{i,h} \Psi_{i,h} (A_{i,g,t}^{\ny} - \mu_A(Z_i)) + o_{\bP}\left( \sqrt{\frac{\log n}{nh}} \right)
\end{align*}
holds uniformly in $(g,t,z) \in \mathcal{A}^{\ny}$.

\paragraph{Evaluation on $(\mathbf{LQR.II})$.}
To evaluate the second term in (\ref{eq:linear_srbc1}), with an abuse of notation, we denote $\bm{\mu}_{A,(h)} = (\mu_A(z), h \cdot \mu_A^{(1)}(z), h^2 \cdot \mu_A^{(2)}(z) / 2)^\top$.
Observe that
\begin{small}
	\begin{align*}
		& \hat \mu_{\bE[A\mid Z]}(z) - \mu_A(z) \\
		&= \frac{1}{n} \bm{e}_0^\top \tilde{\bm{\Gamma}}_{(h,K,2)}^{-1} \tilde{\bm{\Omega}}_{(h,K,2)}  \bE\left[ \bm{A}_{g,t}^{\ny} \; \middle| \; \bm{Z} \right] -  \bm{e}_0^\top \tilde{\bm{\Gamma}}_{(h,K,2)}^{-1}\tilde{\bm{\Gamma}}_{(h,K,2)}\bm{\mu}_{A,(h)}\\
		& = \bm{e}_0^\top \tilde{\bm{\Gamma}}_{(h,K,2)}^{-1}  \left( \frac{1}{nh} \sum_{i=1}^n K_{i,h} \bm{r}_2(u_{i,h})\left[ \mu_A(Z_i) - \mu_A(z) - \mu_A^{(1)}(z)(Z_i-z) - \frac{1}{2}\mu_A^{(2)}(z)(Z_i-z)^2 \right]\right) \\
		& = \left( \bm{e}_0^\top\bm{\Gamma}_{(K,2)}^{-1} + o_{\bP}(1) \right)  \left( \frac{1}{nh} \sum_{i=1}^n K_{i,h} \bm{r}_2(u_{i,h})\left[ \frac{1}{3!}\mu_A^{(3)}(z)(Z_i-z)^3 + \frac{1}{4!}\mu^{(4)}(Z_i)(Z_i-z)^4 + \frac{1}{5!}\mu^{(5)}(\Tilde{Z}_i)(Z_i-z)^5\right]\right) \\
		& = \frac{1}{nh} \sum_{i=1}^n K_{i,h} \Psi_{i,h} \left(  \frac{1}{3!}\mu_A^{(3)}(z)(Z_i-z)^3 + \frac{1}{4!}\mu^{(4)}(Z_i)(Z_i-z)^4 \right) \\
		&  \quad + o_{\bP}\left( \frac{1}{nh} \sum_{i=1}^n K_{i,h} \bm{r}_2(u_{i,h}) \left(  \frac{1}{3!}\mu_A^{(3)}(z)(Z_i-z)^3 + \frac{1}{4!}\mu^{(4)}(Z_i)(Z_i-z)^4 \right) \right) \\
		& = \frac{1}{24 f_Z(z)}\left(2\mu_A^{(3)}(z)f_Z^{(1)}(z) + \mu_A^{(4)}(z)f_Z(z)\right) h^4 \left( \frac{I_{4,K}^2 - I_{2,K}I_{6,K}}{I_{4,K} - I_{2,K}^2} \right) + o_{\bP}(h^4) + O_{\bP}\left( \sqrt{\frac{\log n}{nh} } \right),
	\end{align*}
\end{small}
holds uniformly in $(g,t,z) \in \mathcal{A}^{\ny}$, where the third equality follows from the Taylor expansion of $\mu_A(Z_i)$ around $Z_i=z$ and $\Tilde{Z}_i$ lies between $Z_i$ and $z$, and
the leading term can be derived from Lemma \ref{lem:eval_bckq}, which states that
\begin{align*}
	& \frac{1}{nh} \sum_{i=1}^n K_{i,h} \Psi_{i,h} (Z_i-z)^3 = f_Z^{(1)}(z) h^4 \left( \frac{I_{4,K}^2 - I_{2,K}I_{6,K}}{I_{4,K} - I_{2,K}^2} \right) + o(h^4) + O_{\bP}\left( \sqrt{\frac{\log n}{nh} } \right), \\
	& \frac{1}{nh} \sum_{i=1}^n K_{i,h} \Psi_{i,h} (Z_i-z)^4 = f_Z(z)h^4\left( \frac{I_{4,K}^2 - I_{2,K}I_{6,K}}{I_{4,K} - I_{2,K}^2} \right) + o(h^4) + O_{\bP}\left( \sqrt{\frac{\log n}{nh} } \right),
\end{align*}
holds uniformly in $(g,t,z) \in \mathcal{A}^{\ny}$, and the convergence rates of the remainder terms follow from Lemma \ref{lem:eval_kq}, which states that the following equations hold uniformly in $z \in \mathcal{I}$:
\begin{align*}
	\frac{1}{nh} \sum_{i=1}^n K_{i,h} (Z_i-z)^q 
	= 
	\begin{cases}
		\displaystyle O(h^4) + O_{\bP}\left( h^q \sqrt{\frac{\log n}{nh} } \right) & \text{for $q  = 3, 4$} \\
		\displaystyle o(h^4) + O_{\bP}\left( h^q \sqrt{\frac{\log n}{nh} } \right) & \text{for $q \ge 5$}.
	\end{cases}
\end{align*}

\paragraph{Evaluation on $(\mathbf{LQR.III})$.}
To examine the third term in \eqref{eq:linear_srbc1}, simple algebra can expand $\hat A_{i,g,t}^{\ny} - A_{i,g,t}^{\ny}$ as follows:
\begin{align*}
	\hat A_{i,g,t}^{\ny} - A_{i,g,t}^{\ny}
	& = \left( \frac{ \hat F_{i,g,t}^{\ny} }{ \hat{\mu}_G(z) } - \frac{ \hat E_{i,g,t}^{\ny} }{\hat{\mu}_{\hat{R}}(z)} \right) - \left( \frac{ F_{i,g,t}^{\ny} }{ \mu_G(z) } - \frac{ E_{i,g,t}^{\ny} }{\mu_R(z)} \right) \\
	& =  \frac{ E_{i,g,t}^{\ny} }{ \hat \mu_{\hat R}(z) \mu_R(z) } \Big( \hat \mu_{\hat R}(z) - \mu_R(z) \Big) - \frac{ F_{i,g,t}^{\ny} }{ \hat \mu_{G}(z) \mu_G(z) } \Big( \hat \mu_{G}(z) - \mu_G(z) \Big) \\
	& \quad + \left( \frac{ 1 }{ \hat \mu_G(z) } \left( \hat F_{i,g,t}^{\ny} - F_{i,g,t}^{\ny} \right) - \frac{ 1 }{ \hat \mu_{\hat R}(z) } \left( \hat E_{i,g,t}^{\ny} - E_{i,g,t}^{\ny} \right) \right).
\end{align*}
Then, using \eqref{eq:e0_srbc}, the third term in \eqref{eq:linear_srbc1} can be rewritten as
\begin{align*}
	& \hat \mu_{\hat A}(z) - \hat \mu_A(z) \\
	& = \frac{1}{\hat \mu_{\hat R}(z) \mu_R(z)} \bm{e}_0^\top \tilde{\bm{\Gamma}}_{(h,K,2)}^{-1} \left( \frac{1}{nh} \sum_{i=1}^n K_{i,h}  \bm{r}_2(u_{i,h}) E_{i,g,t}^{\ny} \right) \left( \hat \mu_{\hat R}(z) - \mu_R(z) \right) \\
	& \quad - \frac{1}{\hat \mu_G(z) \mu_G(z)} \bm{e}_0^\top \tilde{\bm{\Gamma}}_{(h,K,2)}^{-1} \left( \frac{1}{nh} \sum_{i=1}^n K_{i,h} \bm{r}_2(u_{i,h}) F_{i,g,t}^{\ny} \right) \left( \hat \mu_G(z) - \mu_G(z) \right) \\
	& \quad + \bm{e}_0^\top \tilde{\bm{\Gamma}}_{(h,K,2)}^{-1} \left( \frac{1}{nh} \sum_{i=1}^n K_{i,h} \bm{r}_2(u_{i,h}) \left( \frac{ 1 }{ \hat \mu_G(z) } \left( \hat F_{i,g,t}^{\ny} - F_{i,g,t}^{\ny} \right) - \frac{ 1 }{ \hat \mu_{\hat R}(z) } \left( \hat E_{i,g,t}^{\ny} - E_{i,g,t}^{\ny} \right) \right) \right) \\
	& \eqqcolon (\mathbf{LQR.III\en I}) - (\mathbf{LQR.III\en II}) + (\mathbf{LQR.III\en III}).
\end{align*}
In what follows, we examine each term.

For $(\mathbf{LQR.III\en  I})$, we have 
\begin{align*}
	(\mathbf{LQR.III \en I})
	& = \frac{1}{\hat \mu_{\hat R}(z) \mu_R(z)} \bm{e}_0^\top \tilde{\bm{\Gamma}}_{(h,K,2)}^{-1} \left( \frac{1}{nh} \sum_{i=1}^n K_{i,h} \bm{r}_2(u_{i,h}) E_{i,g,t}^{\ny} \right) \left( \hat \mu_{\hat R}(z) - \mu_R(z) \right) \\
	& = \frac{1}{\hat \mu_{\hat R}(z) \mu_R(z)} \left( \frac{1}{n} \bm{e}_0^\top \tilde{\bm{\Gamma}}_{(h,K,2)}^{-1} \tilde{\bm{\Omega}}_{(h,K,2)} \bm{E}_{g,t}^{\ny} \right) \left( \hat \mu_{\hat R}(z) - \mu_R(z) \right) \\
	& = \frac{1}{\hat \mu_{\hat R}(z) \mu_R(z)} \hat \mu_E(z) \left( \hat \mu_{\hat R}(z) - \mu_R(z) \right) \\
	& = \frac{\mu_E(z)}{\mu_R^2(z)} \left( \hat \mu_R(z) - \mu_R(z) \right) + o_{\bP}(h^4) + o_{\bP}\left( \sqrt{\frac{\log n}{nh}} \right),
\end{align*}
where we used Lemma \ref{lem:eval_mu_hatr}, which states that 
\begin{align*}
	\hat \mu_{\hat R}(z) = \hat \mu_R(z) + O_{\bP}\left( \frac{1}{\sqrt{n}} \right),
\end{align*}
and Lemma \ref{lem:eval_mu_bce}, which states that
\begin{align*}
	& \hat \mu_R(z) - \mu_R(z) = O(h^4) + O_{\bP}\left( \sqrt{\frac{\log n}{nh}} \right), \quad
	\hat \mu_E(z) - \mu_E(z) = O(h^4) + O_{\bP}\left( \sqrt{\frac{\log n}{nh}} \right),
\end{align*}
holds uniformly in $(g,t,z) \in \mathcal{A}^{\ny}$.
Furthermore, in the same manner as the evaluation of $\hat \mu_A(z) - \hat \mu_{\bE[ A \mid Z ]}(z)$ and $\hat \mu_{\bE[ A \mid Z ]}(z) - \mu_A(z)$, we can show that
\begin{align*}
	& \hat \mu_R(z) - \mu_R(z) \\
	& = \frac{1}{f_Z(z)} \frac{1}{nh} \sum_{i=1}^n K_{i,h} \Psi_{i,h} R_{i,g,t} \\
	&  \quad + \frac{1}{24 f_Z(z)}\left(2\mu_R^{(3)}(z)f_Z^{(1)}(z) + \mu_R^{(4)}(z)f_Z(z)\right) \left( \frac{I_{4,K}^2 - I_{2,K}I_{6,K}}{I_{4,K} - I_{2,K}^2} \right) h^4 + o_{\bP}(h^4) + o_{\bP}\left( \sqrt{\frac{\log n}{nh}} \right),
\end{align*}
uniformly in $(g,t,z) \in \mathcal{A}^{\ny}$.
Thus, we have shown that
\begin{align*}
	(\mathbf{LQR.III \en I})
	& = \frac{ \mu_E(z) }{ f_Z(z) \mu_R^2(z) } \frac{1}{nh} \sum_{i=1}^n K_{i,h} \Psi_{i,h} (R_{i,g,t} - \mu_R(Z_i)) \\
	& \quad + \frac{ \mu_E(z) }{ 24 f_Z(z) \mu_R^2(z) } \left(2\mu_R^{(3)}(z)f_Z^{(1)}(z) + \mu_R^{(4)}(z)f_Z(z)\right) \left( \frac{I_{4,K}^2 - I_{2,K}I_{6,K}}{I_{4,K} - I_{2,K}^2} \right) h^4 \\
	& \quad + o_{\bP}(h^4) + o_{\bP}\left( \sqrt{\frac{\log n}{nh}} \right),
\end{align*}
uniformly in $(g,t,z) \in \mathcal{A}^{\ny}$.

Similar to the case of $(\mathbf{LQR.III \en I})$, we can show that
\begin{align*}
	(\mathbf{LQR.III\en II})
	& = \frac{ \mu_F(z) }{ f_Z(z) \mu_G^2(z) } \frac{1}{nh} \sum_{i=1}^n K_{i,h} \Psi_{i,h} (G_{i,g} - \mu_G(Z_i)) \\
	& \quad + \frac{
		\mu_F(z)
	}{ 24 f_Z(z) \mu^2_G(z)} \left(2\mu_G^{(3)}(z)f_Z^{(1)}(z) + \mu_G^{(4)}(z)f_Z(z)\right) \left( \frac{I_{4,K}^2 - I_{2,K}I_{6,K}}{I_{4,K} - I_{2,K}^2} \right) h^4 \\
	& \quad + o_{\bP}(h^4) + o_{\bP}\left( \sqrt{\frac{\log n}{nh}} \right),
\end{align*}
uniformly in $(g,t,z) \in \mathcal{A}^{\ny}$.

For $(\mathbf{LQR.III\en III})$, the $1/\sqrt{n}$-consistency of the first-stage parametric estimation ensures that
\begin{align*}
	(\mathbf{LQR.III \en III})
	& = \bm{e}_0^\top \tilde{\bm{\Gamma}}_{(h,K,2)}^{-1} \left( \frac{1}{nh} \sum_{i=1}^n K_{i,h} \bm{r}_2(u_{i,h}) \left( \frac{ 1 }{ \hat \mu_G(z) } \left( \hat F_{i,g,t}^{\ny} - F_{i,g,t}^{\ny} \right) - \frac{ 1 }{ \hat \mu_{\hat R}(z) } \left( \hat E_{i,g,t}^{\ny} - E_{i,g,t}^{\ny} \right) \right) \right) \\
	& = O_{\bP}\left( \frac{1}{\sqrt{n}} \right),
\end{align*}
uniformly in $(g,t,z) \in \mathcal{A}^{\ny}$, 
where we used \eqref{eq:e0_srbc}, Assumption \ref{as:asymptotic1}(v), and Lemma \ref{lem:eval_kq}, which states that
\begin{align*}
	\frac{1}{nh} \sum_{i=1}^n K_{i,h}
	= O_{\bP}(1),
	\quad 
	\frac{1}{nh} \sum_{i=1}^n K_{i,h}(Z_i-z)
	= o_{\bP}(1), 
	\quad 
	\frac{1}{nh} \sum_{i=1}^n K_{i,h}(Z_i-z)^2
	= o_{\bP}(1),
\end{align*}
holds uniformly in $z \in \mathcal{I}$.

Summing up, we obtain the desired result.
\qed

\subsubsection{Proof of the asymptotic bias}
Using the evaluation on $(\mathbf{LQR.I})$, $(\mathbf{LQR.II})$, and $(\mathbf{LQR.III})$ in the previous subsection, the asymptotic linear representation in Theorem \ref{thm:bias_variance} can be reformulated as
\begin{align*}
	\hat{\DR}_{g,t}^{\ny}(z) - {\DR}_{g,t}^{\ny}(z)
	& = \frac{1}{f_Z(z)} \frac{1}{nh} \sum_{i=1}^n  \Psi_{i,h}(B_{i,g,t}^{\ny} - \mu_B(Z_i)) K_{i,h} \\
	& \quad + \frac{1}{24 f_Z(z)}\left(2\mu_B^{(3)}(z)f_Z^{(1)}(z) + \mu_B^{(4)}(z)f_Z(z)\right)\left( \frac{I_{4,K}^2 - I_{2,K}I_{6,K}}{I_{4,K} - I_{2,K}^2} \right) h^4 \\
	& \quad + o_{\bP}(h^4) + o_{\bP}\left( \sqrt{\frac{\log n}{nh}} \right).
\end{align*}
Since the expectation of the first term in the right-hand side is $0$, it is easy to see that the asymptotic bias is given by
\begin{align*}
	\bE\left[ \hat{\DR}_{g,t}^{\ny}(z) \;  \middle| \; \bm{Z} \right] - {\DR}_{g,t}^{\ny}(z)
	= h^4 \mathcal{B}_{g,t}^{\ny}(z) + o_{\bP}(h^4) + o_{\bP}\left( \sqrt{\frac{\log n}{nh}} \right).
\end{align*}
\qed

\subsubsection{Proof of the asymptotic variance}
Let $\Psi(u) \coloneqq ( I_{4,K} - u^2 I_{2,K} ) / ( I_{4,K} - I_{2,K}^2 )$.
Given the IID assumption in Assumption \ref{as:iid}, observe that
\begin{align*}
	\Var \left[ \frac{1}{f_Z(z)} \frac{1}{nh} \sum_{i=1}^n K_{i,h} \Psi_{i,h} B_{i,g,t}^{\ny} \; \middle| \; \bm{Z} \right]
	& = \frac{1}{f_Z^2(z)} \frac{1}{n^2 h^2} \sum_{i=1}^n K_{i,h}^2  \Psi_{i,h}^2 \sigma_B^2(Z_i) \\
	& = \frac{1}{nh} \frac{ \sigma_B^2(z) }{ f_Z(z) } \left( \int K^2(u)\Psi^2(u) du \right) + o_{\bP}\left( \frac{1}{nh} \right), 
\end{align*}
where the last line follows from that
\begin{align*}
	\bE\left[ \frac{1}{nh} \sum_{i=1}^n K_{i,h}^2 \sigma_B^2(Z_i) \right]
	& = \int K^2(u) \sigma_B^2(z + uh) f_Z(z + uh) du \\
	& = \sigma_B^2(z) f_Z(z) \left( \int K^2(u)\Psi^2(u) du \right) + o(1),
\end{align*}
and that
\begin{align*}
	\Var\left[ \frac{1}{nh} \sum_{i=1}^n K_{i,h}^2 \sigma_B^2(Z_i) \right]
	& \lesssim \frac{1}{n h^2} \bE[ K_{i,h}^4 \sigma_B^4(Z_i) ] \\
	& = \frac{1}{nh} \sigma_B^4(z) f_Z(z) \left( \int K^4(u)\Psi^4(u) du \right) + o\left( \frac{1}{nh} \right) \\
	& = o(1).
\end{align*}
This completes the proof.
\qed

\subsection{Proof of Theorem \ref{thm:approx_us_distribution}} \label{sec:proof:thm:approx_us_distribution}

Given Theorem \ref{thm:bias_variance}, Theorem \ref{thm:approx_us_distribution} follows from the same proof steps as Theorem 2 of \citet{lee2017doubly}.
Specifically, the proof consists of three steps: 
(i) we approximate the supremum of a linearized process by the supremum of a Gaussian process based on Proposition 3.2 of \citet{chernozhukov2014gaussian};
(ii) we prove certain approximation between the supremum of the Gaussian process and that of a stationary Gaussian process with a feasible covariance function, in a similar way to Lemma 3.4 of \cite{ghosal2000testing}; and 
(iii) we obtain the desired asymptotic expansion for the distribution function of the supremum of the stationary Gaussian process by using Theorem 14.3 of \cite{piterbarg1996asymptotic} and Theorem 3.2 of \cite{konakov1984convergence}.
The first and second steps are based on Lemmas \ref{lem:Approx_us_EP_GP} and \ref{lem:approx_us_GP_sGF} introduced below, whose proofs can be found in the supplementary material of \citet{lee2017doubly}.

\paragraph{Step (i).}
For each $(g, t)$ such that $(g, t, z) \in \mathcal{A}^{\ny}$ for $z \in \mathcal{I}$, we construct an empirical process based on the asymptotic linear representation in Theorem \ref{thm:bias_variance}. 
Define
\begin{align*}
	& \mathfrak{T}_{g,t}^{\ny}(z) 
	\coloneqq \frac{1}{ nh } \sum_{i=1}^n \Psi_{i,h} \big( B_{i,g,t}^{\ny} - \mu_B(Z_i) \big) K \left( \frac{Z_i - z}{h} \right), \\
	& \mathfrak{C}_{g,t}^{\ny}(z)
	\coloneqq \left\{  \frac{1}{h} \bE \left[ \left\{ \Psi_{i,h} \left( B_{i,g,t}^{\ny} - \mu_B(Z_i) \right) K \left( \frac{Z_i - z}{h} \right) \right\}^2 \right] \right\}^{-1/2}.
\end{align*}
The supremum of its normalized version is given by
\begin{align*}
	\mathfrak{W}_n 
	\coloneqq \sqrt{nh} \cdot \sup_{z \in \mathcal{I}} \left\{ \mathfrak{C}_{g,t}^{\ny}(z) \left( \mathfrak{T}_{g,t}^{\ny}(z) - \bE \left[ \mathfrak{T}_{g,t}^{\ny}(z) \right] \right) \right\}.
\end{align*}
For each $n \ge 1$, let $\tilde{\mathfrak{B}}_{n,1}$ denote a centered Gaussian process indexed by $\mathcal{I}$ with the following covariance function:
\begin{align}
	\bE \left[ \tilde{\mathfrak{B}}_{n,1}(z) \tilde{\mathfrak{B}}_{n,1}(z^\prime) \right] 
	= \frac{1}{h} \mathfrak{C}_{g,t}^{\ny}(z) \mathfrak{C}_{g,t}^{\ny}(z^\prime) \Cov \left[ u_{i,g,t}^{\ny} K \left( \frac{Z_i - z}{h} \right) K \left( \frac{Z_i - z^\prime}{h} \right) \right]. \label{eq:cv_us_covariance_1}
\end{align}
Then, using Proposition 3.2 of \citet{chernozhukov2014gaussian}, we can approximate the supremum of the linearized process, $\mathfrak{W}_n$, by the supremum of the Gaussian process $\tilde{\mathfrak{B}}_{n,1}$, as follows.

\begin{lemma} \label{lem:Approx_us_EP_GP}
	Suppose that Assumptions \ref{as:asymptotic1} and \ref{as:asymptotic2} hold.
	Then, for every $n \ge 1$, there is a tight Gaussian random variable $\tilde{\mathfrak{B}}_{n,1}$ in $\ell^{\infty}(\mathcal{I})$ with mean zero and the covariance function given in (\ref{eq:cv_us_covariance_1}), and there is a sequence $\tilde{\mathfrak{W}}_{n,1}$ of random variables such that $\tilde{\mathfrak{W}}_{n,1} \overset{d}{=} \sup_{ z \in \mathcal{I} } \tilde{\mathfrak{B}}_{n,1}(z)$ and, as $n \rightarrow \infty$,
	\begin{align*}
		|\mathfrak{W}_n - \tilde{\mathfrak{W}}_{n,1}| = O_{\bP} \Big\{ (nh)^{-1/6} (\log n) + (nh)^{-1/4} (\log n)^{5/4} + (n^{1/2}h)^{-1/2} (\log n)^{3/2} \Big\}.
	\end{align*}
\end{lemma}

\paragraph{Step (ii).}

We approximate the supremum of the Gaussian process, $\tilde {\mathfrak{W}}_{n,1}$, by the supremum of a homogeneous Gaussian field with zero mean and a feasible covariance function that satisfies certain properties, as follows.

\begin{lemma} \label{lem:approx_us_GP_sGF}
	Suppose that Assumptions \ref{as:asymptotic1} and \ref{as:asymptotic2} hold.
	Then, for every $n \ge 1$ and for $s,s^\prime \in \mathcal{I}_n \coloneq h^{-1}\mathcal{I}$, there is a tight Gaussian variable $\tilde{\mathfrak{B}}_{n,2}$ in $\ell^\infty(\mathcal{I}_n)$ with mean zero and the covariance function 
	\begin{align*}
		\bE \left[ \tilde{\mathfrak{B}}_{n,2}(s) \tilde{\mathfrak{B}}_{n,2}(s^\prime) \right] 
		= \rho(s-s^\prime),
	\end{align*}
	with $\rho(s) \coloneqq [ \int K(u) K(u-s) du ] / [ \int K^2(u) du ]$,
	and there is a sequence of random variables such that $\tilde{\mathfrak{W}}_{n,2} \overset{d}{=} \sup_{z\in\mathcal{I}} \tilde{\mathfrak{B}}_{n,2}(h^{-1}s)$ and, as $n\rightarrow\infty$, 
	\begin{align*}
		|\tilde{\mathfrak{W}}_{n,1} - \tilde{\mathfrak{W}}_{n,2}| 
		= O_{\bP}\left( h \sqrt{\log (h^{-1})} \right).
	\end{align*}
\end{lemma}

\paragraph{Step (iii).}

Noting that the set of $(g, t)$ satisfying $(g, t, z) \in \mathcal{A}^{\ny}$ is finite, we obtain the desired result by Theorem 14.3 of \cite{piterbarg1996asymptotic} and Theorem 3.2 of \cite{konakov1984convergence} in the same manner as the proof of Theorem 2 of \citet{lee2017doubly}.
\qed

\subsection{Proof of Theorem \ref{thm:mb_valid}} \label{sec:proof_mb_valid}

Given Theorem \ref{thm:mb_linearize}, we can prove the validity of the weighted bootstrap inference by verifying the conditions of Theorem 3.2 of \citet{chernozhukov2014anti}.
Let $K_{B,h}(W_i,g,t,z) \coloneqq \Psi_{i,h} B_{i,g,t}^{\ny} K_{i,h}$.
Define the function class $\mathcal{K}_T$ and the standardized process $\mathbb{G}_n k_T$ as follows:
\begin{align*}
	& \mathcal{K}_T \coloneqq \left\{ s \to \frac{1}{h}\frac{K_{B,h}(s,g,t,z)}{ \tilde{\mathcal{S}}_{g,t}^{\ny}(z)}: (g,t,z,h) \in \mathcal{A}^{\ny} \times \mathscr{H}_n \right\}, \\
	& \mathbb{G}_n k_T \coloneqq \frac{1}{\sqrt{n}} \sum_{i=1}^n \frac{1}{h} \frac{\Psi_{i,h} \big( B_{i,g,t}^{\ny} - \mu_B(Z_i) \big) K_{i,h}}{ \tilde{\mathcal{S}}_{g,t}^{\ny}(z)},
\end{align*}
where $\mathscr{H}_n$ is the set of bandwidths that satisfy Assumption \ref{as:asymptotic1}(iii).
Note that by Assumption \ref{as:asymptotic2} the difference between $\mathbb{G}_n k_T$ and the ($f_Z(z)$ times) studentized statistic vanishes at a polynomial rate.
By Lemma \ref{lem:K_vc_type} and Corollary A.1(i) of \citet{chernozhukov2014gaussian}, $\mathcal{K}_T$ is VC type.
By the definition of VC type classes, for some constants $a \ge e$ and $v \ge 1$, we have
\begin{align*}
	\sup_Q N( \mathcal{K}_T, \| \cdot \|_{Q,2}, \varepsilon \| \bar{K}_T \|_{Q,2} ) 
	\le (a/\varepsilon)^v,
	\quad
	0 < \forall \varepsilon \le 1,
	\forall n \ge 1,
\end{align*}
where $N(T, d, \varepsilon)$ denotes the $\varepsilon$-covering number of a semimetric space $(T, d)$ for $\varepsilon > 0$, $\| f \|_{Q,p} \coloneqq (Q|f|^p)^{1/p}$, and the supremum is taken over all finitely discrete probability measures.
Letting $\sigma_n^2$ be  any  positive constant such that $\sup_{k_T\in\mathcal{K}_T} Pk_T^2 \le \sigma_n^2 \le \|\bar{K}_T\|_{P,2}^2$, then $\sigma_n^2 = O(1)$ because
\begin{align}
	Pk_T^2 \lesssim \int K^2(u)f_Z(z+uh)du = O(1). \nonumber
\end{align}
In addition, we can see that $\|\bar{K}_T\|_{Q,2} = O(h^{-1/2})$ and
\begin{align}
	K_n 
	\coloneqq v( (\log n) \vee (\log( a \|\bar{K}_T\|_{Q,2} / \sigma_n ))) 
	=  O((\log n) \vee (\log h^{-1/2}))  = O(\log n). \nonumber
\end{align}
Thus, $\|\bar{K}_T\|_{Q,2}^2 \sigma_n^4 K_n^4 / n \lesssim n^{-1}h^{-1}(\log n)^4 \lesssim n^{-c}$ for some positive constant $c$ under Assumption \ref{as:asymptotic1}(iii).
The above discussion has shown that the conditions of Theorem 3.2 of \citet{chernozhukov2014anti} are satisfied, implying the desired result.
\qed 

\subsection{Proof of the auxiliary results for the local linear regression estimation} \label{subsec:LLR}

Recall that in Section \ref{subsec:bandwidth} the proposed undersmoothing for the LQR estimator is to use the IMSE-optimal bandwidth $h_{\mathrm{LL}}$ for the LLR estimator.
In this subsection, to derive this bandwidth, we prove the following asymptotic linear representation for the LLR estimator $\hat{\DR}_{g,t}^{\mathtt{LL},\ny}(z)$ defined in \eqref{eq:DR-LL-ny}:
\begin{align*}
	\hat{\DR}_{g,t}^{\mathtt{LL},\ny}(z) - {\DR}_{g,t}^{\ny}(z)
	& = \frac{1}{f_Z(z)} \frac{1}{nh} \sum_{i=1}^n  \left( B_{i,g,t}^{\mathtt{LL},\ny} - \mu_B(Z_i) \right) K_{i,h} + \frac{I_{2,K}}{2}\mu_B^{(2)}(z)h^2 + o_{\bP}(h^2) + o_{\bP}\left( \sqrt{\frac{\log n}{nh}} \right).
\end{align*}
Since the conditional expectation given $\bm{Z}$ of the first term in the right-hand side is $0$, it is easy to see that the asymptotic bias is given by
\begin{align*}
	\bE\left[ \hat{\DR}_{g,t}^{\mathtt{LL},\ny}(z) \;  \middle| \; \bm{Z} \right] - {\DR}_{g,t}^{\ny}(z)
	& = \frac{I_{2,K}}{2}\mu_B^{(2)}(z)h^2 + o_{\bP}(h^2) + o_{\bP}\left( \sqrt{\frac{\log n}{nh}} \right).
\end{align*}
Moreover, given the IID assumption in Assumption \ref{as:iid}, in the same manner as the derivation of the asymptotic variance of the LQR estimator, we can show that
\begin{align*}
	\Var \left[ \frac{1}{f_Z(z)} \frac{1}{nh} \sum_{i=1}^n K_{i,h} B_{i,g,t}^{\mathtt{LL},\ny} \; \middle| \; \bm{Z} \right]
	& = \frac{1}{nh} \frac{ \sigma_B^2(z) }{ f_Z(z) } \left( \int K^2(u) du \right) + o_{\bP}\left( \frac{1}{nh} \right).
\end{align*}
Accordingly, we focus on the proof of the asymptotic linear representation presented above.
Since this asymptotic linear representation for the LLR estimator is of the same form as that for the LQR estimator in the main text, we can perform uniform inference for CATT based on the LLR estimator in the same way as in the case of the LQR estimator.
The details, such as the construction of the standard error, the uniform critical values, and the bandwidth selection, can be found in the preprint version of this article (\citealp{imai2023doubly}).

\subsubsection{Proof of the asymptotic linear representation}

To simplify notation, with an abuse of notation, we suppress the superscript ``$\mathrm{LL}$'' in this subsection.

For the LLR estimator with $p = 1$, observe that
\begin{align} \label{eq:linear1}
	\begin{split}
		\hat{\DR}_{g,t}^{\ny}(z) - {\DR}_{g,t}^{\ny}(z)
		& = \hat \mu_{\hat A}(z) - \mu_A(z) \\
		& = \Big( \hat \mu_A(z) - \hat \mu_{\bE[A\mid Z]}(z) \Big) + \Big( \hat \mu_{\bE[A\mid Z]}(z) - \mu_A(z) \Big) + \Big( \hat \mu_{\hat A}(z) - \hat \mu_A(z) \Big) \\
		& \eqqcolon (\mathbf{LLR.I})  + (\mathbf{LLR.II}) + (\mathbf{LLR.III}),
	\end{split}
\end{align}
where each $\hat \mu$ in the right-hand side denotes the corresponding LLR estimator.

To proceed, we provide the following auxiliary results that we use repeatedly.
As in the case of $p=2$, from Lemma \ref{lem:Gamma} and the continuous mapping theorem,
\begin{align} \label{eq:Gamma_inv}
	\tilde{\bm{\Gamma}}_{(h,K,1)}^{-1}
	= \bm{\Gamma}_{(K,1)}^{-1} + o_{\bP}(1)
	= \frac{ 1 }{ f_Z(z) } \diag \left[ 1, \left( \int u^2 K(u) du \right)^{-1} \right] + o_{\bP}(1).
\end{align}
In addition, it is straightforward to show that
\begin{align} \label{eq:e0}
	\bm{e}_0^\top \bm{H}_1^{-1} = \bm{e}_0^\top,
	\qquad
	\bm{e}_0^\top \bm{\Gamma}_{(K,1)}^{-1} = \frac{1}{f_Z(z)} \bm{e}_0^\top,
	\qquad 
	\bm{e}_0^\top \bm{r}_1(u) = 1.
\end{align}

\paragraph{Evaluation on $(\mathbf{LLR.I})$.} 
To evaluate the first term in \eqref{eq:linear1}, observe that
\begin{align*}
	\hat \mu_A(z) - \hat \mu_{\bE[A\mid Z]}(z) 
	&= \frac{1}{n} \bm{e}_0^\top \tilde{\bm{\Gamma}}_{(h,K,1)}^{-1} \tilde{\bm{\Omega}}_{(h,K,1)} \left( \bm{A}_{g,t}^{\ny} - \bE\left[ \bm{A}_{g,t}^{\ny} \mid Z\right] \right) \\
	&= \bm{e}_0^\top \tilde{\bm{\Gamma}}_{(h,K,1)}^{-1} \frac{1}{nh} \sum_{i=1}^n K_{i,h} \bm{r}_1(u_{i,h}) \left(  A_{i,g,t}^{\ny} - \mu_A(Z_i) \right) \\
	& = \left( \bm{e}_0^\top\bm{\Gamma}_{(K,1)}^{-1} + o_{\bP}(1) \right) \frac{1}{nh} \sum_{i=1}^n K_{i,h} \bm{r}_1(u_{i,h}) \left(  A_{i,g,t}^{\ny} - \mu_A(Z_i) \right) \\
	& = \frac{1}{nh} \frac{1}{f_Z(z)} \sum_{i=1}^n K_{i,h} \left(  A_{i,g,t}^{\ny} - \mu_A(Z_i) \right) + o_{\bP}\left(  \frac{1}{nh} \sum_{i=1}^n K_{i,h} \bm{r}_1(u_{i,h}) \left(  A_{i,g,t}^{\ny} - \mu_A(Z_i) \right)\right),
\end{align*}
where the last equality follows from (\ref{eq:Gamma_inv}). 
From Lemma \ref{lem:eval_ka-mu}, which states that for any non-negative integer $q\in\mathbb{Z}_+$,
\begin{align*}
	\frac{1}{nh} \sum_{i=1}^n K_{i,h} (Z_i-z)^q \left(  A_{i,g,t}^{\ny} - \mu_A(Z_i) \right) = O_{\bP}\left( h^q \sqrt{\frac{\log n}{nh}}\right) 
\end{align*}
holds uniformly in $(g,t,z) \in \mathcal{A}^{\ny}$, we can see that
\begin{align*}
	\hat \mu_A(z) - \hat \mu_{\bE[A\mid Z]}(z) = \frac{1}{nh} \frac{1}{f_Z(z)} \sum_{i=1}^n K_{i,h} \left(  A_{i,g,t}^{\ny} - \mu_A(Z_i) \right) + o_{\bP}\left( \sqrt{\frac{\log n}{nh}} \right)
\end{align*}
holds uniformly in $(g,t,z) \in \mathcal{A}^{\ny}$.

\paragraph{Evaluation on $(\mathbf{LLR.II})$.}
To evaluate the second term in (\ref{eq:linear1}), denoting $\bm{\mu}_{A,(h)} = (\mu_A(z), h \cdot \mu_A^{(1)}(z))^\top$, in the same manner as the evaluation of $(\mathbf{LQR.II})$, we can see that 
\begin{align*}
	\hat \mu_{\bE[A\mid Z]}(z) - \mu_A(z) 
	&= \frac{1}{n} \bm{e}_0^\top \tilde{\bm{\Gamma}}_{(h,K,1)}^{-1} \tilde{\bm{\Omega}}_{(h,K,1)}  \bE\left[ \bm{A}_{g,t}^{\ny} \mid Z\right] -  \bm{e}_0^\top \tilde{\bm{\Gamma}}_{(h,K,1)}^{-1}\tilde{\bm{\Gamma}}_{(h,K,1)}\bm{\mu}_{A,(h)}\\
	& = \frac{\mu_A^{(2)}(z)}{2nh f_Z(z)} \sum_{i=1}^n K_{i,h} (Z_i-z)^2 + o_{\bP}\left( \frac{1}{nh} \sum_{i=1}^n K_{i,h} \bm{r}_{1}(u_{i,h}) (Z_i-z)^2 \right)
\end{align*}
holds uniformly in $(g,t,z) \in \mathcal{A}^{\ny}$.
Then, it holds that
\begin{align*}
	\hat \mu_{\bE[A\mid Z]}(z) - \mu_A(z) 
	& = \frac{I_{2,K}}{2}\mu_A^{(2)}(z) h^2 + o_{\bP}(h^2) + o_{\bP}\left( \sqrt{\frac{\log n}{nh}} \right)
\end{align*}
uniformly in $(g,t,z) \in \mathcal{A}^{\ny}$, where the leading term is derived from Lemma \ref{lem:eval_kq}, which states that
\begin{align*}
	\frac{1}{nh}\sum_{i=1}^n K_{i,h} (Z_i-z)^2 = I_{2,K}f_Z(z)h^2 + o_{\bP}(h^2) + O_{\bP}\left(h^2 \sqrt{\frac{\log n}{nh}} \right)
\end{align*}
uniformly in $z \in \mathcal{I}$, and the convergence rate of the remainder term also follows from Lemma \ref{lem:eval_kq}, which states that the following equations hold uniformly in $z \in \mathcal{I}$:
\begin{align*}
	& \frac{1}{nh}\sum_{i=1}^n K_{i,h}(Z_i - z)^2 = O(h^2) + O_{\bP}\left( h^2 \sqrt{ \frac{\log n}{nh} } \right),\\
	&  \frac{1}{nh}\sum_{i=1}^n K_{i,h}(Z_i - z)^3 = o(h^3) + O_{\bP}\left( h^3 \sqrt{ \frac{\log n}{nh} } \right).
\end{align*}

\paragraph{Evaluation on $(\mathbf{LLR.III})$.}
In the same manner as the evaluation of $(\mathbf{LQR.III})$, we can see that
\begin{align*}
	& \hat \mu_{\hat A}(z) - \hat \mu_A(z) \\
	& = \frac{1}{\hat \mu_{\hat R}(z) \mu_R(z)} \bm{e}_0^\top \tilde{\bm{\Gamma}}_{(h,K,1)}^{-1} \left( \frac{1}{nh} \sum_{i=1}^n K_{i,h} \bm{r}_1(u_{i,h}) E_{i,g,t}^{\ny} \right) \left( \hat \mu_{\hat R}(z) - \mu_R(z) \right) \\
	& \quad - \frac{1}{\hat \mu_G(z) \mu_G(z)} \bm{e}_0^\top \tilde{\bm{\Gamma}}_{(h,K,1)}^{-1} \left( \frac{1}{nh} \sum_{i=1}^n K_{i,h} \bm{r}_1(u_{i,h}) F_{i,g,t}^{\ny} \right) \left( \hat \mu_G(z) - \mu_G(z) \right) \\
	& \quad + \bm{e}_0^\top \tilde{\bm{\Gamma}}_{(h,K,1)}^{-1} \left( \frac{1}{nh} \sum_{i=1}^n K_{i,h} \bm{r}_1(u_{i,h}) \left( \frac{ 1 }{ \hat \mu_G(z) } \left( \hat F_{i,g,t}^{\ny} - F_{i,g,t}^{\ny} \right) - \frac{ 1 }{ \hat \mu_{\hat R}(z) } \left( \hat E_{i,g,t}^{\ny} - E_{i,g,t}^{\ny} \right) \right) \right) \\
	& \eqqcolon (\mathbf{LLR.III\en I}) - (\mathbf{LLR.III\en II}) + (\mathbf{LLR.III\en III}).
\end{align*}
In what follows, we examine each term.

For $(\mathbf{LLR.III \en I})$, likewize $(\mathbf{LQR.III \en I})$, we have
\begin{align*}
	(\mathbf{LLR.III \en I})
	& = \frac{1}{\hat \mu_{\hat R}(z) \mu_R(z)} \hat \mu_E(z) \left( \hat \mu_{\hat R}(z) - \mu_R(z) \right) \\
	& = \frac{\mu_E(z)}{\mu_R^2(z)} \left( \hat \mu_R(z) - \mu_R(z) \right) + o_{\bP}(h^2) + o_{\bP}\left( \sqrt{\frac{\log n}{nh}} \right),
\end{align*}
where we used Lemma \ref{lem:eval_mu_hatr}, which states that 
\begin{align*}
	\hat \mu_{\hat R}(z) = \hat \mu_R(z) + O_{\bP}\left( \frac{1}{\sqrt{n}} \right),
\end{align*}
and Lemma \ref{lem:eval_mu_re}, which states that
\begin{align*}
	& \hat \mu_R(z) - \mu_R(z) = O(h^2) + O_{\bP}\left( \sqrt{\frac{\log n}{nh}} \right), 
	\qquad
	\hat \mu_E(z) - \mu_E(z) = O(h^2) + O_{\bP}\left( \sqrt{\frac{\log n}{nh}} \right).
\end{align*}
Furthermore, in the same manner as the evaluation of $\hat \mu_A(z) - \hat{\mu}_{\bE[A \mid Z]}(z)$ and $\hat{\mu}_{\bE[A \mid Z]}(z) - \mu_A(z)$, we can show that
\begin{align*}
	\hat \mu_R(z) - \mu_R(z)
	& = \frac{1}{nh} \frac{1}{f_Z(z)} \sum_{i=1}^n K_{i,h} (R_{i,g,t} - \mu_R(Z_i)) + \frac{I_{2,K}}{2}\mu_R^{(2)}(z) h^2 + o_{\bP}(h^2) + o_{\bP}\left( \sqrt{\frac{\log n}{nh}} \right).
\end{align*}
Thus, we have shown that
\begin{align*}
	(\mathbf{LLR.III \en I})
	= \frac{ \mu_E(z) }{ f_Z(z) \mu_R^2(z) } \frac{1}{nh} \sum_{i=1}^n K_{i,h} (R_{i,g,t}-\mu_R(Z_i)) + \frac{I_{2,K}}{2}\frac{ \mu_E(z) }{ \mu_R^2(z) } \mu_R^{(2)}(z)h^2 + o_{\bP}(h^2) + o_{\bP}\left( \sqrt{\frac{\log n}{nh}} \right).
\end{align*}

Similar to the case of $(\mathbf{LLR.III\en I})$, we can show that
\begin{align*}
	(\mathbf{LLR.III\en II})
	= \frac{ \mu_F(z) }{ f_Z(z) \mu_G^2(z) } \frac{1}{nh} \sum_{i=1}^n K_{i,h} (G_{i,g}-\mu_G(Z_i)) + \frac{I_{2,K}}{2} \frac{ \mu_F(z) }{ \mu_G^2(z) } \mu_G^{(2)}(z)  + o_{\bP}(h^2) + o_{\bP}\left( \sqrt{\frac{\log n}{nh}} \right).
\end{align*}

For $(\mathbf{LLR.III\en III})$, the $1/\sqrt{n}$-consistency of the first-stage parametric estimation ensures that
\begin{align*}
	(\mathbf{LLR.III\en III})
	& = \bm{e}_0^\top \tilde{\bm{\Gamma}}_{(h,K,1)}^{-1} \left( \frac{1}{nh} \sum_{i=1}^n K_{i,h} \bm{r}_1(u_{i,h}) \left( \frac{ 1 }{ \hat \mu_G(z) } \left( \hat F_{i,g,t}^{\ny} - F_{i,g,t}^{\ny} \right) - \frac{ 1 }{ \hat \mu_{\hat R}(z) } \left( \hat E_{i,g,t}^{\ny} - E_{i,g,t}^{\ny} \right) \right) \right) \\
	& = O_{\bP}\left( \frac{1}{\sqrt{n}} \right),
\end{align*}
where we used \eqref{eq:Gamma_inv}, Assumption \ref{as:asymptotic1}(v), and Lemma \ref{lem:eval_kq}, which states that
\begin{align*}
	\frac{1}{nh} \sum_{i=1}^n K_{i,h}
	=  O_{\bP}(1), 
	\qquad
	\frac{1}{nh} \sum_{i=1}^n K_{i,h}(Z_i-z)
	=  o_{\bP}(1).
\end{align*}

Summing up, we obtain the desired result.
\qed

\section{Lemmas} \label{sec:lemma}

\subsection{Lemmas for Lemma \ref{lem:DR}} \label{subsec:lemma:identification}

The following is Lemma A.2 of \citet{callaway2021difference}.
Let ${\CATT}_{g,t}(X) \coloneqq \bE [ Y_t(g) - Y_t(0) \mid X, G_g = 1]$.

\begin{lemma}[Lemma A.2, \citealp{callaway2021difference}] \label{lem:ATT-nev}
	Suppose that Assumptions \ref{as:staggered}--\ref{as:overlap} hold.
	Then, for all $g \in \mathcal{G}$ and $t \in \{ 2, \dots, \mathcal{T}  \}$ such that $g \le t \le \bar{g}$,
	\begin{align*}
		{\CATT}_{g,t}(X) 
		= \bE[ Y_t - Y_{g-1} \mid X, G_g = 1 ] - \bE[ Y_t - Y_{g-1} \mid X, D_{t} = 0, G_g=0 ]
		\quad 
		\text{a.s.}
	\end{align*}
\end{lemma}

To state the next lemma, we define the nonparametric IPW, OR, and DR estimands:
\begin{align*}
	{\IPW}_{g,t}^{\ny}(Z) 
	& \coloneqq \bE \left[ \left( \frac{G_g}{\bE[ G_g \mid Z]} - \frac{ R_{g,t}(W) }{\bE [ R_{g,t}(W) \mid Z ]} \right) (Y_t - Y_{g-1}) \; \middle| \; Z \right], \\
	{\OR}_{g,t}^{\ny}(Z)
	& \coloneqq \bE \left[ \frac{G_g}{\bE[G_g \mid Z]} \left( Y_t - Y_{g-1} - m_{g,t}^{\ny}(X) \right) \; \middle| \; Z \right], \\
	{\DR}_{g,t}^{\ny}(Z)
	& \coloneqq \bE \left[ \left( \frac{G_g}{\bE[ G_g \mid Z]} - \frac{ R_{g,t}(W) }{\bE [ R_{g,t}(W) \mid Z ]} \right) \left( Y_t - Y_{g-1} - m_{g,t}^{\ny}(X) \right) \; \middle| \; Z \right],
\end{align*}
where $R_{g,t}(W) \coloneqq p_{g,t}(X) (1 - D_t) (1 - G_g) / [ 1 - p_{g,t}(X) ]$.

\begin{lemma} \label{lem:nonparametric}
	Under Assumptions \ref{as:staggered}--\ref{as:overlap}, ${\CATT}_{g,t}(z) = {\IPW}_{g,t}^{\ny}(z) = {\OR}_{g,t}^{\ny}(z) = {\DR}_{g,t}^{\ny}(z)$ for all $z \in \mathcal{I}$, $g \in \mathcal{G}$, and $t \in \{ 2, \dots, \mathcal{T}\}$ such that $g\le t \le \bar{g}$.
\end{lemma}

\begin{proof}
	The proof is almost the same as Theorem 1 of \citet{callaway2021difference} and is thus omitted.
\end{proof}

\subsection{Lemmas for Theorems \ref{thm:bias_variance} and \ref{thm:approx_us_distribution}}

\begin{lemma} \label{lem:K_vc_type}
	Under Assumption \ref{as:asymptotic1}(ii), $\mathcal{K} = \{ s \mapsto K(hs + z) : h > 0, z \in \mathbb{R} \}$ is VC type (cf. Definition 2.1 of \citealp{chernozhukov2014gaussian}).
\end{lemma}

\begin{proof}
	Note that $\{ hs + z: s, z \in \mathbb{R} \}$ is a VC-subgraph class.
	This is because (i) $\{ hs + z: s, z \in \mathbb{R} \}$ is a vector space of dimension 2, and (ii) any finite-dimensional vector space $\mathcal{F}$ is a VC-subgraph class with VC index at most $\dim(\mathcal{F}) + 2$ (cf. Lemma 2.6.15 of \citealp{vandervaart96weak}). 
	
	Since the continuous differentiability of the kernel function $K$ in Assumption \ref{as:asymptotic1}(ii) implies that $K$ is a function of bounded variation, $K$ can be written as the difference of two non-decreasing functions.
	Then, Lemmas 2.6.16 and 2.6.18 in \cite{vandervaart96weak} imply that $\mathcal{K}$ is a VC-subgraph class and so VC type. 
\end{proof}

For the subsequent analysis, we suppose that the IID random variables are taken from a common distribution $P$ and use the notation $Pf \coloneqq \int fdP$.

\begin{lemma} \label{lem:eval_kq}
	Under Assumptions \ref{as:iid} and \ref{as:asymptotic1}(i)--(iv), the following equation holds uniformly in $z \in \mathcal{I}$:
	\begin{align*}
		\frac{1}{nh}\sum_{i=1}^n K_{i,h}(Z_i-z)^q 
		=  
		\begin{cases} 
			\displaystyle h^q f(z) \int u^qK(u) du + O(h^{q+2})  +  O_{\bP}\left( h^q \sqrt{ \frac{ \log n }{ nh }  } \right)         & \text{if $q$ is even,} \\
			\displaystyle h^{q+1} f'(z) \int u^{q+1}K(u)du + O(h^{q+3})  +  O_{\bP}\left( h^q \sqrt{ \frac{ \log n }{ nh }  } \right) & \text{if $q$ is odd.} 
		\end{cases}
	\end{align*}
\end{lemma}

\begin{proof}
	Fix a non-negative integer $q \in \mathbb{Z}_{+}$.
	Observe that
	\begin{align*}
		\frac{1}{nh}\sum_{i=1}^n K_{i,h}(Z_i-z)^q
		= \bE\left[ \frac{1}{nh}\sum_{i=1}^n K_{i,h}(Z_i-z)^q \right] + \left( \frac{1}{nh}\sum_{i=1}^n K_{i,h}(Z_i-z)^q - \bE\left[ \frac{1}{nh}\sum_{i=1}^n K_{i,h}(Z_i-z)^q \right] \right).
	\end{align*}
	For the first term, using Assumption \ref{as:asymptotic1}(i)--(ii), it is straightforward to see that
	\begin{align*}
		\bE\left[ \frac{1}{nh}\sum_{i=1}^n K_{i,h}(Z_i-z)^q \right]
		& = \int K(u)(uh)^q f(z+uh)du \\
		& = 
		\begin{cases} 
			\displaystyle h^q f(z) \int u^qK(u) du + O(h^{q+2})         & \text{if $q$ is even,} \\
			\displaystyle h^{q+1} f'(z) \int u^{q+1}K(u)du + O(h^{q+3}) & \text{if $q$ is odd.}
		\end{cases}
	\end{align*}
	
	To evaluate the second term, we define the function class $\mathcal{K}_q$ and the empirical process $\mathbb{G}_nk_q$ as follows:
	\begin{align*}
		& \mathcal{K}_q \coloneqq \left\{ s \mapsto K\left( \frac{s-z}{h} \right) (s-z)^q : h>0,  z\in\mathcal{I} \right\}, \\
		& \mathbb{G}_nk_q \coloneqq \frac{1}{\sqrt{n}}\sum_{i=1}^n \{ K_{i,h}(Z_i-z)^q - \mathbb{E}[K_{i,h}(Z_i-z)^q] \}.
	\end{align*}
	We show that $\mathcal{K}_q$ is VC type.
	To this end, we first observe that $\{ s \mapsto (hs + z)^q: h > 0, z\in\mathcal{I} \}$ is a class of functions of bounded variation.
	When $q = 0$, it is trivial that $s \mapsto 1$ is of bounded variation.
	When $q = 1$, because any non-decreasing function is of bounded variation, $s \mapsto (hs + z)$ is of bounded variation.
	When $q \ge 2$, $s \mapsto (hs + z)^q$ is of bounded variation since it is a continuously differentiable function.
	Next, as in the proof of Lemma \ref{lem:K_vc_type}, we can see that $\{ s \mapsto K(hs + z): h > 0, z\in\mathcal{I} \}$ is a class of functions of bounded variation.
	Then, because the product of functions of bounded variation is also of bounded variation, we have shown that $\mathcal{K}_q$ is a class of functions of bounded variation.
	Then, Lemmas 2.6.16 and 2.6.18 of \cite{vandervaart96weak} imply that $\mathcal{K}_q$ is a VC-subgraph class and so VC type.
	By Definition 2.1 of \citet{chernozhukov2014gaussian}, this implies that there are some positive constants $C$ and $v$ such that 
	\begin{align*}
		\sup_Q N (\mathcal{K}_q, \|\cdot\|_{Q,2}, \varepsilon\|\bar{K}_q\|_{Q,2}) \le (C/\varepsilon)^v,
		\qquad 
		0 < \forall \varepsilon \le 1, \forall n \ge 1,
	\end{align*}
	where the supremum is taken over all finitely discrete probability measures.
	From the definitions of $\mathcal{K}_q$ and $\mathbb{G}_n k_q$, observe that
	\begin{align*}
		& \sup_{z\in\mathcal{I}} \left| \frac{1}{nh}\sum_{i=1}^n K_{i,h}(Z_i-z)^q - \bE\left[ \frac{1}{nh}\sum_{i=1}^n K_{i,h}(Z_i-z)^q \right] \right| = \frac{1}{\sqrt{n}h} \sup_{k_q\in\mathcal{K}_q} |\mathbb{G}_n k_q|.
	\end{align*}
	Let $\bar{K}_q$ be an envelope of $\mathcal{K}_q$ and $\sigma^2>0$ be any positive constant such that $\sup_{k_q \in \mathcal{K}_q} P k_q^2 \le \sigma^2 \le \| \bar{K}_q\|_{P,2}^2$. 
	Here, under Assumption \ref{as:asymptotic1}(i)--(ii), we can see that
	\begin{align*}
		P k_q^2 \lesssim h\int K^2(u)(uh)^{2q}du \lesssim h^{2q+1}, \qquad \| \bar{K}_q \|_{P,2}^2 = O(1).
	\end{align*}
	Then, by Corollary 5.1 of \citet{chernozhukov2014gaussian}, we have
	\begin{align*}
		\bE\left[\sup_{ k_q \in \mathcal{K}_q }\left| \mathbb{G}_n k_q \right|\right] 
		&\lesssim \sqrt{v\sigma^2 \log\left(\frac{C\| \bar{K}_q \|_{P,2}}{\sigma}\right)} + \frac{vB}{\sqrt{n}}\log\left(\frac{C\| \bar{K}_q \|_{P,2}}{\sigma}\right) \\
		& = O\left( h^q \sqrt{ h \log\left(\frac{1}{h^q}\right) } \right) \\
		& = O\left( h^q \sqrt{ h \log n } \right),
	\end{align*}
	with some constant $B > 0$.
	Thus, Markov's inequality implies that
	\begin{align*}
		\frac{ 1 }{ \sqrt{n} h } \sup_{k_q\in\mathcal{K}_q} \left| \mathbb{G}_n k_q \right| = O_{\bP}\left( h^q \sqrt{ \frac{ \log n }{ nh }  } \right).
	\end{align*}
	
	Summing up, we obtain the desired result.
\end{proof}

\begin{lemma} \label{lem:Gamma}
	Suppose that Assumptions \ref{as:iid} and \ref{as:asymptotic1}(i)--(iii) hold.
	For each $p = 1, 2$, $\tilde{\bm{\Gamma}}_{(h,K,p)} = \bm{\Gamma}_{(K,p)} + o_{\bP}(1)$ uniformly in $z\in\mathcal{I}$.
\end{lemma}

\begin{proof}
	Recall that
	\begin{align*}
		\tilde{\bm{\Gamma}}_{(h,K,p)}
		= \frac{1}{nh} \sum_{i=1}^n K_{i,h} \bm{r}_p(u_{i,h}) \bm{r}_p(u_{i,h})^\top 
		= \frac{1}{nh} \sum_{i=1}^n K_{i,h} \left[ u_{i,h}^{l+m-2} \right]_{1 \le l,m \le (p+1)}.
	\end{align*}
	Denoting $q = l + m - 2 \ge 0$, we can see that
	\begin{align*}
		\bE \left[ \frac{1}{nh} \sum_{i=1}^n K_{i,h} u_{i,h}^q \right]
		= \int K(u) u^q f_Z(z + uh) du 
		= f_Z(z) \int u^q K(u) du + O(h)
	\end{align*}
	and, from Lemma \ref{lem:eval_kq}, that
	\begin{align*}
		\sup_{z\in\mathcal{I}} \left|\frac{1}{nh}\sum_{i=1}^n K_{i,h} u_{i,h}^q - \bE\left[\frac{1}{nh}\sum_{i=1}^n K_{i,h} u_{i,h}^q \right]\right| = o_p(1).
	\end{align*}
	Thus, we have the desired result.
\end{proof}

\begin{lemma} \label{lem:eval_bckq}
	Under Assumptions \ref{as:iid} and \ref{as:asymptotic1}(i)--(iv), the following equation holds uniformly in $z \in \mathcal{I}$:
	\begin{align*}
		\frac{1}{ nh } \sum_{i=1}^n K_{i,h} \Psi_{i,h} (Z_i - z)^q = 
		\begin{cases}
			\displaystyle f_Z(z) + O(h^4) + O_{\bP}\left( \sqrt{\frac{\log n}{nh}}\right) & \text{if $q = 0$,} \\
			\displaystyle \frac{f_Z^{(4-q)}(z)}{(4-q)!}h^4 \left( \frac{I_{4,K}^2 - I_{2,K}I_{6,K}}{I_{4,K} - I_{2,K}^2} \right) + o(h^4) + O_{\bP}\left( h^q\sqrt{\frac{\log n}{nh}}\right) & \text{if $1 \le q \le 4$,} \\
			\displaystyle o(h^4) + O_{\bP}\left( h^q\sqrt{\frac{\log n}{nh}}\right) & \text{if $q \ge 5$.}
		\end{cases}
	\end{align*}
\end{lemma}

\begin{proof}
	Fix a non-negative integer $q\in\mathbb{Z}_+$. Observe that 
	\begin{align*}
		& \frac{1}{ nh} \sum_{i=1}^n K_{i,h} \Psi_{i,h} (Z_i-z)^q \\
		& = \mathbb{E}\left[  \frac{1}{ nh} \sum_{i=1}^n K_{i,h} \Psi_{i,h} (Z_i-z)^q \right] + \left(  \frac{1}{ nh} \sum_{i=1}^n K_{i,h} \Psi_{i,h} (Z_i-z)^q - \mathbb{E}\left[  \frac{1}{ nh} \sum_{i=1}^n K_{i,h} \Psi_{i,h} (Z_i-z)^q \right] \right).
	\end{align*}
	The first term can be expanded as
	\begin{align*}
		& \mathbb{E}\left[  \frac{1}{ nh} \sum_{i=1}^n K_{i,h} \Psi_{i,h} (Z_i-z)^q \right] \\ 
		& = \int K(u)\left( \frac{I_{4,K} - u^2I_{2,K}}{I_{4,K} - I_{2,K}^2} \right) (uh)^q f_Z(z+uh)du \\
		& = h^q \int K(u)\left( \frac{u^qI_{4,K} - u^{q+2}I_{2,K}}{I_{4,K} - I_{2,K}^2} \right)  f_Z(z+uh)du \\
		& = \sum_{k=0}^4 \frac{f_Z^{(k)}(z)}{k!} h^{q+k} \left( \frac{I_{q+k,K}I_{4,K} - I_{q+k+2,K}I_{2,K}}{I_{4,K} - I_{2,K}^2} \right) + o(h^{q+4}) \\
		& = 
		\begin{cases}
			\displaystyle f_Z(z) + \frac{f_Z^{(4)}(z)}{4!}h^4 \left( \frac{I_{4,K}^2 - I_{2,K}I_{6,K}}{I_{4,K} - I_{2,K}^2} \right) + o(h^4) & \text{if $q = 0$,} \\
			\displaystyle \frac{f_Z^{(4-q)}(z)}{(4-q)!}h^4 \left( \frac{I_{4,K}^2 - I_{2,K}I_{6,K}}{I_{4,K} - I_{2,K}^2} \right) + o(h^4) & \text{if $1 \le q \le 4$,} \\
			\displaystyle o(h^4) & \text{if $q \ge 5$}.
		\end{cases}
	\end{align*}
	For the second term, the same arguments as in the proof of Lemma \ref{lem:eval_kq} can show that it is $O_{\bP}(h^q \sqrt{(\log n) / (nh)})$.
	This completes the proof.
\end{proof}

\begin{lemma} \label{lem:eval_kaq}
	Under Assumptions \ref{as:iid} and \ref{as:asymptotic1}(i)--(iv), the following equation holds uniformly in $(g,t,z) \in \mathcal{A}^{\ny}$:
	\begin{align*}
		\frac{1}{nh}\sum_{i=1}^n K_{i,h}(Z_i-z)^qA_{i,g,t}^{\ny} 
		=
		\begin{cases}
			\displaystyle O(h^q) + O_{\bP}\left( h^q \sqrt{\frac{\log n}{nh}} \right)     & \text{if $q$ is even,} \\
			\displaystyle O(h^{q+1}) + O_{\bP}\left( h^q \sqrt{\frac{\log n}{nh}} \right) & \text{if $q$ is odd.}
		\end{cases} 
	\end{align*}
\end{lemma}

\begin{proof}
	The proof is almost the same as that for Lemma \ref{lem:eval_kq}.
\end{proof}

\begin{lemma} \label{lem:eval_bckaq}
	Under Assumption \ref{as:iid} and \ref{as:asymptotic1}(i)--(iv), the following equation holds uniformly in $(g,t,z) \in \mathcal{A}^{\ny}$:
	\begin{align*}
		& \frac{1}{ nh } \sum_{i=1}^n K_{i,h} \Psi_{i,h} (Z_i-z)^q A_{i,g,t}^{\ny} \\
		& = 
		\begin{cases}
			\displaystyle f_Z(z)\mu_A(z) + O(h^4) + O_{\bP}\left( \sqrt{\frac{\log n}{nh}}\right) & \text{if $q = 0$}, \\
			\displaystyle \frac{1}{(4-q)!}\frac{\partial^4\{f_Z(z)\mu_A(z)\}}{\partial z^4}h^4 \left( \frac{I_{4,K}^2 - I_{2,K}I_{6,K}}{I_{4,K} - I_{2,K}^2} \right) + o(h^4) + O_{\bP}\left( h^q\sqrt{\frac{\log n}{nh}}\right) & \text{if $1 \le q \le 4$,} \\
			\displaystyle o(h^4) + O_{\bP}\left( h^q\sqrt{\frac{\log n}{nh}}\right) & \text{if $q \ge 5$.}
		\end{cases}
	\end{align*}
\end{lemma}

\begin{proof}
	The proof is almost the same as that for Lemma \ref{lem:eval_bckq}.
\end{proof}

\begin{lemma} \label{lem:eval_ka-mu}
	Suppose that Assumptions \ref{as:iid} and \ref{as:asymptotic1}(i)--(iv) hold.
	The following equation holds uniformly in $(g,t,z) \in \mathcal{A}^{\ny}$:
	\begin{align*}
		\frac{1}{nh} \sum_{i=1}^n K_{i,h} \Psi_{i,h} (Z_i - z)^q \left(  A_{i,g,t}^{\ny} - \mu_A(Z_i) \right) 
		=  O_{\bP}\left( h^q \sqrt{\frac{\log n}{nh}} \right).
	\end{align*}
\end{lemma}

\begin{proof}
	
	Using the law of iterated expectations, we can see that the mean is zero:
	\begin{align*}
		& \bE\left[ \frac{1}{nh} \sum_{i=1}^n K_{i,h} \Psi_{i,h}(Z_i-z)^q \left(  A_{i,g,t}^{\ny} - \mu_A(Z_i) \right) \right] \\
		& = \bE\left[ \frac{1}{nh} \sum_{i=1}^n K_{i,h}\Psi_{i,h} (Z_i-z)^q \left(  \mathbb{E}[A_{i,g,t}^{\ny} \mid Z_i] - \mu_A(Z_i) \right) \right] \\
		& = 0.
	\end{align*}
	Given this, letting $K_{A-\mu,h}(W_i, g, t, z) \coloneq K_{i,h}\Psi_{i,h} ( A_{i,g,t}^{\ny} - \mu_A(Z_i) )$, define the function class $\mathcal{K}_{A-\mu}$ and the empirical process $\mathbb{G}_n k_{A-\mu}$ as follows:
	\begin{align*}
		& \mathcal{K}_{A-\mu} \coloneqq \Big\{ s \mapsto K_{A-\mu,h}(s,g,t,z)(s-z)^q : h>0, (g,t,z) \in \mathcal{A}^{\ny} \Big\}, \\
		& \mathbb{G}_n k_{A-\mu} \coloneq \frac{1}{\sqrt{n}} \sum_{i=1}^n K_{i,h}\Psi_{i,h} (Z_i-z)^q \left(  A_{i,g,t}^{\ny} - \mu_A(Z_i) \right).
	\end{align*}
	Let $\bar{K}_{A-\mu}$ be an envelope of $\mathcal{K}_{A-\mu}$, which exists under Assumption \ref{as:asymptotic1}(iv)-(b).
	By Corollary A.1(i) of \citet{chernozhukov2014gaussian}, $\mathcal{K}_{A-\mu}$ is VC type, that is, there are some positive constants $C$ and $v$ such that
	\begin{align*}
		\sup_Q N(\mathcal{K}_{A-\mu}, \|\cdot\|_{Q,2}, \varepsilon\|\bar{K}_{A-\mu}\|_{Q,2}) \le (C/\varepsilon)^v, \quad 0 < \forall \varepsilon \le 1, \forall n\ge 1,
	\end{align*}
	where the supremum is taken over all finitely discrete probability measures.
	
	To obtain the desired result, it suffices to show that
	\begin{align*}
		\frac{1}{\sqrt{n}h} \sup_{k_{A-\mu} \in \mathcal{K}_{A-\mu}} \left| \mathbb{G}_n k_{A-\mu} \right|
		= O_{\bP}\left( h^q \sqrt{\frac{\log n}{nh}} \right).
	\end{align*}
	Letting $\sigma^2>0$ be any positive constant such that $\sup_{k_{A-\mu}\in\mathcal{K}_{A-\mu}} Pk_{A-\mu}^2 \le \sigma^2 \le \| \bar{K}_{A-\mu} \|_{P,2}^2$, by Corollary 5.1 of \citet{chernozhukov2014gaussian}, we have
	\begin{align*}
		\bE\left[\sup_{k_{A-\mu}\in\mathcal{K}_{A-\mu}}\left|\mathbb{G}_nk_{A-\mu}\right|\right] &\lesssim \sqrt{v\sigma^2 \log\left(\frac{C\| \bar{K}_{A-\mu} \|_{P,2}}{\sigma}\right)} + \frac{v B}{\sqrt{n}}\log\left(\frac{C \| \bar{K}_{A-\mu} \|_{P,2}}{\sigma}\right)
	\end{align*}
	with some positive constant $B > 0$.
	Here, from the law of iterated expectation and Assumption \ref{as:asymptotic1}(ii) and (iv)-(b), we can see that
	\begin{align*}
		Pk_{A-\mu}^2 \lesssim  h \int K^2(u)\Psi^2(u) (uh)^{2q} f_Z(z + uh) du  \lesssim h^{2q+1}, 
		\qquad 
		\| \bar{K}_{A-\mu} \|_{P,2}^2 = O(1).
	\end{align*}
	Thus, by Markov's inequality, we have
	\begin{align*}
		\frac{1}{\sqrt{n}h}\sup_{k_{A-\mu} \in \mathcal{K}_{A-\mu} }\left|\mathbb{G}_n k_{A-\mu}\right| = O_{\bP}\left( h^q \sqrt{\frac{\log (1/h)}{nh}}\right)  = O_{\bP}\left( h^q \sqrt{\frac{\log n}{nh}}\right),
	\end{align*}
	implying the desired result.
\end{proof}

\begin{lemma} \label{lem:eval_mu_hatr}
	Suppose that Assumptions \ref{as:iid} and \ref{as:asymptotic1}(i)--(v) hold.
	The following equation holds uniformly in $(g,t,z) \in \mathcal{A}^{\ny}$:
	\begin{align*}
		\hat{\mu}_{\hat{R}}(z) = \hat{\mu}_R(z) + O_{\bP}(1/\sqrt{n}).
	\end{align*}
\end{lemma}

\begin{proof}
	For the LQR estimator $\hat{\mu}_{\hat{R}}(z)$, we have
	\begin{align*}
		\hat{\mu}_{\hat{R}}(z) 
		& = \bm{e}_0^\top \tilde{\bm{\Gamma}}_{(h,K,2)}^{-1} \tilde{\bm{\Omega}}_{(h,K,2)} \bm{\hat{R}}_{g,t} \\
		& = \bm{e}_0^\top \tilde{\bm{\Gamma}}_{(h,K,2)}^{-1} \frac{1}{nh} \sum_{i=1}^n K_{i,h} \bm{r}_2(u_{i,h}) \left( R_{i,g,t} + \hat{R}_{i,g,t} - R_{i,g,t} \right) \\
		& = \hat \mu_R(z) + \bm{e}_0^\top \tilde{\bm{\Gamma}}_{(h,K,2)}^{-1} \frac{1}{nh} \sum_{i=1}^n K_{i,h} \bm{r}_2(u_{i,h}) \left( \hat{R}_{i,g,t} - R_{i,g,t} \right) \\
		& = \hat \mu_R(z) + \bm{e}_0^\top \bm{\Gamma}_{(K,2)}^{-1} \frac{1}{nh} \sum_{i=1}^n K_{i,h} \bm{r}_2(u_{i,h}) \left( \hat{R}_{i,g,t} - R_{i,g,t} \right) + o_{\bP}\left( \frac{1}{\sqrt{n}} \right) \\
		& = \hat \mu_R(z) + \frac{1}{f_Z(z)} \frac{1}{nh} \sum_{i=1}^n K_{i,h} \Psi_{i,h} \left( \hat{R}_{i,g,t} - R_{i,g,t
		} \right) + o_{\bP}\left( \frac{1}{\sqrt{n}} \right),
	\end{align*}
	where we used $\tilde{\bm{\Gamma}}_{(h,K,2)}^{-1} = \bm{\Gamma}_{(K,2)}^{-1} + o_{\bP}(1)$ by Lemma \ref{lem:Gamma}, Assumption \ref{as:asymptotic1}(v), and Lemma \ref{lem:eval_bckq}, which states that
	\begin{align*}
		\frac{1}{nh} \sum_{i=1}^n K_{i,h}\Psi_{i,h}
		= O_{\bP}(1).
	\end{align*}
	Observe that
	\begin{align*}
		\frac{1}{f_Z(z)} \frac{1}{nh} \sum_{i=1}^n K_{i,h}\Psi_{i,h} \left( \hat{R}_{i,g,t} - R_{i,g,t} \right)
		& \le \max_{1 \le i \le n} \left| \hat{R}_{i,g,t} - R_{i,g,t} \right|  \left( \frac{1}{f_Z(z)} \frac{1}{nh} \sum_{i=1}^n |K_{i,h}\Psi_{i,h}| \right) \\
		& = O_{\bP}\left( \frac{1}{\sqrt{n}} \right) \cdot O_{\bP}(1) \\
		& = O_{\bP}\left( \frac{1}{\sqrt{n}} \right),
	\end{align*}
	where the first equality follows from Assumption \ref{as:asymptotic1}(v).
	This completes the proof.
\end{proof}

\begin{lemma} \label{lem:eval_mu_re}
	Suppose that Assumptions \ref{as:iid} and \ref{as:asymptotic1}(i)--(iv) hold.
	For the LLR estimators $\hat{\mu}_R(z)$ and $\hat{\mu}_E(z)$ with $p = 1$, the following equations hold uniformly in $(g,t,z) \in \mathcal{A}^{\ny}$:
	\begin{align*}
		\hat{\mu}_R(z) - \mu_R(z) 
		& = O(h^2) + O_{\bP}\left( \sqrt{\frac{ \log n }{ nh } } \right), \\
		\hat{\mu}_E(z) - \mu_E(z) 
		& = O(h^2) + O_{\bP}\left( \sqrt{\frac{ \log n }{ nh } } \right).
	\end{align*}
\end{lemma}

\begin{proof}
	Observe that
	\begin{align*}
		\hat{\mu}_R(z)  - \mu_R(z) 
		= \Big( \hat{\mu}_R(z) - \hat{\mu}_{\mathbb{E}[R \mid Z]}(z) \Big) + \Big( \hat{\mu}_{\mathbb{E}[R \mid Z]}(z)   - \mu_R(z) \Big).
	\end{align*}
	Similar to the evaluation on $\mathbf{(LLR.I)}$ and $\mathbf{(LLR.II)}$, we can show that
	\begin{align*}
		\hat{\mu}_R(z)  - \hat{\mu}_{\mathbb{E}[ R \mid Z]}(z)
		& = O_{\bP}\left( \sqrt{\frac{\log n}{nh}} \right), \\
		\hat{\mu}_{\mathbb{E}[R\mid Z]}(z) - \mu_R(z)
		& = O(h^2) + O_{\bP}\left( \sqrt{\frac{\log n}{nh}} \right),
	\end{align*}
	which completes the proof.
\end{proof}

\begin{lemma} \label{lem:eval_mu_bce}
	Suppose that Assumptions \ref{as:iid} and \ref{as:asymptotic1}(i)--(iv) hold.
	For the LQR estimators $\hat{\mu}_R(z)$ and $\hat{\mu}_E(z)$, the following equations hold uniformly in $(g, t, z) \in \mathcal{A}^{\ny}$:
	\begin{align*}
		\hat{\mu}_R(z) - \mu_R(z) 
		& = O(h^4) + O_{\bP}\left( \sqrt{\frac{\log n}{nh}} \right), \\
		\hat{\mu}_E(z) - \mu_E(z) 
		& = O(h^4) + O_{\bP}\left( \sqrt{\frac{\log n}{nh}} \right).
	\end{align*}
\end{lemma}

\begin{proof}
	Observe that
	\begin{align*}
		\hat{\mu}_R(z) - \mu_R(z) 
		= \Big( \hat{\mu}_R(z) - \hat{\mu}_{\bE[R \mid Z]}(z) \Big) + \Big( \hat{\mu}_{\bE[R \mid Z]}(z) - \mu_R(z) \Big).
	\end{align*}
	Similar to the evaluation on $\mathbf{(LQR.I)}$ and $\mathbf{(LQR.II)}$, we can show that
	\begin{align*}
		\hat{\mu}_R(z) - \hat{\mu}_{\bE[R\mid Z]}(z)
		& = O_{\bP}\left( \sqrt{\frac{\log n}{nh}} \right), \\
		\hat{\mu}_{\bE[R \mid Z]}(z) - \mu_R(z)
		& = O(h^4) + O_{\bP}\left( \sqrt{\frac{\log n}{nh}} \right),
	\end{align*}
	which completes the proof.
\end{proof}

\section{Limited Treatment Anticipation} \label{sec:not-yet-treated}

In this section, we consider how to relax the no-anticipation condition in Assumption \ref{as:anticipation-ny} by allowing anticipation behavior in a limited but reasonable way, as in \citet{callaway2021difference}.
It turns out that even with such limited treatment anticipation, we can perform uniform inference for CATT in the same way as in Section \ref{sec:CATTgt}, except for some minor modifications to several variables.
Following the main text, the analysis in this section considers the not-yet-treated group as the comparison group.
Throughout this section, we maintain the staggered treatment adoption, random sampling, and overlap conditions in Assumptions \ref{as:staggered}, \ref{as:iid}, and \ref{as:overlap}.

\subsection{Identification}

Suppose that there is a known integer $\delta \ge 0$ such that each unit can anticipate treatment by $\delta$ periods.
More formally, we impose the following assumptions on anticipation behavior and conditional parallel trends, which are the same as Assumptions 3 and 5 of \citet{callaway2021difference}.

\begin{assumption}[Limited Treatment Anticipation] \label{as:anticipation-delta}
	There is a known $\delta \ge 0$ such that
	\begin{align*}
		\bE [ Y_t(g) \mid X, G_g = 1 ]
		= \bE [ Y_t(0) \mid X, G_g = 1 ] 
		\quad
		\text{a.s. for all $g \in \mathcal{G}$ and $t \in \{ 1, \dots, \mathcal{T} \}$ such that $t < g - \delta$}. 
	\end{align*}
\end{assumption}

\begin{assumption}[Conditional Parallel trends Based on ``Not-Yet-Treated'' Groups] \label{as:parallel-ny-delta}
	For each $g \in \mathcal{G}$ and each $(s, t) \in \{ 2, \dots, \mathcal{T} \} \times \{ 2, \dots, \mathcal{T} \}$ such that $t \ge g - \delta$ and $t + \delta \le s < \bar g$,
	\begin{align*}
		\bE[ Y_t(0) - Y_{t-1}(0) \mid X, G_g = 1 ]
		= \bE[ Y_t(0) - Y_{t-1}(0) \mid X, D_s = 0, G_g = 0 ]
		\quad 
		\text{a.s.}
	\end{align*}
\end{assumption}

Slightly different from the OR function and GPS in \eqref{eq:GPS-OR}, we define
\begin{align*}
	m_{g,t,\delta}^{\ny}(X) 
	& \coloneqq \bE[ Y_t - Y_{g-\delta-1} \mid X, D_{t+\delta} = 0, G_g = 0 ], \\
	p_{g,t+\delta}(X) 
	& \coloneqq \bP[G_g = 1 \mid X, G_g + (1 - D_{t+\delta})(1 - G_g) = 1].
\end{align*}
Let $m_{g,t,\delta}^{\ny}(X; \beta_{g,t,\delta}^{\ny})$ and $p_{g,t+\delta}(X; \pi_{g,t+\delta})$ be the corresponding parametric specifications known up to the finite-dimensional parameters $\beta_{g,t,\delta}^{\ny} \in \mathscr{B}_{g,t,\delta}^{\ny}$ and $\pi_{g,t+\delta} \in \Pi_{g,t+\delta}$.
Define the conditional DR estimand under limited treatment anticipation as follows:
\begin{align*}
	& {\DR}_{g,t,\delta}^{\ny}(Z; \beta_{g,t,\delta}^{\ny}, \pi_{g,t+\delta}) \\
	& \coloneqq \bE \left[ \left( \frac{G_g}{\bE[ G_g \mid Z]} - \frac{ R_{g,t+\delta}(W; \pi_{g,t+\delta}) }{\bE \left[ R_{g,t+\delta}(W; \pi_{g,t+\delta}) \mid Z \right]} \right) \left( Y_t - Y_{g-\delta-1} - m_{g,t,\delta}^{\ny}(X; \beta_{g,t,\delta}^{\ny}) \right) \; \middle| \; Z \right],
\end{align*}
where
\begin{align*}
	R_{g,t+\delta}(W; \pi_{g,t+\delta})
	\coloneqq \frac{p_{g,t+\delta}(X; \pi_{g,t+\delta}) (1 - D_{t+\delta}) (1 - G_g)}{1 - p_{g,t+\delta}(X; \pi_{g,t+\delta})}.
\end{align*}

Lastly, we impose the following assumption, which replaces Assumption \ref{as:parametric-ny} to account for limited treatment anticipation.
Let $\mathcal{G}_{\delta} \coloneqq \mathcal{G} \cap \{ 2 + \delta, 3 + \delta, \dots, \mathcal{T} \}$.

\begin{assumption}[Parametric Models for the ``Not-Yet-Treated'' Group] \label{as:parametric-ny-delta}
	For each $g \in \mathcal{G}_{\delta}$ and $t \in \{ 2, \dots, \mathcal{T} - \delta \}$ such that $g - \delta \le t < \bar g -\delta$, either condition is satisfied:
	\begin{enumerate}[(i)]
		\item There exists a unique $\beta_{g,t,\delta}^{\ny*} \in \mathscr{B}_{g,t,\delta}^{\ny}$ such that $m_{g,t,\delta}^{\ny}(X) = m_{g,t,\delta}^{\ny}(X; \beta_{g,t,\delta}^{\ny*})$ a.s.
		\item There exists a unique $\pi_{g,t+\delta}^* \in \Pi_{g,t+\delta}$ such that $p_{g,t+\delta}(X) = p_{g,t+\delta}(X; \pi_{g,t+\delta}^*)$ a.s.
	\end{enumerate}
\end{assumption}

The following lemma states that even with limited treatment anticipation, $\CATT_{g,t}(z)$ is identifiable from the above conditional DR estimand for each $(g, t, z) \in \mathcal{A}_{\delta}^{\ny}$, where
\begin{align*}
	\mathcal{A}_{\delta}^{\ny} \coloneqq \{ (g, t, z) : g \in \mathcal{G}_{\delta}, t \in \{ 2, \dots, \mathcal{T} - \delta \}, g - \delta \le t \le \bar{g}-\delta, z \in \mathcal{I} \}.
\end{align*}

\begin{lemma}
	Suppose that Assumptions \ref{as:staggered}, \ref{as:iid}, \ref{as:overlap}, \ref{as:anticipation-delta}, and \ref{as:parallel-ny-delta} hold.
	Fix arbitrary $(g, t, z) \in \mathcal{A}_{\delta}^{\ny}$.
	\begin{enumerate}[(i)]
		\item Under Assumption \ref{as:parametric-ny-delta}(i), $\CATT_{g,t}(z) = \DR_{g,t,\delta}^{\ny}(z; \beta_{g,t,\delta}^{\ny*}, \pi_{g,t+\delta})$ for all $\pi_{g,t+\delta} \in \Pi_{g,t+\delta}$.
		\item Under Assumption \ref{as:parametric-ny-delta}(ii), $\CATT_{g,t}(z) = \DR_{g,t,\delta}^{\ny}(z; \beta_{g,t,\delta}^{\ny}, \pi_{g,t+\delta}^*)$ for all $\beta_{g,t,\delta}^{\ny} \in \mathscr{B}_{g,t,\delta}^{\ny}$.
	\end{enumerate}
\end{lemma}

The proof is omitted here because it is the same as Lemma \ref{lem:DR}.

\subsection{Estimation and inference}

For notational convenience, we write $m_{g,t,\delta}^{\ny} \coloneqq m_{g,t,\delta}^{\ny}(X; \beta_{g,t,\delta}^{\ny*})$ and $R_{g,t+\delta} \coloneqq R_{g,t+\delta}(W; \pi_{g,t+\delta}^*)$, and define
\begin{align*}
	A_{g,t,\delta}^{\ny}
	\coloneqq \left( \frac{ G_g }{ \mu_G(z) } - \frac{ R_{g,t+\delta} }{ \mu_R(z) } \right) \left( Y_t - Y_{g-\delta-1} - m_{g,t,\delta}^{\ny} \right),
\end{align*}
where $\mu_G(z) = \bE[ G_g \mid Z = z ]$ and $\mu_R(z) = \bE[ R_{g,t+\delta} \mid Z = z ]$.
The definition of $\mu_R(z)$ here slightly differs from that in Section \ref{sec:estimation}, but we use the same notation for ease of exposition; throughout this section, a similar abuse of notation is made for other variables.
The estimand of interest can be written as $\DR_{g,t,\delta}^{\ny}(z) = \mu_A(z) = \bE[ A_{g,t,\delta}^{\ny} \mid Z = z ]$ for $z \in \mathcal{I}$.

The estimation procedure is essentially the same as in Section \ref{subsec:procedure}.
Specifically, we first obtain some parametric estimators $\hat \beta_{g,t,\delta}^{\ny}$ and $\hat \pi_{g,t+\delta}$.
We then compute $\hat m_{i,g,t,\delta}^{\ny} \coloneqq m_{g,t,\delta}^{\ny}(X_i; \hat \beta_{g,t,\delta}^{\ny})$, $\hat R_{i,g,t+\delta} \coloneqq R_{g,t+\delta}(W_i; \hat \pi_{g,t+\delta})$, and 
\begin{align*}
	\hat A_{i,g,t,\delta}^{\ny}
	\coloneqq \left( \frac{ G_g }{ \hat \mu_G(z) } - \frac{ \hat R_{g,t+\delta} }{ \hat \mu_{\hat R}(z) } \right) \left( Y_t - Y_{g-\delta-1} - \hat m_{g,t,\delta}^{\ny} \right).
\end{align*}
Here, $\hat \mu_G(z)$ and $\hat \mu_{\hat R}(z)$ denote the LQR estimators of $\mu_G(z) = \bE[ G_g \mid Z = z ]$ and $\mu_R(z) = \bE[ R_{g,t+\delta} \mid Z = z ]$, respectively, in the current context.
Finally, the estimator for $\CATT_{g,t}(z)$ is given by
\begin{align} \label{eq:DRest-ny-delta}
	\hat{\DR}_{g,t,\delta}^{\ny}(z) 
	\coloneqq \hat \mu_{\hat A}(z),
\end{align}
where $\hat \mu_{\hat A}(z)$ denotes the LQR estimator for $\mu_A(z)$.

We can show that the leading term of the conditional DR estimator in \eqref{eq:DRest-ny-delta} is given by 
\begin{align*}
	& \hat{\DR}_{g,t,\delta}^{\ny}(z) - {\DR}_{g,t,\delta}^{\ny}(z)
	\approx \frac{1}{f_Z(z)} \frac{1}{nh} \sum_{i=1}^n \Psi_{i,h} (B_{i,g,t,\delta}^{\ny} - \mu_B(Z_i)) K\left( \frac{Z_i - z}{h} \right),
\end{align*}
uniformly in $(g,t,z) \in \mathcal{A}_{\delta}^{\ny}$, where
\begin{align*}
	\begin{split}
		B_{i,g,t,\delta}^{\ny} 
		& \coloneqq A_{i,g,t,\delta}^{\ny} + \frac{ \mu_E(z) }{ \mu_R^2(z) } R_{i,g,t+\delta} - \frac{ \mu_F(z) }{ \mu_G^2(z) } G_{i,g}, \\
		E_{i,g,t,\delta}^{\ny} 
		& \coloneqq R_{i,g,t+\delta} ( Y_{i,t} - Y_{i,g-\delta-1} - m_{i,g,t,\delta}^{\ny} ), \\
		F_{i,g,t,\delta}^{\ny} 
		& \coloneqq G_{i,g} ( Y_{i,t} - Y_{i,g-\delta-1} - m_{i,g,t,\delta}^{\ny} ),
	\end{split}
\end{align*}
and, with an abuse of notation, 
\begin{align*}
	\mu_B(z) \coloneqq \bE[ B_{i,g,t,\delta}^{\ny} \mid Z_i = z ],
	\qquad 
	\mu_E(z) \coloneqq \bE[ E_{i,g,t,\delta}^{\ny} \mid Z_i = z ],
	\qquad 
	\mu_F(z) \coloneqq \bE[ F_{i,g,t,\delta}^{\ny} \mid Z_i = z ].
\end{align*}
Based on this asymptotic linear representation, we can perform uniform inference for CATT in the same way as in Section \ref{sec:CATTgt}, except for the minor changes in the definitions of several variables.
In addition, we can prove the validity of the uniform inference method as in Theorems \ref{thm:bias_variance}, \ref{thm:approx_us_distribution}, \ref{thm:mb_linearize}, and \ref{thm:mb_valid}.

\section{Never-Treated Group} \label{sec:never-treated}

While our main analysis considers the never-treated group as the comparison group, the never-treated group (i.e., the units that have never been treated) can also serve as a valid comparison group.
In this section, we discuss the identification, estimation, and uniform inference methods using the never-treated group. 
Throughout this section, we maintain the staggered treatment adoption, random sampling, overlap, and limited treatment anticipation conditions in Assumptions \ref{as:staggered}, \ref{as:iid}, \ref{as:overlap}, and \ref{as:anticipation-delta}.

\subsection{Identification}

Recall that we set $G = \infty$ if the unit has never been treated.
Letting $C \coloneqq \bm{1}\{ G = \infty \}$, we define the OR function and GPS in the current context as follows:
\begin{align*}
	m_{g,t,\delta}^{\nev}(X) 
	& \coloneqq \bE[ Y_t - Y_{g-\delta-1} \mid X, C=1 ],\\
	p_{g}(X) 
	& \coloneqq \bP[G_g = 1 \mid X, G_g + C = 1],
\end{align*}
where $\delta \ge 0$ is as given in Assumption \ref{as:anticipation-delta}.
Let $m_{g,t,\delta}^{\nev}(X; \beta_{g,t,\delta}^{\nev})$ and $p_{g}(X; \pi_{g})$ be the corresponding parametric specifications known up to the finite-dimensional parameters $\beta_{g,t,\delta}^{\nev} \in \mathscr{B}_{g,t,\delta}^{\nev}$ and $\pi_{g} \in \Pi_{g}$.
The conditional DR estimand based on the never-treated group is defined by
\begin{align*}
	{\DR}_{g,t,\delta}^{\nev}(Z; \beta_{g,t,\delta}^{\nev}, \pi_g)
	\coloneqq \bE \left[ \left( \frac{G_g}{\bE[ G_g \mid Z]} - \frac{ R_g(W; \pi_g)  }{\bE \left[ R_g(W; \pi_g) \; \middle| \; Z \right]} \right) \left( Y_t - Y_{g-\delta-1} - m_{g,t,\delta}^{\nev}(X; \beta_{g,t,\delta}^{\nev}) \right) \; \middle| \; Z \right],
\end{align*}
where $W = (Y_1, \dots, Y_\mathcal{T}, X^\top, D_1, \dots, D_\mathcal{T})^\top$ and
\begin{align*}
	R_g(W; \pi_g) 
	\coloneqq \frac{p_g(X; \pi_g) C}{1 - p_g(X; \pi_g)}.
\end{align*}

To develop the analysis using the never-treated group, we replace the parallel trends condition and the parametric assumptions based on the not-yet-treated group in Assumptions \ref{as:parallel-ny} and \ref{as:parametric-ny} (or Assumptions \ref{as:parallel-ny-delta} and \ref{as:parametric-ny-delta}) with the following assumptions.

\begin{assumption}[Conditional Parallel Trends Based on the ``Never-Treated'' Group] \label{as:parallel-nev-delta}
	For each $g \in \mathcal{G}$ and $t \in \{ 2, \dots, \mathcal{T} \}$ such that $t \ge g - \delta$,
	\begin{align*}
		\bE[ Y_t(0) - Y_{t-1}(0) \mid X, G_g = 1 ]
		= \bE[ Y_t(0) - Y_{t-1}(0) \mid X, C = 1 ]
		\quad 
		\text{a.s.}
	\end{align*}
\end{assumption}

\begin{assumption}[Parametric Models for the ``Never-Treated'' Group] \label{as:parametric-nev-delta}
	For each $g \in \mathcal{G}_{\delta}$ and $t \in \{ 2, \dots, \mathcal{T} - \delta \}$ such that $t \ge g -\delta$, either condition is satisfied:
	\begin{enumerate}[(i)]
		\item There exists a unique $\beta_{g,t,\delta}^{\nev*} \in \mathscr{B}_{g,t,\delta}^{\nev}$ such that $m_{g,t,\delta}^{\nev}(X) = m_{g,t,\delta}^{\nev}(X; \beta_{g,t,\delta}^{\nev*})$ a.s.
		\item There exists a unique $\pi_g^* \in \Pi_g$ such that $p_g(X) = p_g(X; \pi_g^*)$ a.s.
	\end{enumerate}
\end{assumption}

The next lemma states that the conditional DR estimand just defined above identifies CATT, whose proof is omitted here because it is the same as Lemma \ref{lem:DR}.
Let
\begin{align*}
	\mathcal{A}_{\delta}^{\nev} 
	\coloneqq \{ (g, t, z) : g \in \mathcal{G}_{\delta}, t \in \{ 2, \dots, \mathcal{T} - \delta \}, t \ge g-\delta, z \in \mathcal{I} \}.
\end{align*}

\begin{lemma} \label{lem:DR-nev}
	Suppose that Assumptions \ref{as:staggered}, \ref{as:iid}, \ref{as:overlap}, \ref{as:anticipation-delta}, and \ref{as:parallel-nev-delta} hold.
	Fix arbitrary $(g, t, z) \in \mathcal{A}_{\delta}^{\nev}$.
	\begin{enumerate}[(i)]
		\item Under Assumption \ref{as:parametric-nev-delta}(i), ${\CATT}_{g,t}(z) = {\DR}_{g,t,\delta}^{\nev}(z; \beta_{g,t,\delta}^{\nev*}, \pi_g)$ for all $\pi_g \in \Pi_g$.
		\item Under Assumption \ref{as:parametric-nev-delta}(ii), ${\CATT}_{g,t}(z) = {\DR}_{g,t,\delta}^{\nev}(z; \beta_{g,t,\delta}^{\nev}, \pi_g^*)$ for all $\beta_{g,t,\delta}^{\nev} \in \mathscr{B}_{g,t,\delta}^{\nev}$.
	\end{enumerate}
\end{lemma}

\subsection{Estimation and inference}

For notational convenience, we write $m_{g,t,\delta}^{\nev} \coloneqq m_{g,t,\delta}^{\nev}(X; \beta_{g,t,\delta}^{\nev*})$ and $R_{g} \coloneqq R_{g}(W; \pi_{g}^*)$, and define
\begin{align*}
	A_{g,t,\delta}^{\nev}
	\coloneqq \left( \frac{ G_g }{ \mu_G(z) } - \frac{ R_{g} }{ \mu_R(z) } \right) \left( Y_t - Y_{g-\delta-1} - m_{g,t,\delta}^{\nev} \right),
\end{align*}
where, with an abuse of notation, $\mu_G(z) = \bE[ G_g \mid Z = z ]$ and $\mu_R(z) = \bE[ R_{g} \mid Z = z ]$.
The definition of $\mu_R(z)$ here slightly differs from that in Section \ref{sec:estimation}, but we use the same notation for ease of exposition; throughout this section, a similar abuse of notation is made for other variables.
The estimand of interest can be written as $\DR_{g,t,\delta}^{\nev}(z) = \mu_A(z) = \bE[ A_{g,t,\delta}^{\nev} \mid Z = z ]$ for $(g,t,z) \in \mathcal{A}_\delta^{\nev}$.

The estimation procedure is almost the same as described in Section \ref{subsec:procedure}.
Specifically, we first obtain some parametric estimators $\hat \beta_{g,t,\delta}^{\nev}$ and $\hat \pi_{g}$.
We then compute $\hat m_{i,g,t,\delta}^{\nev} \coloneqq m_{g,t,\delta}^{\nev}(X_i; \hat \beta_{g,t,\delta}^{\nev})$, $\hat R_{i,g} \coloneqq R_{g}(W_i; \hat \pi_{g})$, and 
\begin{align*}
	\hat A_{i,g,t,\delta}^{\nev}
	\coloneqq \left( \frac{ G_{i,g} }{ \hat \mu_G(z) } - \frac{ \hat R_{i,g} }{ \hat \mu_{\hat R}(z) } \right) \left( Y_{i,t} - Y_{i,g-\delta-1} - \hat m_{i,g,t,\delta}^{\nev} \right).
\end{align*}
Here, with an abuse of notation, $\hat \mu_G(z)$ and $\hat \mu_{\hat R}(z)$ denote the LQR estimators for $\mu_G(z) = \bE[ G_g \mid Z = z ]$ and $\mu_R(z) = \bE[ R_{g} \mid Z = z ]$, respectively.
Finally, the estimator for $\CATT_{g,t}(z)$ is given by the LQR estimator $\hat \mu_{\hat A}(z)$ for $\mu_A(z)$:
\begin{align} \label{eq:DRest-nev}
	\hat{\DR}_{g,t,\delta}^{\nev}(z) 
	\coloneqq \hat \mu_{\hat A}(z).
\end{align}

We can show that the leading term of the conditional DR estimator in \eqref{eq:DRest-nev} is given by 
\begin{align*}
	& \hat{\DR}_{g,t,\delta}^{\nev}(z) - {\DR}_{g,t,\delta}^{\nev}(z)
	\approx \frac{1}{f_Z(z)} \frac{1}{nh} \sum_{i=1}^n  \Psi_{i,p,h} (B_{i,g,t,\delta}^{\nev} - \mu_B(Z_i)) K\left( \frac{Z_i - z}{h} \right),
\end{align*}
uniformly in $(g,t,z) \in \mathcal{A}_{\delta}^{\nev}$, where
\begin{align*}
	\begin{split}
		B_{i,g,t,\delta}^{\nev} 
		& \coloneqq A_{i,g,t,\delta}^{\nev} + \frac{ \mu_E(z) }{ \mu_R^2(z) } R_{i,g} - \frac{ \mu_F(z) }{ \mu_G^2(z) } G_{i,g}, \\
		E_{i,g,t,\delta}^{\nev} 
		& \coloneqq R_{i,g} ( Y_{i,t} - Y_{i,g-\delta-1} - m_{i,g,t,\delta}^{\nev} ), \\
		F_{i,g,t,\delta}^{\nev} 
		& \coloneqq G_{i,g} ( Y_{i,t} - Y_{i,g-\delta-1} - m_{i,g,t,\delta}^{\nev} ),
	\end{split}
\end{align*}
and, with an abuse of notation, 
\begin{align*}
	\mu_B(z) \coloneqq \bE[ B_{i,g,t,\delta}^{\nev} \mid Z_i = z ],
	\qquad 
	\mu_E(z) \coloneqq \bE[ E_{i,g,t,\delta}^{\nev} \mid Z_i = z ],
	\qquad 
	\mu_F(z) \coloneqq \bE[ F_{i,g,t,\delta}^{\nev} \mid Z_i = z ].
\end{align*}
Based on this asymptotic linear representation, we can perform uniform inference for CATT in the same way as in Section \ref{sec:CATTgt}, except for the minor changes in the definitions of several variables.
In addition, we can prove the validity of the uniform inference method as in Theorems \ref{thm:bias_variance}, \ref{thm:approx_us_distribution}, \ref{thm:mb_linearize}, and \ref{thm:mb_valid}.

\section{Additional Discussions for Summary Parameters} \label{sec:supp-summary}

Recall that the aggregated parameter of general form in \eqref{eq:summary} is given by
\begin{align} \label{eq:supp-summary}
	\theta(z) 
	= \sum_{g \in \mathcal{G}} \sum_{t=2}^{\mathcal{T}} w_{g,t}(z) \cdot {\CATT}_{g,t}(z),
\end{align}
where $w_{g,t}(z)$ is a known or estimable weighting function.
In this section, we consider a variety of useful summary parameters that can be written in this form, and discuss how to construct the uniform confidence bands for them.
For presentation purposes, we focus on the analysis using the not-yet-treated group under the no-anticipation condition in Assumption \ref{as:anticipation-ny}.
To make this section self-contained, some of the discussion below may overlap with the discussion in the main text.

Throughout this section, in slightly different notation from the main text, we write $\mu_{G_g}(z) = \bE[ G_g \mid Z = z ]$ to emphasize that this quantity depends on $g$.

\subsection{Examples of summary parameters} \label{subsec:examples}

There are many candidates for useful summary parameters that can be written in the form \eqref{eq:supp-summary}, but we focus on the following parameters for empirical relevance:
(i) the ``event-study-type'' conditional average treatment effect;
(ii) the group-specific conditional average treatment effect;
(iii) the calendar-time conditional average treatment effect; and
(iv) the simple weighted conditional average treatment effect.

\subsubsection{The event-study-type conditional average treatment effect}

Let $e = t - g \ge 0$ denote elapsed treatment time.
To examine the treatment effect heterogeneity with respect to elapsed treatment time $e$ and covariate value $z$, we consider the event-study-type conditional average treatment effect:
\begin{align*}
	\theta_{\es}(e, z)
	& \coloneqq \bE[ Y_{G+e}(G) - Y_{G+e}(0) \mid G + e < \bar g, Z = z ] \\
	& = \bE \Big[ \bE[ Y_{G+e}(G) - Y_{G+e}(0) \mid G, G + e < \bar g, Z = z ]  \; \Big| \; G + e < \bar g, Z = z \Big]\\
	& = \sum_{g \in \mathcal{G}} \bm{1}\{ g + e < \bar g \} \cdot \Pr(G = g \mid G + e < \bar g, Z = z) \cdot {\CATT}_{g, g + e}(z).
\end{align*}
This is the conditional counterpart of the event-study-type summary parameter in equation (3.4) of \citet{callaway2021difference}.
To further rewrite this parameter, observe that
\begin{align*}
	\bm{1}\{ g + e < \bar g \} \cdot \Pr(G = g \mid G + e < \bar{g}, Z = z)
	& = \bm{1}\{ g < \bar g - e \} \cdot \frac{ \Pr( G = g, G < \bar g - e \mid Z = z ) }{ \Pr( G < \bar g - e \mid Z = z ) } \\
	& = \bm{1}\{ g < \bar g - e \} \cdot \frac{ \Pr( G = g \mid Z = z ) }{ \sum_{g'=2}^{\bar g - 1 - e} \Pr( G = g' \mid Z = z ) } \\
	& = \bm{1}\{ g < \bar g - e \} \cdot \frac{ \mu_{G_g}(z) }{ \sum_{g'=2}^{\bar g - 1 - e} \mu_{G_{g'}}(z) }.
\end{align*}
Thus, we can write
\begin{align} \label{eq:supp-es}
	\theta_{\es}(e, z)
	= \sum_{g \in \mathcal{G}} \sum_{t=2}^{\mathcal{T}} w_{g,t}^{\es}(e, z) \cdot {\CATT}_{g,t}(z),
\end{align}
where the weighting function is estimable and given by
\begin{align*}
	w_{g,t}^{\es}(e, z)
	\coloneqq \bm{1}\{ g < \bar g - e, t = g + e \} \cdot \frac{ \mu_{G_g}(z) }{ \sum_{g'=2}^{\bar g - 1 - e} \mu_{G_{g'}}(z) }.
\end{align*}

\subsubsection{The group-specific conditional average treatment effect}

To examine the treatment effect heterogeneity with respect to group $g$ and covariate value $z$, we consider the group-specific conditional average treatment effect:
\begin{align*}
	\theta_{\sel}(g', z)
	\coloneqq \frac{ \sum_{t=2}^{\mathcal{T}} \bm{1}\{ g' \le t < \bar g \} \cdot {\CATT}_{g',t}(z) }{ \sum_{t'=2}^{\mathcal{T}} \bm{1}\{ g' \le t' < \bar g \} },
\end{align*}
where the subscript ``$\sel$'' comes from the fact that the groups are typically determined by ``selective'' treatment timings.
This is the conditional counterpart of the group-specific summary parameter in equation (3.7) of \citet{callaway2021difference}.
It is easy to see that this parameter can be rewritten as
\begin{align} \label{eq:supp-sel}
	\theta_{\sel}(g', z)
	= \sum_{g\in\mathcal{G}} \sum_{t=2}^{\mathcal{T}} w_{g,t}^{\sel}(g', z) \cdot {\CATT}_{g,t}(z),
\end{align}
where the weighting function is known and given by
\begin{align*}
	w_{g,t}^{\sel}(g', z) 
	\coloneqq \frac{ \bm{1}\{ g = g', g' \le t < \bar g \} }{ \sum_{t'=2}^{\mathcal{T}} \bm{1}\{ g' \le t' < \bar g \} }.
\end{align*}

\subsubsection{The calendar-time conditional average treatment effect}

To examine the treatment effect heterogeneity with respect to calendar time $t$ and covariate value $z$, we consider the calendar-time conditional average treatment effect:
\begin{align*}
	\theta_{\cc}(t', z)
	& \coloneqq \bE[ Y_{t'}(G) - Y_{t'}(0) \mid G \le t', Z = z ] \\
	& = \bE \Big[  \bE[ Y_{t'}(G) - Y_{t'}(0) \mid G, G \le t', Z = z ] \; \Big| \; G \le t', Z = z \Big] \\
	& = \sum_{g \in \mathcal{G}} \bm{1}\{ g \le t' \} \cdot \Pr(G = g \mid G \le t', Z = z) \cdot {\CATT}_{g,t'}(z).
\end{align*}
This is the conditional counterpart of the calendar-time summary parameter in equation (3.8) of \citet{callaway2021difference}.
To further rewrite this parameter, observe that
\begin{align*}
	\bm{1}\{ g \le t' \} \cdot \Pr(G = g \mid G \le t', Z = z)
	& = \bm{1}\{ g \le t' \} \cdot \frac{ \Pr( G = g, G \le t' \mid Z = z ) }{ \Pr( G \le t' \mid Z = z ) } \\
	& = \bm{1}\{ g \le t' \} \cdot \frac{ \Pr( G = g \mid Z = z ) }{ \sum_{g'=2}^{t'} \Pr( G = g' \mid Z = z ) } \\
	& = \bm{1}\{ g \le t' \} \cdot \frac{ \mu_{G_g}(z) }{ \sum_{g'=2}^{t'} \mu_{G_{g'}}(z) }.
\end{align*}
Thus, we can write
\begin{align} \label{eq:supp-calendar}
	\theta_{\cc}(t', z)
	= \sum_{g \in \mathcal{G}} \sum_{t=2}^{\mathcal{T}} w_{g,t}^{\cc}(t', z) \cdot {\CATT}_{g,t}(z),
\end{align}
where the weighting function is estimable and given by
\begin{align*}
	w_{g,t}^{\cc}(t', z)
	\coloneqq \bm{1}\{ g \le t', t = t' \} \cdot \frac{ \mu_{G_g}(z) }{ \sum_{g'=2}^{t'} \mu_{G_{g'}}(z) }.
\end{align*}

\subsubsection{The simple weighted conditional average treatment effect}

By simply weighted averaging all identified CATT's, we consider the simple weighted conditional average treatment effect:
\begin{align*}
	\theta_{\W}^{\OO}(z)
	& \coloneqq \frac{1}{\kappa(z)} \sum_{t=2}^{\mathcal{T}} \bE[ Y_t(G) - Y_t(0) \mid G < \bar g, Z = z ] \\
	& = \frac{1}{\kappa(z)} \sum_{g \in \mathcal{G}} \sum_{t=2}^{\mathcal{T}} \bm{1}\{ g \le t < \bar g \} \cdot \Pr(G = g \mid G < \bar g, Z = z) \cdot {\CATT}_{g,t}(z),
\end{align*}
where
\begin{align*}
	\kappa(z)
	& \coloneqq \sum_{t=2}^{\mathcal{T}} \bE[ \bm{1}\{ G \le t < \bar g \} \mid G < \bar g, Z = z] \\
	& = \sum_{g \in \mathcal{G}} \sum_{t=2}^{\mathcal{T}} \bm{1}\{ g \le t < \bar g \} \cdot \Pr(G = g \mid G < \bar g, Z = z).
\end{align*}
The subscript ``$\W$'' and the superscript ``$\OO$'' come from the fact that this summary parameter aggregates $\CATT_{g,t}(z)$'s with simple ``weights'' into a single ``overall'' effect.
This is the conditional counterpart of the overall treatment effect parameter in equation (3.10) of \citet{callaway2021difference}.
To further rewrite this parameter, observe that
\begin{align*}
	\bm{1}\{ g \le t < \bar g \} \cdot \Pr(G = g \mid G < \bar{g}, Z = z)
	& = \bm{1}\{ g \le t < \bar g \} \cdot \frac{ \Pr( G = g, G < \bar g \mid Z = z ) }{ \Pr( G < \bar g \mid Z = z ) } \\
	& = \bm{1}\{ g \le t < \bar g \} \cdot \frac{ \Pr( G = g \mid Z = z ) }{ \sum_{g'=2}^{\bar g - 1} \Pr( G = g' \mid Z = z ) } \\
	& = \bm{1}\{ g \le t < \bar g \} \cdot \frac{ \mu_{G_g}(z) }{ \sum_{g'=2}^{\bar g - 1} \mu_{G_{g'}}(z) }.
\end{align*}
Thus, we can write
\begin{align} \label{eq:supp-overall}
	\theta_{\W}^{\OO}(z)
	= \sum_{g \in \mathcal{G}} \sum_{t=2}^{\mathcal{T}} w_{g,t}^{\OW}(z) \cdot {\CATT}_{g,t}(z),
\end{align}
where the weighting function is estimable and given by 
\begin{align*}
	w_{g,t}^{\OW}(z)
	\coloneqq \frac{1}{\kappa(z)} \cdot \bm{1}\{ g \le t < \bar g \} \cdot \frac{ \mu_{G_g}(z) }{ \sum_{g'=2}^{\bar g - 1} \mu_{G_{g'}}(z) }
\end{align*}
with
\begin{align*}
	\kappa(z)
	= \sum_{g \in \mathcal{G}} \sum_{t=2}^{\mathcal{T}} \bm{1}\{ g \le t < \bar g \} \cdot \frac{ \mu_{G_g}(z) }{ \sum_{g'=2}^{\bar g - 1} \mu_{G_{g'}}(z) }.
\end{align*}

\subsection{Inference for the aggregated parameter}

We discuss the uniform inference method for the general aggregated parameter $\theta(z)$ in \eqref{eq:supp-summary}.
Since the identifiability of $\theta(z)$ is straightforward from the identification result for ${\CATT}_{g,t}(z)$ in Lemma \ref{lem:DR}, we proceed directly to the discussion of the estimation and uniform inference methods.

\subsubsection{Estimation}

We estimate the aggregated parameter $\theta(z)$ in \eqref{eq:supp-summary} by
\begin{align*}
	\hat{\theta}(z) 
	\coloneqq \sum_{g \in \mathcal{G}} \sum_{t=2}^{\mathcal{T}} \hat{w}_{g,t}(z) \cdot \hat{\DR}_{g,t}(z), 
\end{align*}
where $\hat{\DR}_{g,t}(z)$ is the conditional DR estimator defined in \eqref{eq:DRest-ny}, and $\hat{w}_{g,t}(z) = w_{g,t}(z)$ if the weighting function $w_{g,t}(z)$ is known, otherwise $\hat{w}_{g,t}(z)$ is a nonparametric estimator constructed with certain LQR estimation.
For example, the weighting function $w_{g,t}^{\es}(e, z)$ for the event-study-type summary parameter in \eqref{eq:supp-es} can be estimated by
\begin{align*}
	\hat{w}_{g,t}^{\es}(e, z)
	\coloneqq \bm{1}\{ g < \bar g - e, t = g + e \} \cdot \frac{ \hat{\mu}_{G_g}(z) }{ \sum_{g'=2}^{\bar g - 1 - e} \hat{\mu}_{G_{g'}}(z) },
\end{align*}
where $\hat{\mu}_{G_g}(z)$ denotes the LQR estimator for $\mu_{G_g}(z) = \bE[ G_g \mid Z = z ]$.

\subsubsection{Asymptotic linear representation}

To derive the asymptotic linear representation for the aggregated estimator $\hat{\theta}(z)$, the next assumption requires that the weighting function estimator $\hat{w}_{g,t}(z)$ exhibits the same form of asymptotic linearity as in Theorem \ref{thm:bias_variance}.
Although this is a high-level condition, we will demonstrate in the next subsection that it is satisfied with the summary parameters discussed in the previous subsection.

\begin{assumption} \label{as:supp-weight}
	Either condition is satisfied for all $(g, t, z) \in \mathcal{A}$:
	\begin{enumerate}[(i)]
		\item $w_{g,t}(z)$ is a known weighting function and $\hat{w}_{g,t}(z) = w_{g,t}(z)$.
		\item The weighting function estimator $\hat{w}_{g,t}(z)$ is consistent for $w_{g,t}(z)$ and exhibits the following asymptotic linear representation:
		\begin{align*}
			& \hat{w}_{g,t}(z) - w_{g,t}(z) \\
			& = \frac{1}{ f_Z(z) } \frac{1}{nh} \sum_{i=1}^n \Psi_{i,h} (\xi_{i,g,t} - \mu_{\xi}(Z_i)) K \left( \frac{Z_i - z}{h} \right) + \mathrm{Bias} \left[ \hat{w}_{g,t}(z) \; \middle| \; \bm{Z} \right] + o_{\bP}(h^4) + o_{\bP}\left( \sqrt{\frac{\log n}{nh}} \right),
		\end{align*}
		where $\xi_{i,g,t}$ is an estimable random variable, $\mu_{\xi}(Z_i) = \bE[ \xi_{i,g,t} \mid Z_i ]$, and
		\begin{align*}
			\mathrm{Bias}\left[ \hat{w}_{g,t}(z) \;\middle|\; \bm{Z} \right] 
			& = h^4 \mathcal{B}_{w,g,t}(z) + o_{\bP}(h^4) = o_{\bP}\left( \frac{1}{nh} \right), \\
			\mathcal{B}_{w,g,t}(z)
			& \coloneqq \frac{1}{24 f_Z(z)} \left( 2 \mu_{\xi}^{(3)}(z) f_Z^{(1)}(z) + \mu_{\xi}^{(4)}(z) f_Z(z) \right) \left( \frac{I_{4,K}^2 - I_{2,K} I_{6,K} }{ I_{4,K} - I_{2,K}^2 } \right),
		\end{align*}
		and the convergence rates of the remainder terms hold uniformly in $(g, t, z) \in \mathcal{A}$.
	\end{enumerate}
\end{assumption}

The next theorem formalizes the asymptotic linear representation for $\hat{\theta}(z) - \theta(z)$, whose proof is given at the end of this section.

\begin{theorem} \label{thm:supp-summary}
	Suppose that Assumptions \ref{as:staggered}--\ref{as:asymptotic1} and \ref{as:supp-weight} hold.
	When $n \to \infty$, we have
	\begin{align*}
		\hat{\theta}(z) - \theta(z) 
		= \frac{1}{ f_Z(z) } \frac{1}{nh} \sum_{i=1}^n \Psi_{i,h} (J_i - \mu_J(Z_i)) K \left( \frac{Z_i - z}{h} \right) + \mathrm{Bias} \left[ \hat \theta(z) \; \middle| \; \bm{Z} \right] + o_{\bP}(h^4) + o_{\bP}\left( \sqrt{\frac{\log n}{nh}} \right),
	\end{align*}
	where
	\begin{align*}
		J_i \coloneqq 
		\begin{cases}
			\displaystyle \sum_{g \in \mathcal{G}} \sum_{t=2}^{\mathcal{T}} w_{g,t}(z) \cdot B_{i,g,t} & \text{under Assumption \ref{as:supp-weight}(i)}, \\
			\displaystyle \sum_{g \in \mathcal{G}} \sum_{t=2}^{\mathcal{T}} \Big( w_{g,t}(z) \cdot B_{i,g,t} + {\DR}_{g,t}(z) \cdot \xi_{i,g,t} \Big) & \text{under Assumption \ref{as:supp-weight}(ii)},
		\end{cases}
	\end{align*}
	and
	\begin{align*}
		\mathrm{Bias}\left[ \hat \theta(z) \;\middle|\; \bm{Z} \right] 
		& = h^4 \mathcal{B}_{\theta}(z) + o_{\bP}(h^4) = o_{\bP}\left( \frac{1}{nh} \right), \\
		\mathrm{Var}\left[ \hat \theta(z) \;\middle|\; \bm{Z} \right]
		& = \frac{1}{nh} \mathcal{V}_{\theta}(z) + o_{\bP}\left( \frac{1}{nh} \right),
	\end{align*}
	and
	\begin{align*}
		\mathcal{B}_{\theta}(z)
		& \coloneqq \frac{1}{24 f_Z(z)}\left(2\mu_J^{(3)}(z)f_Z^{(1)}(z) + \mu_J^{(4)}(z)f_Z(z)\right)  \left( \frac{I_{4,K}^2 - I_{2,K}I_{6,K}}{I_{4,K} - I_{2,K}^2} \right), \\
		\mathcal{V}_{\theta}(z)
		& \coloneqq \frac{\sigma_J^2(z)}{f_Z(z)} \left( \frac{ I_{4,K}^2 I_{0,K^2} - 2I_{2,K}I_{4,K}I_{2,K^2} + I_{2,K}^2I_{4,K^2} }{ ( I_{4,K} - I_{2,K}^2 )^2 } \right),
	\end{align*}
	with denoting $\mu_J(z) \coloneqq \bE[ J_i \mid Z_i = z ]$, $\sigma_J^2(z) \coloneqq \Var[ J_i \mid Z_i = z]$, and
	\begin{align*}
		\mu_J^{(\nu)}(z) \coloneqq 
		\begin{cases}
			\displaystyle \sum_{g \in \mathcal{G}} \sum_{t=2}^{\mathcal{T}} w_{g,t}(z) \cdot \mu_{B}^{(\nu)}(z) & \text{under Assumption \ref{as:supp-weight}(i)}, \\
			\displaystyle \sum_{g \in \mathcal{G}} \sum_{t=2}^{\mathcal{T}} \Big( w_{g,t}(z) \cdot \mu_{B}^{(\nu)}(z) + {\DR}_{g,t}(z) \cdot \mu_{\xi}^{(\nu)}(z) \Big) & \text{under Assumption \ref{as:supp-weight}(ii)}.
		\end{cases}
	\end{align*}
\end{theorem}

For presentation purposes, the discussion in the rest of this subsection focuses on the analysis under Assumption \ref{as:supp-weight}(ii), but essentially the same arguments can apply to the case under Assumption \ref{as:supp-weight}(i).

\subsubsection{Standard error}

We can compute the standard error of the aggregated estimator $\hat{\theta}(z)$ in the same way as in Section \ref{subsec:se}.
Specifically, letting $\hat{\xi}_{i,g,t}$ be a ``consistent'' estimator for $\xi_{i,g,t}$, we compute
\begin{align*}
	\hat{J}_i
	\coloneqq \sum_{g \in \mathcal{G}} \sum_{t=2}^{\mathcal{T}} \Big( \hat{w}_{g,t}(z) \cdot \hat{B}_{i,g,t} + \hat{\DR}_{g,t}(z) \cdot \hat{\xi}_{i,g,t} \Big),
\end{align*}
where $\hat{B}_{i,g,t}$ is defined in Section \ref{subsec:se}.
Next, we estimate the conditional variance $\sigma_J^2(z)$ by the same LLR estimation as in Section \ref{subsec:se}, except that we replace $\hat{U}_{i,g,t}$ with $\hat{U}_i^J \coloneqq \hat{J}_i - \hat \mu_{\hat{J}}(Z_i)$, which leads to the conditional variance estimator $\hat{\sigma}_{\hat J}^2(z)$.
Then, we compute
\begin{align*}
	\hat{\mathcal{V}}_{\theta}(z)
	\coloneqq \frac{\hat{\sigma}_{\hat J}^2(z)}{\hat{f}_Z(z)} \left( \frac{ I_{4,K}^2 I_{0,K^2} - 2I_{2,K}I_{4,K}I_{2,K^2} + I_{2,K}^2I_{4,K^2} }{ ( I_{4,K} - I_{2,K}^2 )^2 } \right).
\end{align*}
Finally, the standard error of $\hat{\theta}(z)$ is given by
\begin{align*}
	\hat{{\SE}}_{\theta}(z)
	\coloneqq \left( \frac{1}{nh} \hat{\mathcal{V}}_{\theta}(z) \right)^{1/2}.
\end{align*}

\subsubsection{Critical value}

\paragraph{Analytical method.}

As shown in Theorem \ref{thm:supp-summary}, the aggregated estimator $\hat{\theta}(z)$ has the same form of asymptotic linearity as for the conditional DR estimator $\hat{\DR}_{g,t}(z)$ in Theorem \ref{thm:bias_variance}.
As a result, we can construct the uniform confidence band for the aggregated parameter $\theta(z)$ by the same analytical method for CATT as discussed in Section \ref{subsec:UCB}.
Specifically, the $(1 - \alpha)$ uniform confidence band for $\theta(z)$ is given by $\hat{\mathcal{C}}_{\theta} \coloneqq \{ \hat{\mathcal{C}}_{\theta}(z) \}$, where
\begin{align*}
	\hat{\mathcal{C}}_{\theta}(z)
	\coloneqq 
	\left[ 
	\hat{\theta}(z) - \hat c(1 - \alpha) \cdot \hat{{\SE}}_{\theta}(z),
	\qquad 
	\hat{\theta}(z) + \hat c(1 - \alpha) \cdot \hat{{\SE}}_{\theta}(z)
	\right],
\end{align*}
where $\hat c(1 - \alpha)$ is defined in \eqref{eq:analytical-critical}.
Note that $\hat c(1 - \alpha)$ must not depend on $z$ and another variable (if any) specific to the summary parameter of interest (e.g., elapsed treatment time $e$ for the event-study-type parameter).
This uniform confidence band can be justified in exactly the same way as in Theorem \ref{thm:approx_us_distribution}, in conjunction with the asymptotic linear representation in Theorem \ref{thm:supp-summary}.

\paragraph{Multiplier bootstrapping.}

Since the aggregated parameter $\theta(z)$ is not characterized as a solution of an optimization problem, we cannot use the same type of weighted bootstrapping as CATT, which randomizes the objective function of the LQR estimation with bootstrap weights.
However, as implied by Theorem \ref{thm:mb_linearize} in the case of CATT, this type of bootstrapping is asymptotically equivalent to the multiplier bootstrap inference that randomizes the influence function with bootstrap weights.
Based on this insight, we consider the multiplier bootstrap inference for the aggregated parameter $\theta(z)$ by randomizing the influence function for the aggregated estimator $\hat{\theta}(z)$.
To be specific, recall the asymptotic linear representation for $\hat \theta(z) - \theta(z)$ in Theorem \ref{thm:supp-summary}, and let
\begin{align*}
	\hat{\theta}^{\star,b}(z) 
	\coloneqq \hat{\theta}(z) + \frac{1}{ \hat f_Z(z) } \frac{1}{nh} \sum_{i=1}^n ( V_i^{\star,b} - 1 ) \Psi_{i,h} \left( \hat J_i - \hat \mu_{\hat J}(Z_i) \right) K\left( \frac{Z_i - z}{h} \right),
\end{align*}
where $\{ V_i^{\star,b} \}$ is the same set of bootstrap weights as in Section \ref{subsec:UCB}.
Then, the bootstrap counterpart of the studentized statistic is given by
\begin{align*}
	M_{\theta}^{\star,b} 
	\coloneqq \sup \frac{ | \hat{\theta}^{\star,b}(z) - \hat{\theta}(z) | }{ \hat{{\SE}}_{\theta}(z) },
\end{align*}
where we take the supremum over $z \in \mathcal{I}$ and another variable (if any) specific to the chosen summary parameter (e.g., elapsed treatment time $e$ for the event-study-type parameter).
Letting $\tilde{c}_{\theta}(1 - \alpha)$ denote the empirical $(1 - \alpha)$ quantile of $\{ M_{\theta}^{\star,b} \}_{b=1}^B$, the $(1 - \alpha)$ uniform confidence band for $\theta(z)$ is $\tilde{\mathcal{C}}_{\theta} \coloneqq \{ \tilde{\mathcal{C}}_{\theta}(z) \}$, where
\begin{align*}
	\tilde{\mathcal{C}}_{\theta}(z)
	\coloneqq \left[ \hat{\theta}(z) - \tilde{c}_{\theta}(1 - \alpha) \cdot \hat{{\SE}}_{\theta}(z), \quad \hat{\theta}(z) + \tilde{c}_{\theta}(1 - \alpha) \cdot \hat{{\SE}}_{\theta}(z) \right].
\end{align*}
We can prove the validity of this uniform confidence band in exactly the same way as in Theorem \ref{thm:mb_valid}.

\subsubsection{Bandwidth selection}

Similar to the bandwidth choice based on the insight of the simple RBC inference in the case of CATT, we propose to compute the aggregated estimator $\hat \theta(z)$ via the LQR estimation by using the IMSE-optimal bandwidth for the LLR estimator of $\theta(z)$.
To be specific, we write the aggregated estimator computed with the LLR estimation as
\begin{align*}
	\hat{\theta}^{\mathtt{LL}}(z) 
	\coloneqq \sum_{g \in \mathcal{G}} \sum_{t=2}^{\mathcal{T}} \hat{w}_{g,t}^{\mathtt{LL}}(z) \cdot \hat{\DR}_{g,t}^{\mathtt{LL}}(z), 
\end{align*}
where the estimators $\hat{w}_{g,t}^{\mathtt{LL}}(z)$ and $\hat{\DR}_{g,t}^{\mathtt{LL}}(z)$ are computed with the LLR estimation.
By the same arguments as in \eqref{eq:bias_var_nev_1}, we can show that the asymptotic bias and variance of $\hat{\theta}^{\mathtt{LL}}(z)$ are given by
\begin{align*}
	\mathrm{Bias}\left[ \hat \theta^{\mathtt{LL}}(z) \;\middle|\; \bm{Z} \right] 
	& \approx h^2 \frac{I_{2,K}}{2}\mu^{(2)}_J(z) , \quad \mathrm{Var}\left[ \hat \theta^{\mathtt{LL}}(z) \;\middle|\; \bm{Z} \right] \approx \frac{I_{0,K^2}}{nh} \frac{\sigma^2_J(z)}{f_Z(z)}.
\end{align*}
Thus, the IMSE over $z \in \mathcal{I}$ can be written as
\begin{align*}
	h^4 \frac{I_{2,K}^2}{4} \int_{z \in \mathcal{I}} \left[  \mu_J^{(2)}(z)  \right]^2 dz + \frac{I_{0,K^2}}{nh} \int_{z \in \mathcal{I}}\frac{ \sigma_J^2(z) }{ f_Z(z) }  dz,
\end{align*}
and the infeasible IMSE-optimal bandwidth for the LLR estimator is given by
\begin{align*}
	h_{\theta}^{\mathtt{LL}}
	\coloneqq \left( \frac{ I_{0,K^2}\int_{z \in \mathcal{I}} f_Z^{-1}(z)\sigma
		_J^2(z)dz }{I_{2,K}^2 \int_{z \in \mathcal{I}} [\mu_J^{(2)}(z)]^2 dz } \right)^{1/5} n^{-1/5}.
\end{align*}
In practice, we can easily obtain the feasible IMSE-optimal bandwidth, say $\hat{h}_{\theta}^{\mathtt{LL}}$, by estimating the unknown quantities in the above equation.
If the chosen summary parameter has another variable in addition to $z$, such as elapsed treatment time $e$ for the event-study-type parameter, we take the minimum of the bandwidths over its values to obtain the common bandwidth, as in the case of CATT in Section \ref{subsec:bandwidth}.

\subsection{Examples of summary parameters (continued)}

Using the uniform inference method for the general aggregated parameter $\theta(z)$ developed in the previous subsection, we can easily construct the uniform confidence bands for the summary parameters introduced in Section \ref{subsec:examples}.
For completeness, in this subsection, we present the definitions of the key variables $J_i$ and $\xi_{i,g,t}$ in Theorem \ref{thm:supp-summary} and Assumption \ref{as:supp-weight} for each summary parameter.

\subsubsection{The event-study-type conditional average treatment effect}

The event-study-type conditional average treatment effect in \eqref{eq:supp-es} can be estimated by
\begin{align*}
	\hat \theta_{\es}(e, z)
	\coloneqq \sum_{g \in \mathcal{G}} \sum_{t=2}^{\mathcal{T}} \hat{w}_{g,t}^{\es}(e, z) \cdot \hat{\DR}_{g,t}(z),
\end{align*}
where the weighting function estimator is given by
\begin{align*}
	\hat{w}_{g,t}^{\es}(e, z)
	\coloneqq \bm{1}\{ g < \bar g - e, t = g + e \} \cdot \frac{ \hat{\mu}_{G_g}(z) }{ \sum_{g'=2}^{\bar g - 1 - e} \hat{\mu}_{G_{g'}}(z) }
\end{align*}
with denoting the LQR estimator for $\mu_{G_g}(z)$ as $\hat{\mu}_{G_g}(z)$.

The event-study-type summary estimator exhibits the asymptotic linear representation in Theorem \ref{thm:supp-summary} when replacing the variable $J_i$ with
\begin{align*}
	J_i^{\es}
	\coloneqq \sum_{g \in \mathcal{G}} \sum_{t=2}^{\mathcal{T}} \Big( w_{g,t}^{\es}(e, z) \cdot B_{i,g,t} + {\DR}_{g,t}(z) \cdot \xi_{i,g,t}^{\es} \Big),
\end{align*}
where
\begin{align*}
	\xi_{i,g,t}^{\es}
	\coloneqq \bm{1}\{ g < \bar g - e, t = g + e \} \cdot \left( \frac{ G_{i,g} }{ \sum_{g' = 2}^{\bar g - 1 - e} \mu_{G_{g'}}(z) } - \frac{ \mu_{G_g}(z) }{ \left( \sum_{g' = 2}^{\bar g - 1 - e} \mu_{G_{g'}}(z) \right)^2 } \sum_{g' = 2}^{\bar g - 1 - e} G_{i,g'} \right).
\end{align*}
To see this result, it suffices to show that the weighting function estimator $\hat{w}_{g,t}^{\es}(z)$ exhibits the asymptotic linear representation in Assumption \ref{as:supp-weight}(ii) when replacing the variable $\xi_{i,g,t}$ with $\xi_{i,g,t}^{\es}$.
By simple algebra, observe that
\begin{align*}
	& \hat{w}_{g,t}^{\es}(e, z) - w_{g,t}^{\es}(e, z) \\
	& = \bm{1}\{ g < \bar g - e, t = g + e \} \cdot \left( \frac{ \hat \mu_{G_g}(z) }{ \sum_{g'=2}^{\bar g - 1 - e} \hat \mu_{G_{g'}}(z) } - \frac{ \mu_{G_g}(z) }{ \sum_{g'=2}^{\bar g - 1 - e} \mu_{G_{g'}}(z) } \right) \\
	& = \bm{1}\{ g < \bar g - e, t = g + e \} \cdot \left( \frac{ \hat{\mu}_{G_g}(z) }{ \sum_{g'=2}^{\bar g - 1 - e} \hat{\mu}_{G_{g'}}(z) } - \frac{ \hat{\mu}_{G_g}(z) }{ \sum_{g'=2}^{\bar g - 1 - e} \mu_{G_{g'}}(z) } + \frac{ \hat{\mu}_{G_g}(z) }{ \sum_{g'=2}^{\bar g - 1 - e} \mu_{G_{g'}}(z) } - \frac{ \mu_{G_g}(z) }{ \sum_{g'=2}^{\bar g - 1 - e} \mu_{G_{g'}}(z) } \right) \\
	& = \bm{1}\{ g < \bar g - e, t = g + e \} \cdot \left( \frac{ \hat{\mu}_{G_g}(z) - \mu_{G_g}(z) }{ \sum_{g'=2}^{\bar g - 1 - e} \mu_{G_{g'}}(z) } - \frac{ \hat{\mu}_{G_g}(z) \cdot \sum_{g'=2}^{\bar g - 1 - e}[ \hat{\mu}_{G_{g'}}(z) - \mu_{G_{g'}}(z) ] }{ \left( \sum_{g'=2}^{\bar g - 1 - e} \hat{\mu}_{G_{g'}}(z) \right) \cdot \left( \sum_{g'=2}^{\bar g - 1 - e} \mu_{G_{g'}}(z) \right)} \right).
\end{align*}
Here, in the same manner as the proof for the asymptotic linear representation for $\hat{\mu}_A(z)$ in Section \ref{subsubsec:proof-linear}, we can show that
\begin{align} \label{eq:Glinear}
	\begin{split}
		& \hat{\mu}_{G_g}(z) - \mu_{G_g}(z) \\
		& = \frac{1}{ f_Z(z) } \frac{1}{nh} \sum_{i=1}^n \Psi_{i,h} ( G_{i,g} - \mu_{G_g}(Z_i) ) K_{i,h} + \mathrm{Bias} \left[ \hat{\mu}_{G_g}(z) \; \middle| \; \bm{Z} \right] + o_{\bP}(h^4) + o_{\bP}\left( \sqrt{\frac{\log n}{nh}}\right) \\
		& = O_{\bP}\left( \sqrt{\frac{\log n}{nh}}\right).
	\end{split}
\end{align}
Combining these two equations leads to
\begin{align*}
	& \hat{w}_{g,t}^{\es}(e, z) - w_{g,t}^{\es}(e, z) \\
	& = \bm{1}\{ g < \bar g - e, t = g + e \} \cdot \Bigg\{ \frac{1}{ \sum_{g' = 2}^{\bar g - 1 - e} \mu_{G_{g'}}(z) } \left( \frac{1}{ f_Z(z) } \frac{1}{nh} \sum_{i=1}^n \Psi_{i,h} ( G_{i_g} - \mu_{G_g}(Z_i) ) K_{i,h} + \mathrm{Bias} \left[ \hat \mu_{G_g}(z) \; \middle| \; \bm{Z} \right] \right) \\
	& \quad - \frac{ \mu_{G_g}(z) }{ \left( \sum_{g' = 2}^{\bar g - 1 - e} \mu_{G_{g'}}(z) \right)^2 } \sum_{g' = 2}^{\bar g - 1 - e} \left( \frac{1}{ f_Z(z) } \frac{1}{nh} \sum_{i=1}^n \Psi_{i,h} ( G_{i,g'} - \mu_{G_{g'}}(Z_i) ) K_{i,h} + \mathrm{Bias} \left[ \hat \mu_{G_{g'}}(z) \; \middle| \; \bm{Z} \right] \right) \Bigg\} \\
	& \quad + o_{\bP}(h^4) + o_{\bP}\left( \sqrt{\frac{\log n}{nh}}\right).
\end{align*}
Rearranging this equation, we have 
\begin{align*}
	& \hat{w}_{g,t}^{\es}(e, z) - w_{g,t}^{\es}(e, z) \\
	& = \frac{1}{ f_Z(z) } \frac{1}{nh} \sum_{i=1}^n \Psi_{i,h} (\xi_{i,g,t}^{\es} - \mu_{\xi,\es}(Z_i)) K_{i,h} + \mathrm{Bias}[ \hat w_{g,t}^{\es}(e, z) \mid \bm{Z} ] + o_{\bP}(h^4) + o_{\bP}\left( \sqrt{\frac{\log n}{nh}} \right),
\end{align*}
implying that the event-study-type summary estimator exhibits the asymptotic linearity in Theorem \ref{thm:supp-summary}.

\subsubsection{The group-specific conditional average treatment effect}

Because the group-specific conditional average treatment effect in \eqref{eq:supp-sel} has the known weighting function $w_{g,t}^{\sel}(g', z)$, its estimator is simply given by
\begin{align*}
	\hat{\theta}_{\sel}(g', z)
	\coloneqq \sum_{g\in\mathcal{G}} \sum_{t=2}^{\mathcal{T}} w_{g,t}^{\sel}(g', z) \cdot \hat{\DR}_{g,t}(z).
\end{align*}
This summary estimator exhibits the asymptotic linear representation in Theorem \ref{thm:supp-summary} when replacing the variable $J_i$ with
\begin{align*}
	J_i^{\sel}
	\coloneqq \sum_{g \in \mathcal{G}} \sum_{t=2}^{\mathcal{T}} w_{g,t}^{\sel}(g', z) \cdot B_{i,g,t}.
\end{align*}
The proof is straightforward from the asymptotic linearity for the conditional DR estimator in Theorem \ref{thm:bias_variance} and the fact that $w_{g,t}^{\sel}(g', z)$ is the known and non-stochastic weighting function.

\subsubsection{The calendar-time conditional average treatment effect}

The calendar-time conditional average treatment effect in \eqref{eq:supp-calendar} can be estimated by
\begin{align*}
	\hat \theta_{\cc}(t', z)
	\coloneqq \sum_{g \in \mathcal{G}} \sum_{t=2}^{\mathcal{T}} \hat{w}_{g,t}^{\cc}(t', z) \cdot \hat{\DR}_{g,t}(z),
\end{align*}
where the weighting function estimator is given by
\begin{align*}
	\hat{w}_{g,t}^{\cc}(t', z)
	\coloneqq \bm{1}\{ g \le t', t = t' \} \cdot \frac{ \hat \mu_{G_g}(z) }{ \sum_{g'=2}^{t'} \hat \mu_{G_{g'}}(z) }.
\end{align*}
By the same arguments as for the case of the event-study-type summary parameter, we can show that the calendar-time summary estimator exhibits the asymptotic linearity in Theorem \ref{thm:supp-summary} with the following definitions of $J_i$ and $\xi_{i,g,t}$:
\begin{align*}
	J_i^{\cc}
	& \coloneqq \sum_{g \in \mathcal{G}} \sum_{t=2}^{\mathcal{T}} \Big( w_{g,t}^{\cc}(t', z) \cdot B_{i,g,t} + {\DR}_{g,t}(z) \cdot \xi_{i,g,t}^{\cc} \Big), \\
	\xi_{i,g,t}^{\cc}
	& \coloneqq \bm{1}\{ g \le t', t = t' \} \cdot \left( \frac{ G_{i,g} }{ \sum_{g' = 2}^{t'} \mu_{G_{g'}}(z) } - \frac{ \mu_{G_g}(z) }{ \left( \sum_{g' = 2}^{t'} \mu_{G_{g'}}(z) \right)^2 } \sum_{g' = 2}^{t'} G_{i,g'} \right).
\end{align*}

\subsubsection{The simple weighted conditional average treatment effect}

The simple weighted conditional average treatment effect in \eqref{eq:supp-overall} can be estimated by
\begin{align*}
	\hat \theta_{\W}^{\OO}(z)
	\coloneqq \sum_{g \in \mathcal{G}} \sum_{t=2}^{\mathcal{T}} \hat{w}_{g,t}^{\OW}(z) \cdot \hat{\DR}_{g,t}(z),
\end{align*}
where the weighting function estimator is given by
\begin{align*}
	\hat{w}_{g,t}^{\OW}(z)
	\coloneqq \frac{1}{\hat \kappa(z)} \cdot \bm{1}\{ g \le t < \bar g \} \cdot \frac{ \hat \mu_{G_g}(z) }{ \sum_{g'=2}^{\bar g - 1} \hat \mu_{G_{g'}}(z) }
\end{align*}
with
\begin{align*}
	\hat \kappa(z)
	\coloneqq \sum_{g \in \mathcal{G}} \sum_{t=2}^{\mathcal{T}} \bm{1}\{ g \le t < \bar g \} \cdot \frac{ \hat \mu_{G_g}(z) }{ \sum_{g'=2}^{\bar g - 1} \hat \mu_{G_{g'}}(z) }.
\end{align*}

This overall summary estimator exhibits the asymptotic linear representation in Theorem \ref{thm:supp-summary} when replacing the variable $J_i$ with
\begin{align*}
	J_i^{\OW}
	\coloneqq \sum_{g \in \mathcal{G}} \sum_{t=2}^{\mathcal{T}} \Big( w_{g,t}^{\OW}(z) \cdot B_{i,g,t} + {\DR}_{g,t}(z) \cdot \xi_{i,g,t}^{\OW} \Big),
\end{align*}
where
\begin{align*}
	\xi_{i,g,t}^{\OW}
	& \coloneqq \bm{1}\{ g \le t < \bar g \} \cdot \left( \frac{ G_{i,g} }{ \kappa(z) \cdot \sum_{g'=2}^{\bar g - 1} \mu_{G_{g'}}(z) } - \frac{ \mu_{G_g}(z) }{ \kappa(z) \cdot \left( \sum_{g'=2}^{\bar g - 1} \mu_{G_{g'}}(z) \right)^2 } \cdot \left( \sum_{g'=2}^{\bar g - 1} G_{i,g'} \right) \right. \\
	& \left. \quad - \frac{ \mu_{G_g}(z) }{ \left( \kappa(z) \cdot \sum_{g'=2}^{\bar g - 1} \mu_{G_{g'}}(z) \right)^2 } \cdot \sum_{g' \in \mathcal{G}} \sum_{t=2}^{\mathcal{T}} \bm{1}\{ g' \le t < \bar g \} \cdot \left( G_{i,g'} - \frac{ \mu_{G_{g'}}(z) }{ \left( \sum_{g''=2}^{\bar g - 1} \mu_{G_{g''}}(z) \right) } \sum_{g''=2}^{\bar g - 1} G_{i,g''} \right) \right).
\end{align*}
To see this result, it suffices to show that the weighting function estimator $\hat{w}_{g,t}^{\OW}(z)$ exhibits the asymptotic linear representation in Assumption \ref{as:supp-weight}(ii) when replacing the variable $\xi_{i,g,t}$ with $\xi_{i,g,t}^{\OW}$.
By simple algebra, we have
\begin{align*}
	& \hat{w}_{g,t}^{\OW}(z) - w_{g,t}^{\OW}(z) \\
	& = \bm{1}\{ g \le t < \bar g \} \cdot \left( \frac{1}{\hat \kappa(z)} \cdot \frac{ \hat \mu_{G_g}(z) }{ \sum_{g'=2}^{\bar g - 1} \hat \mu_{G_{g'}}(z) } - \frac{1}{\kappa(z)} \cdot \frac{ \mu_{G_g}(z) }{ \sum_{g'=2}^{\bar g - 1} \mu_{G_{g'}}(z) } \right) \\
	& = \bm{1}\{ g \le t < \bar g \} \\
	& \quad\quad  \cdot \left( \frac{1}{\hat \kappa(z)} \cdot \frac{ \hat \mu_{G_g}(z) }{ \sum_{g'=2}^{\bar g - 1} \hat \mu_{G_{g'}}(z) } - \frac{1}{\kappa(z)} \cdot \frac{ \hat \mu_{G_g}(z) }{ \sum_{g'=2}^{\bar g - 1} \mu_{G_{g'}}(z) } + \frac{1}{\kappa(z)} \cdot \frac{ \hat \mu_{G_g}(z) }{ \sum_{g'=2}^{\bar g - 1} \mu_{G_{g'}}(z) } - \frac{1}{\kappa(z)} \cdot \frac{ \mu_{G_g}(z) }{ \sum_{g'=2}^{\bar g - 1} \mu_{G_{g'}}(z) } \right) \\
	& = \bm{1}\{ g \le t < \bar g \} \cdot \left( \frac{ \hat \mu_{G_g}(z) - \mu_{G_g}(z) }{ \kappa(z) \cdot \sum_{g'=2}^{\bar g - 1} \mu_{G_{g'}}(z) } - \frac{ \hat \mu_{G_g}(z) \cdot \left( \hat \kappa(z) \cdot \sum_{g'=2}^{\bar g - 1} \hat \mu_{G_{g'}}(z) - \kappa(z) \cdot \sum_{g'=2}^{\bar g - 1} \mu_{G_{g'}}(z) \right) }{ \left( \hat \kappa(z) \cdot \sum_{g'=2}^{\bar g - 1} \hat \mu_{G_{g'}}(z) \right) \cdot \left( \kappa(z) \cdot \sum_{g'=2}^{\bar g - 1} \mu_{G_{g'}}(z) \right) } \right).
\end{align*}
Here, observe that
\begin{align*}
	\hat \kappa(z) \cdot \sum_{g'=2}^{\bar g - 1} \hat \mu_{G_{g'}}(z) - \kappa(z) \cdot \sum_{g'=2}^{\bar g - 1} \mu_{G_{g'}}(z)
	= \hat \kappa(z) \cdot \sum_{g'=2}^{\bar g - 1} \left( \hat \mu_{G_{g'}}(z) - \mu_{G_{g'}}(z) \right) + \left( \hat \kappa(z) - \kappa(z) \right) \cdot \sum_{g'=2}^{\bar g - 1} \mu_{G_{g'}}(z),
\end{align*}
and
\begin{align*}
	& \hat \kappa(z) - \kappa(z) \\
	& = \sum_{g \in \mathcal{G}} \sum_{t=2}^{\mathcal{T}} \bm{1}\{ g \le t < \bar g \} \cdot \left( \frac{ \hat \mu_{G_g}(z) }{ \sum_{g'=2}^{\bar g - 1} \hat \mu_{G_{g'}}(z) } - \frac{ \mu_{G_g}(z) }{ \sum_{g'=2}^{\bar g - 1} \mu_{G_{g'}}(z) } \right) \\
	& = \sum_{g \in \mathcal{G}} \sum_{t=2}^{\mathcal{T}} \bm{1}\{ g \le t < \bar g \} \cdot \left( \frac{ \hat \mu_{G_g}(z) - \mu_{G_g}(z) }{ \sum_{g'=2}^{\bar g - 1} \mu_{G_{g'}}(z) } - \frac{ \hat \mu_{G_g}(z) \cdot \sum_{g'=2}^{\bar g - 1} \left( \hat \mu_{G_{g'}}(z) - \mu_{G_{g'}}(z) \right) }{ \left( \sum_{g'=2}^{\bar g - 1} \hat \mu_{G_{g'}}(z) \right) \cdot \left( \sum_{g'=2}^{\bar g - 1} \mu_{G_{g'}}(z) \right) } \right).
\end{align*}
By combining these equations with \eqref{eq:Glinear}, we can see that
\begin{align*}
	& \hat{w}_{g,t}^{\OW}(z) - w_{g,t}^{\OW}(z) \\
	& = \bm{1}\{ g \le t < \bar g \} \cdot \Bigg\{ \frac{ 1 }{  \kappa(z) \cdot \sum_{g'=2}^{\bar g - 1} \mu_{G_{g'}}(z) } \left( \frac{1}{ f_Z(z) } \frac{1}{nh} \sum_{i=1}^n \Psi_{i,h} ( G_{i,g} - \mu_{G_g}(Z_i) ) K_{i,h} + \mathrm{Bias}[ \hat \mu_{G_g}(z) \mid \bm{Z} ] \right) \\
	& \quad - \frac{ \mu_{G_g}(z) }{ \left( \kappa(z) \cdot \sum_{g'=2}^{\bar g - 1} \mu_{G_{g'}}(z) \right)^2 } \cdot \Bigg[ \kappa(z) \cdot \sum_{g'=2}^{\bar g - 1} \left( \frac{1}{ f_Z(z) } \frac{1}{nh} \sum_{i=1}^n \Psi_{i,h} ( G_{i,g'} - \mu_{G_{g'}}(Z_i) ) K_{i,h} + \mathrm{Bias}[ \hat \mu_{G_{g'}}(z) \mid \bm{Z} ] \right) \\
	& \quad\quad + \sum_{g' \in \mathcal{G}} \sum_{t=2}^{\mathcal{T}} \bm{1}\{ g' \le t < \bar g \} \cdot \Bigg\{ \left( \frac{1}{ f_Z(z) } \frac{1}{nh} \sum_{i=1}^n \Psi_{i,h} ( G_{i,g'} - \mu_{G_{g'}}(Z_i) ) K_{i,h} + \mathrm{Bias}[ \hat \mu_{G_{g'}}(z) \mid \bm{Z} ] \right) \\
	& \quad\quad\quad - \frac{ \mu_{G_{g'}}(z) }{ \left( \sum_{g''=2}^{\bar g - 1} \mu_{G_{g''}}(z) \right) } \sum_{g''=2}^{\bar g - 1} \left( \frac{1}{ f_Z(z) } \frac{1}{nh} \sum_{i=1}^n \Psi_{i,h} ( G_{i,g''} - \mu_{G_{g''}}(Z_i) ) K_{i,h} + \mathrm{Bias}[ \hat \mu_{G_{g''}}(z) \mid \bm{Z} ] \right) \Bigg\} \Bigg] \Bigg\} \\
	& \quad + o_{\bP}(h^4) + o_{\bP}\left( \sqrt{\frac{\log n}{nh}} \right).
\end{align*}
Rearranging this equation, we can write 
\begin{align*}
	& \hat{w}_{g,t}^{\OW}(z) - w_{g,t}^{\OW}(z) \\
	& = \frac{1}{ f_Z(z) } \frac{1}{nh} \sum_{i=1}^n \Psi_{i,h} (\xi_{i,g,t}^{\OW} - \mu_{\xi,\OW}(Z_i)) K_{i,h} + \mathrm{Bias}[ \hat w_{g,t}^{\OW}(z) \mid \bm{Z} ] + o_{\bP}(h^4) + o_{\bP}\left( \sqrt{\frac{\log n}{nh}} \right),
\end{align*}
implying that the overall summary estimator exhibits the asymptotic linearity in Theorem \ref{thm:supp-summary}.

\subsection{Proof of Theorem \ref{thm:supp-summary}}

We focus on the proof under Assumption \ref{as:supp-weight}(ii) because the case under Assumption \ref{as:supp-weight}(i) is trivial. 
Observe that
\begin{align*}
	\hat \theta(z) - \theta(z)
	& = \sum_{g \in \mathcal{G}} \sum_{t=2}^{\mathcal{T}} \left( \hat{w}_{g,t}(z) \cdot \hat{\DR}_{g,t}(z) - w_{g,t}(z) \cdot {\DR}_{g,t}(z) \right) \\
	& = \sum_{g \in \mathcal{G}} \sum_{t=2}^{\mathcal{T}} \left( w_{g,t}(z) \cdot \left[ \hat{\DR}_{g,t}(z) - {\DR}_{g,t}(z) \right] + \hat{\DR}_{g,t}(z) \cdot \left[ \hat{w}_{g,t}(z) - w_{g,t}(z) \right] \right) \\
	& = \sum_{g \in \mathcal{G}} \sum_{t=2}^{\mathcal{T}} \left( w_{g,t}(z) \cdot \left[ \hat{\DR}_{g,t}(z) - {\DR}_{g,t}(z) \right] + {\DR}_{g,t}(z) \cdot \left[ \hat w_{g,t}(z) - w_{g,t}(z) \right] \right) \\
	& \quad + \sum_{g \in \mathcal{G}} \sum_{t=2}^{\mathcal{T}} \left[ \hat{\DR}_{g,t}(z) - {\DR}_{g,t}(z) \right] \cdot \left[ \hat{w}_{g,t}(z) - w_{g,t}(z) \right].
\end{align*}
Here, by Theorem \ref{thm:bias_variance} and Assumption \ref{as:supp-weight}, we have
\begin{align*}
	& \hat{\DR}_{g,t}(z) - {\DR}_{g,t}(z) \\
	& = \frac{1}{f_Z(z)} \frac{1}{nh} \sum_{i=1}^n \Psi_{i,h} (B_{i,g,t} - \mu_B(Z_i)) K_{i,h} + \mathrm{Bias}\left[ \hat{\DR}_{g,t}(z) \; \middle| \; \bm{Z} \right] + o_{\bP}\left( h^4 \right) + o_{\bP}\left( \sqrt{\frac{\log n}{nh}} \right) \\
	& = O_{\bP}\left( \sqrt{\frac{\log n}{nh}} \right),
\end{align*}
and
\begin{align*}
	& \hat{w}_{g,t}(z) - w_{g,t}(z) \\
	& = \frac{1}{ f_Z(z) } \frac{1}{nh} \sum_{i=1}^n \Psi_{i,h} (\xi_{i,g,t} - \mu_{\xi}(Z_i)) K_{i,h} + \mathrm{Bias} \left[ \hat{w}_{g,t}(z) \; \middle| \; \bm{Z} \right] + o_{\bP}(h^4) + o_{\bP}\left( \sqrt{\frac{\log n}{nh}} \right) \\
	& = O_{\bP}\left( \sqrt{\frac{\log n}{nh}} \right).
\end{align*}
Combining these three equations, we obtain the desired result.
\qed 

\section{Pre-Trends} \label{sec:pretrends}

\phantomsection\label{page:AE-7-2}\Copy{AE-7-2}{
	In this section, we discuss how to assess the credibility of the identifying assumptions by examining pre-trends.
	Because the conditional parallel trends assumption is particularly important for the DiD analysis, we focus on assessing its plausibility.
	For presentation purposes, we focus on the analysis based on the no-anticipation condition and the conditional parallel trends based on the not-yet-treated group in Assumptions \ref{as:anticipation-ny} and \ref{as:parallel-ny}, respectively.
	
	Throughout this section, we assume the same type of conditional parallel trends as in Assumption \ref{as:parallel-ny} in both the pre-treatment and post-treatment periods, that is, for all $g \in \mathcal{G}$, $t \ge 2$, $s \ge 2$ such that $g \le s < \bar g$, 
	\begin{align*}
		\bE[ Y_t(0) - Y_{t-1}(0) \mid X, G_g = 1 ]
		= \bE[ Y_t(0) - Y_{t-1}(0) \mid X, D_s = 0, G_g = 0 ]
		\quad 
		\text{a.s.}
	\end{align*}
	From a theoretical viewpoint, this assumption is stronger than Assumption \ref{as:parallel-ny} in that only the former requires parallel trends in the pre-treatment periods.
	However, in many practical situations, there should not be a significant difference between the two parallel trends assumptions.
	
	To consider our treatment parameters in the pre-treatment periods, take any $g \in \mathcal{G}$, $t \ge 2$ such that $t < g$, and $z \in \mathcal{I}$.
	Under Assumptions \ref{as:staggered}--\ref{as:parametric-ny} and the conditional parallel trends assumption above, the same arguments as in Lemma \ref{lem:DR} can show that ${\CATT}_{g,t}(Z) = {\DR}_{g,t}^{\mathrm{pre}}(Z)$ a.s., where
	\begin{align} \label{eq:DRestimand-pre}
		{\DR}_{g,t}^{\mathrm{pre}}(Z)
		\coloneqq \bE \left[ \left( \frac{G_g}{\bE[ G_g \mid Z]} - \frac{ R_{g,g}  }{\bE \left[ R_{g,g} \; \middle| \; Z \right]} \right) \left( Y_t - Y_{g-1} - m_{g,t}^{\mathrm{pre}}(X) \right) \; \middle| \; Z \right]
	\end{align}
	with $R_{g,g}$ corresponding to $R_{g,t}$ in \eqref{eq:Rdef} with $t = g$ and
	\begin{align*}
		m_{g,t}^{\mathrm{pre}}(X)
		\coloneqq \bE[ Y_t - Y_{g-1} \mid X, D_g = 0, G_g = 0 ]. 
	\end{align*}
	Notice the difference in the DR estimands between \eqref{eq:DRestimand-pre} and \eqref{eq:DRestimand}.
	In addition, as a direct consequence of the no-anticipation condition in Assumption \ref{as:anticipation-ny}, it holds that ${\CATT}_{g,t}(z) = \bE[ Y_t(0) - Y_t(0) \mid G_g = 1, Z = z ] = 0$ in the pre-treatment periods.
	As a result, we have ${\CATT}_{g,t}(z) = {\DR}_{g,t}^{\mathrm{pre}}(z) = 0$.
	However, this result can be meaningful only when $t \neq g - 1$, because this result should be trivial from the construction of the DR estimand when $t = g - 1$.
	Therefore, we obtain the following testable implication for the identifying assumptions:
	\begin{align} \label{eq:pre-trends1}
		{\CATT}_{g,t}(z) 
		= {\DR}_{g,t}^{\mathrm{pre}}(z) 
		= 0
		\quad \text{for all $g \in \mathcal{G}$, $t \ge 2$ such that $t \le g - 2$, and $z \in \mathcal{I}$}.
	\end{align}
	Moreover, it is easy to see from \eqref{eq:pre-trends1} that the event-study-type conditional average treatment effect should also vanish in the pre-treatment periods:
	\begin{align} \label{eq:pre-trends2}
		\theta_{\es}(e, z)
		= 0
		\quad \text{for all $e \le -2$ and $z \in \mathcal{I}$}.
	\end{align}
	Note that \eqref{eq:pre-trends2} excludes $e = -1$ as the baseline for the same reason that \eqref{eq:pre-trends1} excludes $t = g - 1$.
	
	Based on the testable implications in \eqref{eq:pre-trends1} and \eqref{eq:pre-trends2}, we can assess the plausibility of conditional parallel trends using our estimation and uniform inference methods.
	Specifically, we obtain the LQR estimates and the uniform confidence bands for $\CATT_{g,t}(z) = \DR_{g,t}^{\mathrm{pre}}(z)$ and $\theta_{\es}(e, z)$ in the pre-treatment periods in the same manner as in the main text.
	If there are many estimates far from zero so that the resulting uniform confidence bands exclude zero at many evaluation points, this is inconsistent with the testable implications of no pre-trends in \eqref{eq:pre-trends1} and \eqref{eq:pre-trends2} and suggests violations of conditional parallel trends.
	Conversely, the conditional parallel trends assumption is not refuted if the resulting uniform confidence bands include zero at most evaluation points.
	
	For the conditional DR estimands in \eqref{eq:DRestimand-pre} and \eqref{eq:DRestimand}, we consider ``long-differences'' by comparing $Y_t$ with $Y_{g-1}$ rather than ``short-differences'' that compare $Y_t$ with $Y_{t-1}$.
	We prefer long-differences because, as pointed out by \citet{roth2024interpreting}, if there are violations of parallel trends, event-study plots produced by recent DiD methods based on short-differences should lead to different interpretations than traditional dynamic two-way fixed effects event-study plots.
	As a solution to this problem, \citet{roth2024interpreting} recommends taking long-differences in both the pre-treatment and post-treatment periods.
	Our definition of the DR estimand in \eqref{eq:DRestimand-pre} follows this recommendation, though our goal is not to produce event-study plots, but rather to construct uniform confidence bands.
	
	Importantly, we view the procedure discussed here as complementary to, rather than a substitute for, the analysis of pre-trends for the group-time average treatment effect and the event-study-type estimand by \citet{callaway2021difference}.
	This is because their estimators achieve the parametric rate, making their approach more effective at detecting violations of conditional parallel trends than our approach based on nonparametric kernel smoothing.
	We recommend using our approach in conjunction with theirs to examine the possible heterogeneity in pre-trends with respect to groups, periods, and covariate values.
	
	Lastly, we should be cautious about using our uniform inference results for pre-trends as a ``pre-test'' to determine whether researchers can apply our estimation and uniform inference methods to the post-treatment periods.
	This type of pre-testing may distort both point estimates and uniform inference results in post-treatment periods, as pointed out by \citet{roth2022pre}.
	Instead, we recommend examining the estimation and uniform inference results in the pre-treatment and post-treatment periods simultaneously and checking pre-trends to assess the credibility of the identifying assumptions originally justified by the context of the application.
}

\section{Monte Carlo Experiments} \label{sec:simulation}

In this section, we evaluate the finite sample performance of our proposal through Monte Carlo experiments.

\subsection{Simulation design} \label{subsec:simulation_design}

We generate the group $G$ for each unit by the same simulation design as in \citet{callaway2021difference}.
To be specific, let $\mathcal{G} = \{ 0, 2, 3, \dots, \mathcal{T} \}$ denote the support of $G$, where $0$ indicates the never-treated group (for exposition purposes, we use the slightly different notation from the other sections).
Let $X = (Z, X_{\sub}^\top)^\top$ be the $k$-dimensional vector of pre-treatment covariates such that $X \sim \mathrm{Normal}(\bm{0}_k, \bm{I}_k)$, where $\bm{0}_k$ and $\bm{I}_k$ denote the zero vector and the identity matrix of size $k$, respectively.
The group $G$ is determined with
\begin{align*}
	\bP(G = g \mid X) 
	= \bP(G = g \mid Z) 
	= \frac{ \exp(Z \gamma_g) }{ \sum_{g \in \mathcal{G}} \exp(Z \gamma_g) },
\end{align*}
where $\gamma_g = 0.5 g / \mathcal{T}$ for $g \in \mathcal{G}$.
Here, we let the group choice probability depend only on $Z$ to make it easier to compute the true parameter values for the summary parameters.

In a slight departure from the simulation design in \citet{callaway2021difference}, we consider the following potential outcome equations that allow for three practically relevant issues:
(i) treatment effect heterogeneity with respect to observable covariates;
(ii) treated potential outcomes that may be nonlinear in observable covariates; and
(iii) heteroscedastic error terms.
Specifically, we first generate the untreated potential outcome by 
\begin{align*}
	Y_{i,t}(0) 
	= \delta_t + \eta_i + X_i^\top \beta_t(0) + u_{i,t}(0),
\end{align*}
where $\delta_t = t$, $\beta_t(0) = (t, t/2, \dots, t/k)^\top$, $\eta_i | (G_i, X_i) \sim \mathrm{Normal}(G_i, 1)$, and $u_{i,t}(0) | (G_i, X_i) \sim \mathrm{Normal}(0, \sigma_0^2(X_i))$.
We then generate the potential outcome for group $g$ by
\begin{align*}
	Y_{i,t}(g) 
	= Y_{i,t}(0) + M_{g,t}(Z_i) + \delta_e + \Big( u_{i,t}(g) - u_{i,t}(0) \Big),
\end{align*}
where $\delta_e = e + 1 \coloneqq t - g + 1$, $u_{i,t}(g) | (G_i, X_i, u_{i,t}(0)) \sim \mathrm{Normal}(0, \sigma_g^2(X_i))$, and $M_{g,t}$ is a function for which we consider two scenarios:
\begin{enumerate}[(i)]
	\item The nonlinear potential outcome equation where $M_{g,t}(Z_i) = (g / t) \sin(\pi Z_i)$.
	\item The linear potential outcome equation where $M_{g,t}(Z_i) = (g / t) Z_i$.
\end{enumerate}
For the conditional variances $\sigma_0^2(X_i)$ and $\sigma_g^2(X_i)$, we consider the following two cases:
\begin{enumerate}[(i)]
	\item The homoscedastic case where $\sigma_0^2(X_i) = \sigma_g^2(X_i) = 1$.
	\item The heteroscedastic case where $\sigma_0^2(X_i) = 0.5 + \Phi(Z_i)$ and $\sigma_g^2(X_i) = (g/\mathcal{T}) + \Phi(Z_i)$.
\end{enumerate}
Here, $\Phi$ denotes the standard normal cumulative distribution function.

Most of this section considers CATT as the target parameter, which is given by
\begin{align*}
	{\CATT}_{g,t}(z) 
	= \bE [ Y_t(g) - Y_t(0) \mid G_g = 1, Z = z  ]
	= M_{g,t}(z) + \delta_e.
\end{align*}
For the choice of $(g, t, z)$, we focus on $g = t = 2$ and a grid of 41 equally spaced points over $\mathcal{I} = [-1, 1]$.

To study the aggregated parameter $\theta(z)$, at the end of this section, we also examine the finite sample performance of the uniform inference for the event-study-type conditional average treatment effect at $e = 0$:
\begin{align*}
	\theta_{\es}(0, z)
	& \coloneqq \bE[ Y_{i,G}(G) - Y_{i,G}(0) \mid G \neq 0, Z = z ] \\
	& = \sum_{g=2}^{\mathcal{T}} \Pr(G = g \mid G \neq 0, Z = z) \cdot {\CATT}_{g,g}(z).
\end{align*}

\subsection{Methods}

We consider the uniform inference methods proposed in the main text.
Specifically, using the not-yet-treated group as the comparison group, we first obtain the estimates of the GPS and OR function by the parametric logit method and OLS, respectively.
We then estimate the nuisance parameters and the conditional DR estimand in the second and third stages, as well as the weighting function for the aggregated parameter, using the $p$-th order LPR estimation with the Gaussian kernel.

For the choice of the local polynomial order $p$ and the bandwidth $h$, we consider the four alternatives:
\begin{enumerate}[(i)]
	\item The LLR estimation based on the rule-of-thumb undersmoothing $\hat h_{\US} \coloneqq \hat h_{\mathrm{LL}} \cdot n^{1/5} \cdot n^{-2/7}$, which is theoretically justified in the preprint version of this article (\citealp{imai2023doubly}).
	\item The LLR estimation based on the IMSE-optimal bandwidth $\hat h_{\mathrm{LL}}$ for the LLR estimation.
	\item The LQR estimation based on the IMSE-optimal bandwidth $\hat h_{\mathrm{LL}}$ for the LLR estimation, which is theoretically justified in the main text.
	\item The LQR estimation based on the IMSE-optimal bandwidth $\hat h_{\mathrm{LQ}}$ for the LQR estimation.
\end{enumerate}
Our theoretical investigations suggest that the undersmoothing methods (i) and (iii) outperform the IMSE-optimal methods (ii) and (iv) in terms of uniform inference (i.e., correct uniform coverage probability).

We construct the uniform critical values by the analytical method and weighted/multiplier bootstrapping with \citeauthor{mammen1993bootstrap}'s (\citeyear{mammen1993bootstrap}) weights.

\subsection{Baseline results} \label{subsec:MC:baseline}

As a baseline, Tables \ref{table:MC1} and \ref{table:MC2} report the simulation results for CATT when $n \in \{ 500, 1000 \}$, $\mathcal{T} \in \{ 2, 4 \}$, the number of the pre-treatment covariates is set to $k = 1$ such that $X = Z$, the function in the treated potential outcome equation is given by $M_{g,t}(Z_i) = (g/t) \sin(\pi Z_i)$, the error terms $u_{it}(0)$ and $u_{it}(g)$ are homoscedastic, and the number of Monte Carlo replications is set to 1,000.
The tables show the pointwise bias, the pointwise root mean squared error (RMSE), and the empirical uniform coverage probability (UCP) and the pointwise length of the 95\% uniform confidence bands based on the analytical method and weighted bootstrapping.
For presentation purposes, the tables focus only on the simulation results for $z \in \{ -1, 0, 1 \}$.

The simulation results highlight the satisfactory performance of the proposed undersmoothing methods (i.e., the LLR-based inference using the rule-of-thumb undersmoothing and the LQR-based inference using the IMSE-optimal bandwidth for the LLR estimation), especially when coupled with weighted bootstrapping.
The bias and RMSE of both methods are sufficiently small regardless of the number of units $n$ and the length of the time series $\mathcal{T}$.
Interestingly, the empirical uniform coverage probability for the analytical method is somewhat under-coverage, but that for weighted bootstrapping is satisfactorily close to the desired level, suggesting that weighted bootstrapping may have some asymptotic refinements that the analytical method does not.
Note that the weighted bootstrap inference produces a wider confidence band on average than the analytical method, but the difference is modest and the length for weighted bootstrapping should also be acceptable.

For the comparison between the LLR-based inference using the rule-of-thumb undersmoothing and the LQR-based inference using the IMSE-optimal bandwidth for the LLR estimation, it seems difficult to rank them from the simulation results here.
In some cases, the LLR-based inference seems to be preferable, but in other cases, the LQR-based inference seems to be preferable, and the difference is small.
Given this, it should be recommended in practice to implement both methods and see the difference.
Specifically, if they produce substantially different inference results, this may indicate that there is some violation of assumptions and/or some other problems in the data (e.g., outliers).

The simulation results also show that the IMSE-optimal methods, that is, the LLR-based (resp. LQR-based) inference using the IMSE-optimal bandwidth for the LLR (resp. LQR) estimation, does not work well in terms of uniform inference.
This is because the IMSE-optimal methods lead to non-negligible bias, which can be more than twice as large as the undersmoothing methods.
As a result, the empirical uniform coverage probability for the IMSE-optimal methods can be far from the nominal level, despite that they achieve smaller RMSE than the undersmoothing methods.

\subsection{Additional results: the number of pre-treatment covariates}

To examine whether the simulation results are sensitive to the number of the pre-treatment covariates $k$, we change from $k = 1$ to $k = 5$ in this subsection.
The other simulation settings are the same as in Section \ref{subsec:MC:baseline}, but to save space we only report the simulation results for $\mathcal{T} = 2$.

Table \ref{table:MC3} shows the simulation results, which are almost the same as the baseline results in Table \ref{table:MC1}.
Thus, the same comments as in Section \ref{subsec:MC:baseline} apply to this simulation setting as well.

\subsection{Additional results: the heteroscedastic error terms} \label{subsec:heteroscedastic}

To examine whether the simulation results are sensitive to the distributions of the error terms $u_{it}(0)$ and $u_{it}(g)$, we consider the heteroscedastic case in this subsection.
The other simulation settings are the same as in Section \ref{subsec:MC:baseline}, but to save space we only report the simulation results for $\mathcal{T} = 2$.

Table \ref{table:MC4} shows almost the same simulation results as the baseline results in Table \ref{table:MC1}.

\subsection{Additional results: the linear potential outcome equation} \label{subsec:MC:linear}

To examine whether the simulation results are sensitive to the functional form of $M_{g,t}(Z_i)$ in the potential outcome equation, we consider the linear potential outcome equation such that $M_{g,t}(Z_i) = (g/t) Z_i$ in this section.
The other simulation settings are the same as in Section \ref{subsec:MC:baseline}, but to save space we only report the simulation results for $\mathcal{T} = 2$.

Table \ref{table:MC5} shows the simulation results.
Many of the same comments made in Section \ref{subsec:MC:baseline} apply to this simulation setting as well, but it should be noted that the IMSE-optimal methods perform as well as the undersmoothing methods in this setting.
This is because in this setting the CATT function is linear in $Z$ and the LLR and LQR estimators have small bias regardless of the choice of bandwidth.
However, in practical situations, researchers do not know the functional form of CATT a priori, and it should be recommended to use the undersmoothing methods rather than the IMSE-optimal methods.

\subsection{Additional results: the aggregated parameter}

In this subsection, we examine the finite sample performance of the uniform inference method for the event-study-type conditional average treatment effect at $e = 0$, that is, $\theta_{\es}(0, z)$.
The simulation setting is the same as in Sections \ref{subsec:MC:baseline} and \ref{subsec:MC:linear}.
Since $\theta_{\es}(0, z)$ reduces to $\CATT_{2,2}(z)$ when $\mathcal{T} = 2$, we here focus on the simulation results only for $\mathcal{T} = 4$.
In addition, we set the number of Monte Carlo replications to 500 to reduce computation time.

Tables \ref{table:MC6} and \ref{table:MC7} show the simulation results for the nonlinear and linear potential outcome equations, respectively.
Similar to the simulation results for CATT, the undersmoothing methods outperform the IMSE-optimal methods.
It is also noteworthy that the estimation of the event-study-type summary parameter is substantially more precise than the estimation of CATTs, which can be expected from the fact that the summary parameter is obtained by aggregating CATTs.

\subsection{Recommendations based on the simulation results}

In conclusion, the simulation results for CATT and the event-study-type summary parameter suggest that the LQR-based inference using the IMSE-optimal bandwidth for the LLR estimation (i.e., the simple RBC approach) works satisfactorily in many situations, especially when coupled with weighted/multiplier bootstrapping.
The LLR-based inference using the rule-of-thumb undersmoothing also performs adequately.
The standard IMSE-optimal methods are not recommended for uniform inference.

\section{Additional Details on Empirical Illustration} \label{sec:empirical2}

\subsection{Data}

The data we use in Section \ref{sec:application} is {\ttfamily min\_wage\_CS.rds}, which was taken from the Pedro H. C. Sant'Anna's GitHub repository: \url{https://github.com/pedrohcgs/CS_RR/tree/main/data}.
The details of the data construction can be found in Section 5 and Appendix SB of \citet{callaway2021difference}, but there are several things to note here.
First, the information on teen employment at the county level comes from the Quarterly Workforce Indicators (QWI).
Second, the pre-treatment characteristics, including the fraction of the population below the poverty line in 1997 (i.e., the poverty rate), are taken from the 2000 County Data Book.
Lastly, the sample consists of counties from 29 states, and the other states are excluded because (i) their minimum wages are higher than the federal minimum wage in 2000, (ii) their information on teen employment is unavailable, or (iii) they are located in the Northern Census region.

\subsection{Summary statistics}

To examine the distribution of the poverty rate in terms of which we assess the treatment effect heterogeneity, Figure \ref{fig:density} shows the kernel density of the poverty rate based on the Epanechnikov kernel and the MSE-optimal bandwidth computed with the {\ttfamily nprobust} package in {\ttfamily R} (\citealp{calonico2019nprobust}).
In the figure, the dashed line indicates the median of the poverty rate (14\%). 
The density is right skewed, indicating a higher concentration of counties with poverty rates below the median. 

To check whether or not the distributions of the pre-treatment variables are balanced between the treated groups and the not-yet-treated groups, Table \ref{table:descriptive} shows their means and their standard deviations (SD).
It is clear that the distributions of the pre-treatment variables are substantially different between the treated groups and the not-yet-treated groups.
It seems that counties located in the Midwest are more likely to increase minimum wages than counties in other regions.
Moreover, minimum wage increases occur earlier in counties with lower poverty rates, larger populations, higher proportions of high school graduates and white residents, and higher median incomes.
Accordingly, the minimum wage seems to depend on the pre-treatment variables and we should include the pre-treatment variables in the DiD analysis.

\subsection{Pre-trends}

\phantomsection\label{page:AE-7-3}\Copy{AE-7-3}{
	In light of the discussion about pre-trends in Appendix \ref{sec:pretrends}, Figure \ref{fig:boot-pre} shows the LQR estimates and the uniform inference results for $\CATT_{g,t}(z)$ and the event-study-type conditional average treatment effect $\theta_{\es}(e, z)$ in the pre-treatment periods such that $t \le g - 2$ and $e \le -2$.
	For presentation purposes, panel (a) presents the results of $\CATT_{g,t}(z)$ for a subset of the pre-treatment periods, but the results for $\theta_{\es}(e, z)$ in panel (b) are obtained by using data from all available pre-treatment periods.
	Note that the figure excludes the results for the base periods $t = g - 1$ and $e = -1$.
	The options for uniform inference are the same as in Figure \ref{fig:boot} in the main text, except that the results here are based on the bandwidth obtained by taking the minimum of the integrated (over $z \in \mathcal{I}$) MSE-optimal bandwidths for the LLR estimation across all $(g, t)$ or $e$, including both the pre-treatment and post-treatment periods.
	Nevertheless, the differences in the bandwidths between Figures \ref{fig:boot} and \ref{fig:boot-pre} are sufficiently small and less than 0.01.
	We also set the vertical axis scale of Figure \ref{fig:boot-pre} slightly larger than that of Figure \ref{fig:boot} to display the relatively wide, uniform confidence band of $\CATT_{g,t}(z)$ for $g = 2006$ and $t = 2002$.
	
	Overall, the estimates in the pre-treatment periods are close to zero, with a small amount of treatment effect heterogeneity, and the corresponding uniform confidence bands include zero at most evaluation points.
	This result is consistent with the testable implications of no pre-trends in \eqref{eq:pre-trends1} and \eqref{eq:pre-trends2}, which strengthens the plausibility of the conditional parallel trends assumption in this empirical context.
	However, for the same dataset, the empirical analysis of \citet{callaway2021difference} finds statistical evidence against the conditional parallel trends assumption.
	Therefore, we should be careful about the possibility that our result of no pre-trends stems from relatively imprecise kernel smoothing estimates.
}


\begin{table}[p]
	\caption{Baseline Monte Carlo simulation results for CATT: $\mathcal{T} = 2$} \label{table:MC1}
	\centering
	\begin{tabular}[t]{rrrrrrrrrrr}
		\toprule
		\multicolumn{7}{c}{ } & \multicolumn{2}{c}{Analytical UCB} & \multicolumn{2}{c}{Bootstrap UCB} \\
		\cmidrule(l{3pt}r{3pt}){8-9} \cmidrule(l{3pt}r{3pt}){10-11}
		$\mathcal{T}$ & $n$ & $p$ & Bandwidth & $z$ & Bias & RMSE & UCP & Length & UCP & Length\\
		\midrule
		\rowcolor{Gray}
		2 & 500 & 1 & US1 & -1 & -0.049 & 0.377 & 0.930 & 2.176 & 0.950 & 2.259\\
		\rowcolor{Gray}
		2 & 500 & 1 & US1 & 0 & 0.025 & 0.286 & 0.930 & 1.655 & 0.950 & 1.718\\
		\rowcolor{Gray}
		2 & 500 & 1 & US1 & 1 & -0.029 & 0.374 & 0.930 & 2.098 & 0.950 & 2.178\\
		2 & 500 & 1 & IMSE1 & -1 & -0.071 & 0.305 & 0.703 & 1.535 & 0.791 & 1.668\\
		2 & 500 & 1 & IMSE1 & 0 & 0.044 & 0.229 & 0.703 & 1.174 & 0.791 & 1.276\\
		2 & 500 & 1 & IMSE1 & 1 & -0.027 & 0.297 & 0.703 & 1.485 & 0.791 & 1.614\\
		\rowcolor{Gray}
		2 & 500 & 2 & IMSE1 & -1 & -0.091 & 0.384 & 0.909 & 1.997 & 0.955 & 2.208\\
		\rowcolor{Gray}
		2 & 500 & 2 & IMSE1 & 0 & 0.022 & 0.292 & 0.909 & 1.514 & 0.955 & 1.674\\
		\rowcolor{Gray}
		2 & 500 & 2 & IMSE1 & 1 & 0.025 & 0.364 & 0.909 & 1.923 & 0.955 & 2.125\\
		2 & 500 & 2 & IMSE2 & -1 & -0.193 & 0.374 & 0.812 & 1.562 & 0.912 & 1.784\\
		2 & 500 & 2 & IMSE2 & 0 & 0.042 & 0.254 & 0.812 & 1.190 & 0.912 & 1.359\\
		2 & 500 & 2 & IMSE2 & 1 & 0.087 & 0.323 & 0.812 & 1.509 & 0.912 & 1.723\\
		\midrule
		\rowcolor{Gray}
		2 & 1000 & 1 & US1 & -1 & -0.013 & 0.295 & 0.945 & 1.770 & 0.954 & 1.829\\
		\rowcolor{Gray}
		2 & 1000 & 1 & US1 & 0 & 0.015 & 0.209 & 0.945 & 1.342 & 0.954 & 1.386\\
		\rowcolor{Gray}
		2 & 1000 & 1 & US1 & 1 & -0.014 & 0.284 & 0.945 & 1.717 & 0.954 & 1.774\\
		2 & 1000 & 1 & IMSE1 & -1 & -0.035 & 0.229 & 0.756 & 1.212 & 0.827 & 1.297\\
		2 & 1000 & 1 & IMSE1 & 0 & 0.031 & 0.165 & 0.756 & 0.922 & 0.827 & 0.987\\
		2 & 1000 & 1 & IMSE1 & 1 & -0.021 & 0.219 & 0.756 & 1.177 & 0.827 & 1.259\\
		\rowcolor{Gray}
		2 & 1000 & 2 & IMSE1 & -1 & -0.034 & 0.285 & 0.928 & 1.571 & 0.963 & 1.732\\
		\rowcolor{Gray}
		2 & 1000 & 2 & IMSE1 & 0 & 0.011 & 0.207 & 0.928 & 1.187 & 0.963 & 1.309\\
		\rowcolor{Gray}
		2 & 1000 & 2 & IMSE1 & 1 & 0.018 & 0.268 & 0.928 & 1.522 & 0.963 & 1.678\\
		2 & 1000 & 2 & IMSE2 & -1 & -0.142 & 0.276 & 0.817 & 1.159 & 0.925 & 1.320\\
		2 & 1000 & 2 & IMSE2 & 0 & 0.029 & 0.176 & 0.817 & 0.879 & 0.925 & 1.001\\
		2 & 1000 & 2 & IMSE2 & 1 & 0.078 & 0.233 & 0.817 & 1.125 & 0.925 & 1.281\\
		\bottomrule
	\end{tabular}
	
	\begin{flushleft}
		Note: 
		The gray rows are the simulation results for the LLR and LQR estimation based on the undersmoothing methods.
		
		\bigskip  
		
		Abbreviations: 
		UCB = the uniform confidence band;
		RMSE = the root mean squared error;
		UCP = the (empirical) uniform coverage probability;
		US1 = the rule-of-thumb undersmoothing $\hat h_{\US}$ for the LLR estimation;
		IMSE1 = the IMSE-optimal bandwidth $\hat h_{\mathrm{LL}}$ for the LLR estimation;
		IMSE2 = the IMSE-optimal bandwidth $\hat h_{\mathrm{LQ}}$ for the LQR estimation.
	\end{flushleft}
\end{table}

\begin{table}[p]
	\caption{Baseline Monte Carlo simulation results for CATT: $\mathcal{T} = 4$} \label{table:MC2}
	\centering
	\begin{tabular}[t]{rrrrrrrrrrr}
		\toprule
		\multicolumn{7}{c}{ } & \multicolumn{2}{c}{Analytical UCB} & \multicolumn{2}{c}{Bootstrap UCB} \\
		\cmidrule(l{3pt}r{3pt}){8-9} \cmidrule(l{3pt}r{3pt}){10-11}
		$\mathcal{T}$ & $n$ & $p$ & Bandwidth & $z$ & Bias & RMSE & UCP & Length & UCP & Length\\
		\midrule
		\rowcolor{Gray}
		4 & 500 & 1 & US1 & -1 & -0.015 & 0.415 & 0.923 & 2.347 & 0.974 & 2.745\\
		\rowcolor{Gray}
		4 & 500 & 1 & US1 & 0 & -0.009 & 0.310 & 0.923 & 1.864 & 0.974 & 2.182\\
		\rowcolor{Gray}
		4 & 500 & 1 & US1 & 1 & 0.019 & 0.399 & 0.923 & 2.372 & 0.974 & 2.775\\
		4 & 500 & 1 & IMSE1 & -1 & -0.016 & 0.334 & 0.710 & 1.639 & 0.859 & 1.988\\
		4 & 500 & 1 & IMSE1 & 0 & -0.011 & 0.240 & 0.710 & 1.316 & 0.859 & 1.596\\
		4 & 500 & 1 & IMSE1 & 1 & 0.023 & 0.330 & 0.710 & 1.658 & 0.859 & 2.010\\
		\rowcolor{Gray}
		4 & 500 & 2 & IMSE1 & -1 & -0.073 & 0.414 & 0.909 & 2.148 & 0.972 & 2.666\\
		\rowcolor{Gray}
		4 & 500 & 2 & IMSE1 & 0 & -0.009 & 0.312 & 0.909 & 1.707 & 0.972 & 2.119\\
		\rowcolor{Gray}
		4 & 500 & 2 & IMSE1 & 1 & 0.080 & 0.405 & 0.909 & 2.168 & 0.972 & 2.691\\
		4 & 500 & 2 & IMSE2 & -1 & -0.140 & 0.382 & 0.831 & 1.709 & 0.951 & 2.167\\
		4 & 500 & 2 & IMSE2 & 0 & -0.011 & 0.273 & 0.831 & 1.372 & 0.951 & 1.740\\
		4 & 500 & 2 & IMSE2 & 1 & 0.151 & 0.387 & 0.831 & 1.728 & 0.951 & 2.191\\
		\midrule
		\rowcolor{Gray}
		4 & 1000 & 1 & US1 & -1 & -0.009 & 0.299 & 0.939 & 1.900 & 0.985 & 2.212\\
		\rowcolor{Gray}
		4 & 1000 & 1 & US1 & 0 & -0.005 & 0.246 & 0.939 & 1.506 & 0.985 & 1.753\\
		\rowcolor{Gray}
		4 & 1000 & 1 & US1 & 1 & -0.017 & 0.320 & 0.939 & 1.947 & 0.985 & 2.267\\
		4 & 1000 & 1 & IMSE1 & -1 & -0.008 & 0.234 & 0.755 & 1.293 & 0.856 & 1.540\\
		4 & 1000 & 1 & IMSE1 & 0 & -0.002 & 0.191 & 0.755 & 1.030 & 0.856 & 1.227\\
		4 & 1000 & 1 & IMSE1 & 1 & -0.006 & 0.245 & 0.755 & 1.323 & 0.856 & 1.575\\
		\rowcolor{Gray}
		4 & 1000 & 2 & IMSE1 & -1 & -0.046 & 0.291 & 0.928 & 1.683 & 0.986 & 2.073\\
		\rowcolor{Gray}
		4 & 1000 & 2 & IMSE1 & 0 & -0.002 & 0.243 & 0.928 & 1.333 & 0.986 & 1.643\\
		\rowcolor{Gray}
		4 & 1000 & 2 & IMSE1 & 1 & 0.019 & 0.305 & 0.928 & 1.725 & 0.986 & 2.125\\
		4 & 1000 & 2 & IMSE2 & -1 & -0.120 & 0.269 & 0.810 & 1.270 & 0.949 & 1.600\\
		4 & 1000 & 2 & IMSE2 & 0 & 0.001 & 0.206 & 0.810 & 1.012 & 0.949 & 1.275\\
		4 & 1000 & 2 & IMSE2 & 1 & 0.106 & 0.276 & 0.810 & 1.301 & 0.949 & 1.638\\
		\bottomrule
	\end{tabular}
	
	\begin{flushleft}
		Note: 
		The gray rows are the simulation results for the LLR and LQR estimation based on the undersmoothing methods.
		
		\bigskip  
		
		Abbreviations: 
		UCB = the uniform confidence band;
		RMSE = the root mean squared error;
		UCP = the (empirical) uniform coverage probability;
		US1 = the rule-of-thumb undersmoothing $\hat h_{\US}$ for the LLR estimation;
		IMSE1 = the IMSE-optimal bandwidth $\hat h_{\mathrm{LL}}$ for the LLR estimation;
		IMSE2 = the IMSE-optimal bandwidth $\hat h_{\mathrm{LQ}}$ for the LQR estimation.
	\end{flushleft}
\end{table}

\begin{table}[p]
	\caption{Additional Monte Carlo simulation results for CATT: the number of covariates ($k = 5$)} \label{table:MC3}
	\centering
	\begin{tabular}[t]{rrrrrrrrrrr}
		\toprule
		\multicolumn{7}{c}{ } & \multicolumn{2}{c}{Analytical UCB} & \multicolumn{2}{c}{Bootstrap UCB} \\
		\cmidrule(l{3pt}r{3pt}){8-9} \cmidrule(l{3pt}r{3pt}){10-11}
		$\mathcal{T}$ & $n$ & $p$ & Bandwidth & $z$ & Bias & RMSE & UCP & Length & UCP & Length\\
		\midrule
		\rowcolor{Gray}
		2 & 500 & 1 & US1 & -1 & -0.011 & 0.388 & 0.932 & 2.189 & 0.953 & 2.273\\
		\rowcolor{Gray}
		2 & 500 & 1 & US1 & 0 & 0.006 & 0.267 & 0.932 & 1.670 & 0.953 & 1.733\\
		\rowcolor{Gray}
		2 & 500 & 1 & US1 & 1 & -0.009 & 0.377 & 0.932 & 2.101 & 0.953 & 2.180\\
		2 & 500 & 1 & IMSE1 & -1 & -0.039 & 0.311 & 0.719 & 1.543 & 0.806 & 1.678\\
		2 & 500 & 1 & IMSE1 & 0 & 0.030 & 0.210 & 0.719 & 1.185 & 0.806 & 1.288\\
		2 & 500 & 1 & IMSE1 & 1 & -0.023 & 0.305 & 0.719 & 1.486 & 0.806 & 1.616\\
		\rowcolor{Gray}
		2 & 500 & 2 & IMSE1 & -1 & -0.054 & 0.389 & 0.921 & 2.007 & 0.950 & 2.221\\
		\rowcolor{Gray}
		2 & 500 & 2 & IMSE1 & 0 & 0.003 & 0.270 & 0.921 & 1.528 & 0.950 & 1.691\\
		\rowcolor{Gray}
		2 & 500 & 2 & IMSE1 & 1 & 0.042 & 0.373 & 0.921 & 1.924 & 0.950 & 2.129\\
		2 & 500 & 2 & IMSE2 & -1 & -0.160 & 0.367 & 0.817 & 1.570 & 0.906 & 1.795\\
		2 & 500 & 2 & IMSE2 & 0 & 0.025 & 0.234 & 0.817 & 1.202 & 0.906 & 1.374\\
		2 & 500 & 2 & IMSE2 & 1 & 0.092 & 0.330 & 0.817 & 1.511 & 0.906 & 1.726\\
		\midrule 
		\rowcolor{Gray}
		2 & 1000 & 1 & US1 & -1 & -0.009 & 0.277 & 0.949 & 1.768 & 0.959 & 1.827\\
		\rowcolor{Gray}
		2 & 1000 & 1 & US1 & 0 & 0.012 & 0.220 & 0.949 & 1.345 & 0.959 & 1.390\\
		\rowcolor{Gray}
		2 & 1000 & 1 & US1 & 1 & 0.006 & 0.290 & 0.949 & 1.722 & 0.959 & 1.780\\
		2 & 1000 & 1 & IMSE1 & -1 & -0.037 & 0.217 & 0.750 & 1.212 & 0.830 & 1.299\\
		2 & 1000 & 1 & IMSE1 & 0 & 0.032 & 0.171 & 0.750 & 0.924 & 0.830 & 0.991\\
		2 & 1000 & 1 & IMSE1 & 1 & -0.009 & 0.223 & 0.750 & 1.180 & 0.830 & 1.265\\
		\rowcolor{Gray}
		2 & 1000 & 2 & IMSE1 & -1 & -0.036 & 0.273 & 0.929 & 1.570 & 0.966 & 1.727\\
		\rowcolor{Gray}
		2 & 1000 & 2 & IMSE1 & 0 & 0.009 & 0.216 & 0.929 & 1.190 & 0.966 & 1.309\\
		\rowcolor{Gray}
		2 & 1000 & 2 & IMSE1 & 1 & 0.037 & 0.278 & 0.929 & 1.526 & 0.966 & 1.679\\
		2 & 1000 & 2 & IMSE2 & -1 & -0.143 & 0.268 & 0.830 & 1.163 & 0.922 & 1.324\\
		2 & 1000 & 2 & IMSE2 & 0 & 0.030 & 0.180 & 0.830 & 0.884 & 0.922 & 1.007\\
		2 & 1000 & 2 & IMSE2 & 1 & 0.091 & 0.242 & 0.830 & 1.131 & 0.922 & 1.288\\
		\bottomrule
	\end{tabular}
	
	\begin{flushleft}
		Note: 
		The gray rows are the simulation results for the LLR and LQR estimation based on the undersmoothing methods.
		
		\bigskip  
		
		Abbreviations: 
		UCB = the uniform confidence band;
		RMSE = the root mean squared error;
		UCP = the (empirical) uniform coverage probability;
		US1 = the rule-of-thumb undersmoothing $\hat h_{\US}$ for the LLR estimation;
		IMSE1 = the IMSE-optimal bandwidth $\hat h_{\mathrm{LL}}$ for the LLR estimation;
		IMSE2 = the IMSE-optimal bandwidth $\hat h_{\mathrm{LQ}}$ for the LQR estimation.
	\end{flushleft}
\end{table}

\begin{table}[p]
	\caption{Additional Monte Carlo simulation results for CATT: heteroscedasticity} \label{table:MC4}
	\centering
	\begin{tabular}[t]{rrrrrrrrrrr}
		\toprule
		\multicolumn{7}{c}{ } & \multicolumn{2}{c}{Analytical UCB} & \multicolumn{2}{c}{Bootstrap UCB} \\
		\cmidrule(l{3pt}r{3pt}){8-9} \cmidrule(l{3pt}r{3pt}){10-11}
		$\mathcal{T}$ & $n$ & $p$ & Bandwidth & $z$ & Bias & RMSE & UCP & Length & UCP & Length\\
		\midrule
		\rowcolor{Gray}
		2 & 500 & 1 & US1 & -1 & -0.048 & 0.312 & 0.919 & 1.807 & 0.937 & 1.84\\
		\rowcolor{Gray}
		2 & 500 & 1 & US1 & 0 & 0.028 & 0.319 & 0.919 & 1.827 & 0.937 & 1.86\\
		\rowcolor{Gray}
		2 & 500 & 1 & US1 & 1 & -0.031 & 0.529 & 0.919 & 2.868 & 0.937 & 2.92\\
		2 & 500 & 1 & IMSE1 & -1 & -0.072 & 0.257 & 0.675 & 1.275 & 0.769 & 1.38\\
		2 & 500 & 1 & IMSE1 & 0 & 0.046 & 0.256 & 0.675 & 1.296 & 0.769 & 1.40\\
		2 & 500 & 1 & IMSE1 & 1 & -0.023 & 0.423 & 0.675 & 2.029 & 0.769 & 2.19\\
		\rowcolor{Gray}
		2 & 500 & 2 & IMSE1 & -1 & -0.104 & 0.327 & 0.899 & 1.657 & 0.939 & 1.79\\
		\rowcolor{Gray}
		2 & 500 & 2 & IMSE1 & 0 & 0.025 & 0.326 & 0.899 & 1.675 & 0.939 & 1.81\\
		\rowcolor{Gray}
		2 & 500 & 2 & IMSE1 & 1 & 0.032 & 0.511 & 0.899 & 2.636 & 0.939 & 2.85\\
		2 & 500 & 2 & IMSE2 & -1 & -0.203 & 0.342 & 0.792 & 1.329 & 0.887 & 1.50\\
		2 & 500 & 2 & IMSE2 & 0 & 0.044 & 0.289 & 0.792 & 1.350 & 0.887 & 1.52\\
		2 & 500 & 2 & IMSE2 & 1 & 0.090 & 0.453 & 0.792 & 2.121 & 0.887 & 2.39\\
		\midrule 
		\rowcolor{Gray}
		2 & 1000 & 1 & US1 & -1 & -0.012 & 0.241 & 0.936 & 1.458 & 0.940 & 1.48\\
		\rowcolor{Gray}
		2 & 1000 & 1 & US1 & 0 & 0.017 & 0.232 & 0.936 & 1.470 & 0.940 & 1.49\\
		\rowcolor{Gray}
		2 & 1000 & 1 & US1 & 1 & -0.018 & 0.399 & 0.936 & 2.338 & 0.940 & 2.37\\
		2 & 1000 & 1 & IMSE1 & -1 & -0.037 & 0.190 & 0.739 & 0.997 & 0.800 & 1.06\\
		2 & 1000 & 1 & IMSE1 & 0 & 0.034 & 0.185 & 0.739 & 1.009 & 0.800 & 1.07\\
		2 & 1000 & 1 & IMSE1 & 1 & -0.021 & 0.309 & 0.739 & 1.602 & 0.800 & 1.70\\
		\rowcolor{Gray}
		2 & 1000 & 2 & IMSE1 & -1 & -0.042 & 0.237 & 0.922 & 1.294 & 0.953 & 1.40\\
		\rowcolor{Gray}
		2 & 1000 & 2 & IMSE1 & 0 & 0.013 & 0.231 & 0.922 & 1.303 & 0.953 & 1.41\\
		\rowcolor{Gray}
		2 & 1000 & 2 & IMSE1 & 1 & 0.023 & 0.375 & 0.922 & 2.079 & 0.953 & 2.24\\
		2 & 1000 & 2 & IMSE2 & -1 & -0.152 & 0.252 & 0.786 & 0.978 & 0.902 & 1.10\\
		2 & 1000 & 2 & IMSE2 & 0 & 0.032 & 0.200 & 0.786 & 0.989 & 0.902 & 1.11\\
		2 & 1000 & 2 & IMSE2 & 1 & 0.082 & 0.323 & 0.786 & 1.575 & 0.902 & 1.77\\
		\bottomrule
	\end{tabular}
	
	\begin{flushleft}
		Note: 
		The gray rows are the simulation results for the LLR and LQR estimation based on the undersmoothing methods.
		
		\bigskip  
		
		Abbreviations: 
		UCB = the uniform confidence band;
		RMSE = the root mean squared error;
		UCP = the (empirical) uniform coverage probability;
		US1 = the rule-of-thumb undersmoothing $\hat h_{\US}$ for the LLR estimation;
		IMSE1 = the IMSE-optimal bandwidth $\hat h_{\mathrm{LL}}$ for the LLR estimation;
		IMSE2 = the IMSE-optimal bandwidth $\hat h_{\mathrm{LQ}}$ for the LQR estimation.
	\end{flushleft}
\end{table}

\begin{table}[p]
	\caption{Additional Monte Carlo simulation results for CATT: the linear potential outcome equation} \label{table:MC5}
	\centering
	\begin{tabular}[t]{rrrrrrrrrrr}
		\toprule
		\multicolumn{7}{c}{ } & \multicolumn{2}{c}{Analytical UCB} & \multicolumn{2}{c}{Bootstrap UCB} \\
		\cmidrule(l{3pt}r{3pt}){8-9} \cmidrule(l{3pt}r{3pt}){10-11}
		$\mathcal{T}$ & $n$ & $p$ & Bandwidth & $z$ & Bias & RMSE & UCP & Length & UCP & Length\\
		\midrule
		\rowcolor{Gray}
		2 & 500 & 1 & US1 & -1 & -0.007 & 0.299 & 0.923 & 1.485 & 0.957 & 1.652\\
		\rowcolor{Gray}
		2 & 500 & 1 & US1 & 0 & 0.026 & 0.232 & 0.923 & 1.152 & 0.957 & 1.281\\
		\rowcolor{Gray}
		2 & 500 & 1 & US1 & 1 & 0.004 & 0.305 & 0.923 & 1.475 & 0.957 & 1.639\\
		2 & 500 & 1 & IMSE1 & -1 & 0.023 & 0.244 & 0.851 & 1.043 & 0.954 & 1.283\\
		2 & 500 & 1 & IMSE1 & 0 & 0.048 & 0.192 & 0.851 & 0.811 & 0.954 & 0.998\\
		2 & 500 & 1 & IMSE1 & 1 & 0.026 & 0.257 & 0.851 & 1.038 & 0.954 & 1.277\\
		\rowcolor{Gray}
		2 & 500 & 2 & IMSE1 & -1 & -0.025 & 0.298 & 0.883 & 1.360 & 0.959 & 1.601\\
		\rowcolor{Gray}
		2 & 500 & 2 & IMSE1 & 0 & 0.013 & 0.239 & 0.883 & 1.055 & 0.959 & 1.242\\
		\rowcolor{Gray}
		2 & 500 & 2 & IMSE1 & 1 & -0.008 & 0.295 & 0.883 & 1.348 & 0.959 & 1.586\\
		2 & 500 & 2 & IMSE2 & -1 & -0.024 & 0.285 & 0.865 & 1.295 & 0.947 & 1.540\\
		2 & 500 & 2 & IMSE2 & 0 & 0.009 & 0.230 & 0.865 & 1.005 & 0.947 & 1.195\\
		2 & 500 & 2 & IMSE2 & 1 & -0.004 & 0.284 & 0.865 & 1.284 & 0.947 & 1.526\\
		\midrule 
		\rowcolor{Gray}
		2 & 1000 & 1 & US1 & -1 & 0.015 & 0.220 & 0.930 & 1.117 & 0.960 & 1.229\\
		\rowcolor{Gray}
		2 & 1000 & 1 & US1 & 0 & 0.016 & 0.161 & 0.930 & 0.861 & 0.960 & 0.947\\
		\rowcolor{Gray}
		2 & 1000 & 1 & US1 & 1 & 0.008 & 0.215 & 0.930 & 1.113 & 0.960 & 1.223\\
		2 & 1000 & 1 & IMSE1 & -1 & 0.041 & 0.182 & 0.860 & 0.757 & 0.949 & 0.922\\
		2 & 1000 & 1 & IMSE1 & 0 & 0.038 & 0.135 & 0.860 & 0.585 & 0.949 & 0.712\\
		2 & 1000 & 1 & IMSE1 & 1 & 0.025 & 0.177 & 0.860 & 0.755 & 0.949 & 0.920\\
		\rowcolor{Gray}
		2 & 1000 & 2 & IMSE1 & -1 & -0.002 & 0.209 & 0.898 & 0.986 & 0.961 & 1.158\\
		\rowcolor{Gray}
		2 & 1000 & 2 & IMSE1 & 0 & 0.004 & 0.163 & 0.898 & 0.760 & 0.961 & 0.892\\
		\rowcolor{Gray}
		2 & 1000 & 2 & IMSE1 & 1 & -0.001 & 0.203 & 0.898 & 0.981 & 0.961 & 1.151\\
		2 & 1000 & 2 & IMSE2 & -1 & 0.000 & 0.199 & 0.879 & 0.900 & 0.953 & 1.079\\
		2 & 1000 & 2 & IMSE2 & 0 & 0.001 & 0.155 & 0.879 & 0.694 & 0.953 & 0.831\\
		2 & 1000 & 2 & IMSE2 & 1 & 0.000 & 0.192 & 0.879 & 0.896 & 0.953 & 1.073\\
		\bottomrule
	\end{tabular}
	
	\begin{flushleft}
		Note: 
		The gray rows are the simulation results for the LLR and LQR estimation based on the undersmoothing methods.
		
		\bigskip  
		
		Abbreviations: 
		UCB = the uniform confidence band;
		RMSE = the root mean squared error;
		UCP = the (empirical) uniform coverage probability;
		US1 = the rule-of-thumb undersmoothing $\hat h_{\US}$ for the LLR estimation;
		IMSE1 = the IMSE-optimal bandwidth $\hat h_{\mathrm{LL}}$ for the LLR estimation;
		IMSE2 = the IMSE-optimal bandwidth $\hat h_{\mathrm{LQ}}$ for the LQR estimation.
	\end{flushleft}
\end{table}

\begin{table}[p]
	\caption{Monte Carlo simulation results for $\theta_{\es}(0, z)$: the nonlinear potential outcome equation} \label{table:MC6}
	\centering
	\begin{tabular}[t]{rrrrrrrrrrr}
		\toprule
		\multicolumn{7}{c}{ } & \multicolumn{2}{c}{Analytical UCB} & \multicolumn{2}{c}{Bootstrap UCB} \\
		\cmidrule(l{3pt}r{3pt}){8-9} \cmidrule(l{3pt}r{3pt}){10-11}
		$\mathcal{T}$ & $n$ & $p$ & Bandwidth & $z$ & Bias & RMSE & UCP & Length & UCP & Length\\
		\midrule
		\rowcolor{Gray}
		4 & 500 & 1 & US1 & -1 & -0.008 & 0.294 & 0.954 & 1.731 & 0.968 & 1.791\\
		\rowcolor{Gray}
		4 & 500 & 1 & US1 & 0 & -0.004 & 0.229 & 0.954 & 1.367 & 0.968 & 1.413\\
		\rowcolor{Gray}
		4 & 500 & 1 & US1 & 1 & -0.001 & 0.296 & 0.954 & 1.766 & 0.968 & 1.825\\
		4 & 500 & 1 & IMSE1 & -1 & -0.014 & 0.228 & 0.766 & 1.208 & 0.786 & 1.241\\
		4 & 500 & 1 & IMSE1 & 0 & 0.006 & 0.174 & 0.766 & 0.961 & 0.786 & 0.987\\
		4 & 500 & 1 & IMSE1 & 1 & -0.003 & 0.234 & 0.766 & 1.228 & 0.786 & 1.261\\
		\rowcolor{Gray}
		4 & 500 & 2 & IMSE1 & -1 & -0.033 & 0.292 & 0.940 & 1.606 & 0.970 & 1.733\\
		\rowcolor{Gray}
		4 & 500 & 2 & IMSE1 & 0 & -0.005 & 0.230 & 0.940 & 1.272 & 0.970 & 1.372\\
		\rowcolor{Gray}
		4 & 500 & 2 & IMSE1 & 1 & 0.028 & 0.293 & 0.940 & 1.635 & 0.970 & 1.763\\
		4 & 500 & 2 & IMSE2 & -1 & -0.118 & 0.263 & 0.856 & 1.177 & 0.914 & 1.274\\
		4 & 500 & 2 & IMSE2 & 0 & 0.004 & 0.186 & 0.856 & 0.940 & 0.914 & 1.016\\
		4 & 500 & 2 & IMSE2 & 1 & 0.101 & 0.260 & 0.856 & 1.198 & 0.914 & 1.295\\
		\midrule
		\rowcolor{Gray}
		4 & 1000 & 1 & US1 & -1 & -0.009 & 0.205 & 0.970 & 1.391 & 0.968 & 1.414\\
		\rowcolor{Gray}
		4 & 1000 & 1 & US1 & 0 & 0.007 & 0.168 & 0.970 & 1.095 & 0.968 & 1.113\\
		\rowcolor{Gray}
		4 & 1000 & 1 & US1 & 1 & 0.000 & 0.237 & 0.970 & 1.433 & 0.968 & 1.456\\
		4 & 1000 & 1 & IMSE1 & -1 & -0.014 & 0.157 & 0.776 & 0.945 & 0.792 & 0.961\\
		4 & 1000 & 1 & IMSE1 & 0 & 0.012 & 0.132 & 0.776 & 0.749 & 0.792 & 0.762\\
		4 & 1000 & 1 & IMSE1 & 1 & -0.003 & 0.180 & 0.776 & 0.971 & 0.792 & 0.987\\
		\rowcolor{Gray}
		4 & 1000 & 2 & IMSE1 & -1 & -0.025 & 0.197 & 0.960 & 1.246 & 0.976 & 1.321\\
		\rowcolor{Gray}
		4 & 1000 & 2 & IMSE1 & 0 & 0.005 & 0.165 & 0.960 & 0.983 & 0.976 & 1.043\\
		\rowcolor{Gray}
		4 & 1000 & 2 & IMSE1 & 1 & 0.019 & 0.227 & 0.960 & 1.282 & 0.976 & 1.360\\
		4 & 1000 & 2 & IMSE2 & -1 & -0.103 & 0.187 & 0.848 & 0.866 & 0.884 & 0.927\\
		4 & 1000 & 2 & IMSE2 & 0 & 0.011 & 0.135 & 0.848 & 0.688 & 0.884 & 0.736\\
		4 & 1000 & 2 & IMSE2 & 1 & 0.082 & 0.194 & 0.848 & 0.890 & 0.884 & 0.952\\
		\bottomrule
	\end{tabular}
	
	\begin{flushleft}
		Note: 
		The gray rows are the simulation results for the LLR and LQR estimation based on the undersmoothing methods.
		
		\bigskip  
		
		Abbreviations: 
		UCB = the uniform confidence band;
		RMSE = the root mean squared error;
		UCP = the (empirical) uniform coverage probability;
		US1 = the rule-of-thumb undersmoothing $\hat h_{\US}$ for the LLR estimation;
		IMSE1 = the IMSE-optimal bandwidth $\hat h_{\mathrm{LL}}$ for the LLR estimation;
		IMSE2 = the IMSE-optimal bandwidth $\hat h_{\mathrm{LQ}}$ for the LQR estimation.
	\end{flushleft}
\end{table}

\begin{table}[p]
	\caption{Monte Carlo simulation results for $\theta_{\es}(0, z)$: the linear potential outcome equation} \label{table:MC7}
	\centering
	\begin{tabular}[t]{rrrrrrrrrrr}
		\toprule
		\multicolumn{7}{c}{ } & \multicolumn{2}{c}{Analytical UCB} & \multicolumn{2}{c}{Bootstrap UCB} \\
		\cmidrule(l{3pt}r{3pt}){8-9} \cmidrule(l{3pt}r{3pt}){10-11}
		$\mathcal{T}$ & $n$ & $p$ & Bandwidth & $z$ & Bias & RMSE & UCP & Length & UCP & Length\\
		\midrule
		\rowcolor{Gray}
		4 & 500 & 1 & US1 & -1 & 0.010 & 0.213 & 0.928 & 1.020 & 0.946 & 1.055\\
		\rowcolor{Gray}
		4 & 500 & 1 & US1 & 0 & 0.003 & 0.163 & 0.928 & 0.817 & 0.946 & 0.844\\
		\rowcolor{Gray}
		4 & 500 & 1 & US1 & 1 & 0.008 & 0.222 & 0.928 & 1.053 & 0.946 & 1.088\\
		4 & 500 & 1 & IMSE1 & -1 & 0.022 & 0.175 & 0.858 & 0.712 & 0.878 & 0.743\\
		4 & 500 & 1 & IMSE1 & 0 & 0.014 & 0.131 & 0.858 & 0.571 & 0.878 & 0.595\\
		4 & 500 & 1 & IMSE1 & 1 & 0.016 & 0.184 & 0.858 & 0.734 & 0.878 & 0.765\\
		\rowcolor{Gray}
		4 & 500 & 2 & IMSE1 & -1 & 0.005 & 0.210 & 0.894 & 0.931 & 0.938 & 1.018\\
		\rowcolor{Gray}
		4 & 500 & 2 & IMSE1 & 0 & -0.004 & 0.169 & 0.894 & 0.748 & 0.938 & 0.818\\
		\rowcolor{Gray}
		4 & 500 & 2 & IMSE1 & 1 & 0.003 & 0.213 & 0.894 & 0.961 & 0.938 & 1.049\\
		4 & 500 & 2 & IMSE2 & -1 & 0.006 & 0.205 & 0.888 & 0.918 & 0.932 & 1.005\\
		4 & 500 & 2 & IMSE2 & 0 & -0.005 & 0.165 & 0.888 & 0.738 & 0.932 & 0.808\\
		4 & 500 & 2 & IMSE2 & 1 & 0.004 & 0.208 & 0.888 & 0.947 & 0.932 & 1.036\\
		\midrule 
		\rowcolor{Gray}
		4 & 1000 & 1 & US1 & -1 & 0.003 & 0.141 & 0.944 & 0.746 & 0.948 & 0.763\\
		\rowcolor{Gray}
		4 & 1000 & 1 & US1 & 0 & 0.008 & 0.118 & 0.944 & 0.594 & 0.948 & 0.608\\
		\rowcolor{Gray}
		4 & 1000 & 1 & US1 & 1 & 0.009 & 0.164 & 0.944 & 0.777 & 0.948 & 0.794\\
		4 & 1000 & 1 & IMSE1 & -1 & 0.019 & 0.121 & 0.842 & 0.503 & 0.856 & 0.520\\
		4 & 1000 & 1 & IMSE1 & 0 & 0.018 & 0.098 & 0.842 & 0.402 & 0.856 & 0.415\\
		4 & 1000 & 1 & IMSE1 & 1 & 0.014 & 0.131 & 0.842 & 0.523 & 0.856 & 0.541\\
		\rowcolor{Gray}
		4 & 1000 & 2 & IMSE1 & -1 & -0.002 & 0.138 & 0.902 & 0.656 & 0.942 & 0.712\\
		\rowcolor{Gray}
		4 & 1000 & 2 & IMSE1 & 0 & 0.003 & 0.120 & 0.902 & 0.524 & 0.942 & 0.569\\
		\rowcolor{Gray}
		4 & 1000 & 2 & IMSE1 & 1 & 0.004 & 0.153 & 0.902 & 0.682 & 0.942 & 0.741\\
		4 & 1000 & 2 & IMSE2 & -1 & -0.001 & 0.136 & 0.882 & 0.634 & 0.932 & 0.690\\
		4 & 1000 & 2 & IMSE2 & 0 & 0.002 & 0.118 & 0.882 & 0.506 & 0.932 & 0.551\\
		4 & 1000 & 2 & IMSE2 & 1 & 0.003 & 0.146 & 0.882 & 0.658 & 0.932 & 0.716\\
		\bottomrule
	\end{tabular}
	
	\begin{flushleft}
		Note: 
		The gray rows are the simulation results for the LLR and LQR estimation based on the undersmoothing methods.
		
		\bigskip  
		
		Abbreviations: 
		UCB = the uniform confidence band;
		RMSE = the root mean squared error;
		UCP = the (empirical) uniform coverage probability;
		US1 = the rule-of-thumb undersmoothing $\hat h_{\US}$ for the LLR estimation;
		IMSE1 = the IMSE-optimal bandwidth $\hat h_{\mathrm{LL}}$ for the LLR estimation;
		IMSE2 = the IMSE-optimal bandwidth $\hat h_{\mathrm{LQ}}$ for the LQR estimation.
	\end{flushleft}
\end{table}



\begin{figure}[t]
	\centering
	\begin{minipage}[b]{\hsize}
		\centering
		\includegraphics[width = 0.8\linewidth, bb = 0 0 1080 540]{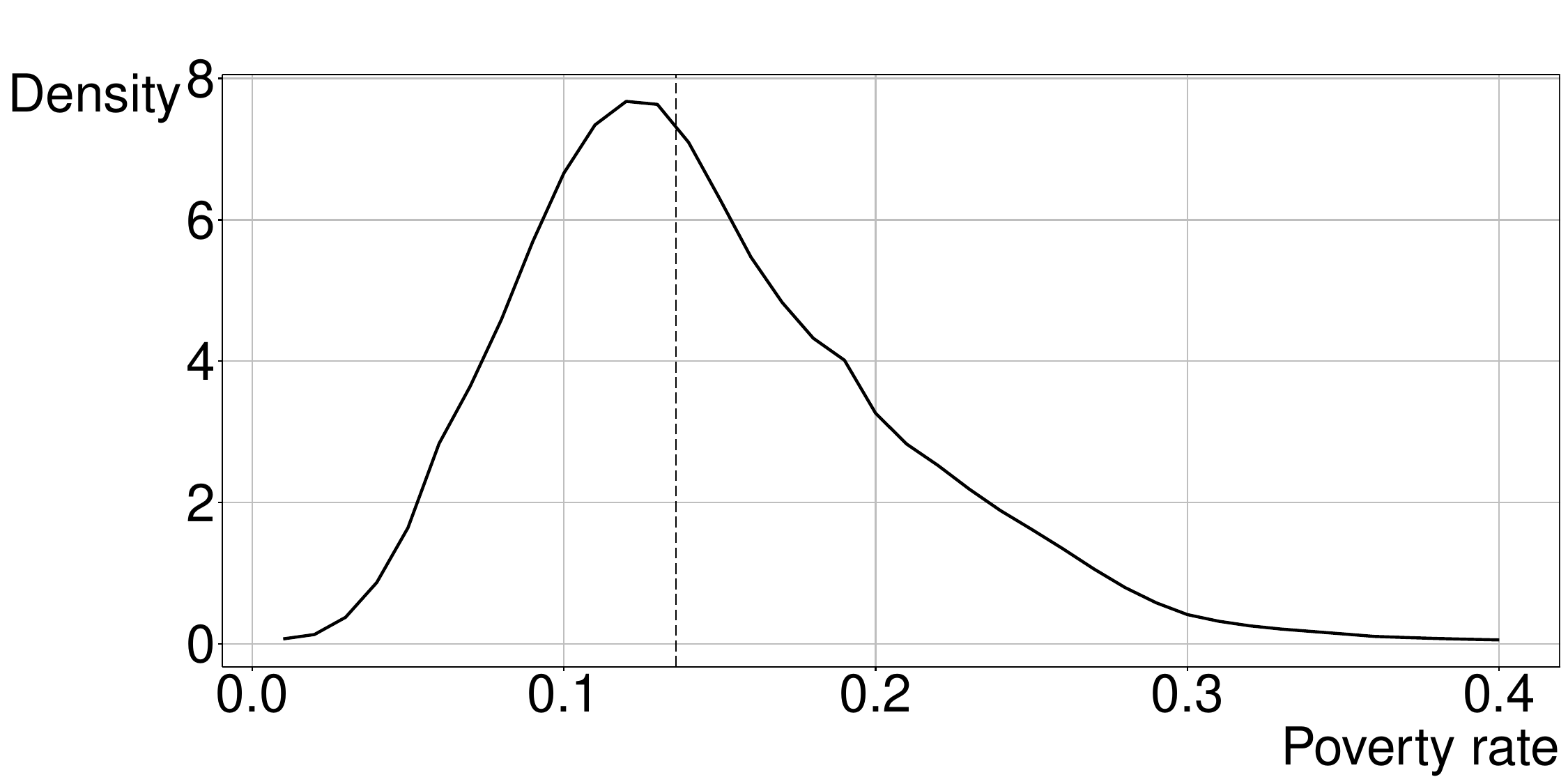}
	\end{minipage}
	\caption{The kernel density of the poverty rate.} \label{fig:density}
\end{figure}

\begin{table}[t]
	\caption{The summary statistics for the pre-treatment variables} \label{table:descriptive}
	(a) The treated groups
	\centering
	\begin{tabular}[t]{lrrrrrr}
		\toprule
		\multicolumn{1}{c}{ } & \multicolumn{2}{c}{Group 2004} & \multicolumn{2}{c}{Group 2006} & \multicolumn{2}{c}{Group 2007} \\
		\cmidrule(l{3pt}r{3pt}){2-3} \cmidrule(l{3pt}r{3pt}){4-5} \cmidrule(l{3pt}r{3pt}){6-7}
		& \multicolumn{1}{c}{Mean} & \multicolumn{1}{c}{SD} & \multicolumn{1}{c}{Mean} & \multicolumn{1}{c}{SD} & \multicolumn{1}{c}{Mean} & \multicolumn{1}{c}{SD} \\
		\midrule
		Midwest & 1 & 0 & 0.704 & 0.458 & 0.483 & 0.500\\
		South & 0 & 0 & 0.296 & 0.458 & 0.301 & 0.459\\
		West & 0 & 0 & 0 & 0 & 0.216 & 0.412\\
		Poverty rate & 0.117 & 0.043 & 0.119 & 0.051 & 0.138 & 0.051\\
		Population (10000s) & 124.087 & 546.636 & 107.631 & 212.113 & 84.142 & 181.244\\
		HS graduates & 0.606 & 0.042 & 0.600 & 0.045 & 0.577 & 0.062\\
		White & 0.924 & 0.083 & 0.901 & 0.110 & 0.885 & 0.132\\
		Median income (1000s) & 36.443 & 8.028 & 34.958 & 7.570 & 33.080 & 8.164\\
		\midrule 
		Observations & \multicolumn{2}{c}{100} & \multicolumn{2}{c}{223} & \multicolumn{2}{c}{584} \\ 
		\bottomrule
	\end{tabular}
	
	\bigskip \bigskip 
	(b) The not-yet-treated (NYT) groups
	\begin{tabular}[t]{lrrrrrr}
		\toprule
		\multicolumn{1}{c}{ } & \multicolumn{2}{c}{NYT 2004} & \multicolumn{2}{c}{NYT 2006} & \multicolumn{2}{c}{NYT 2007} \\
		\cmidrule(l{3pt}r{3pt}){2-3} \cmidrule(l{3pt}r{3pt}){4-5} \cmidrule(l{3pt}r{3pt}){6-7}
		& \multicolumn{1}{c}{Mean} & \multicolumn{1}{c}{SD} & \multicolumn{1}{c}{Mean} & \multicolumn{1}{c}{SD} & \multicolumn{1}{c}{Mean} & \multicolumn{1}{c}{SD} \\
		\midrule
		Midwest & 0.413 & 0.492 & 0.379 & 0.485 & 0.336 & 0.472\\
		South & 0.484 & 0.500 & 0.506 & 0.500 & 0.593 & 0.492\\
		West & 0.103 & 0.304 & 0.115 & 0.319 & 0.072 & 0.258\\
		Poverty rate & 0.148 & 0.062 & 0.152 & 0.062 & 0.157 & 0.065\\
		Population (10000s) & 67.174 & 168.776 & 62.573 & 162.551 & 53.425 & 153.095\\
		HS graduates & 0.564 & 0.074 & 0.560 & 0.075 & 0.553 & 0.079\\
		White & 0.850 & 0.153 & 0.844 & 0.156 & 0.826 & 0.162\\
		Median income (1000s) & 32.521 & 7.679 & 32.244 & 7.644 & 31.889 & 7.387\\
		\midrule 
		Observations & \multicolumn{2}{c}{2184} & \multicolumn{2}{c}{1961} & \multicolumn{2}{c}{1377} \\
		\bottomrule
	\end{tabular}
\end{table}

\begin{figure}[ht]
	\centering
	\begin{minipage}[t]{\textwidth}
		\centering
		\includegraphics[width=\linewidth, bb=0 0 2160 1080]{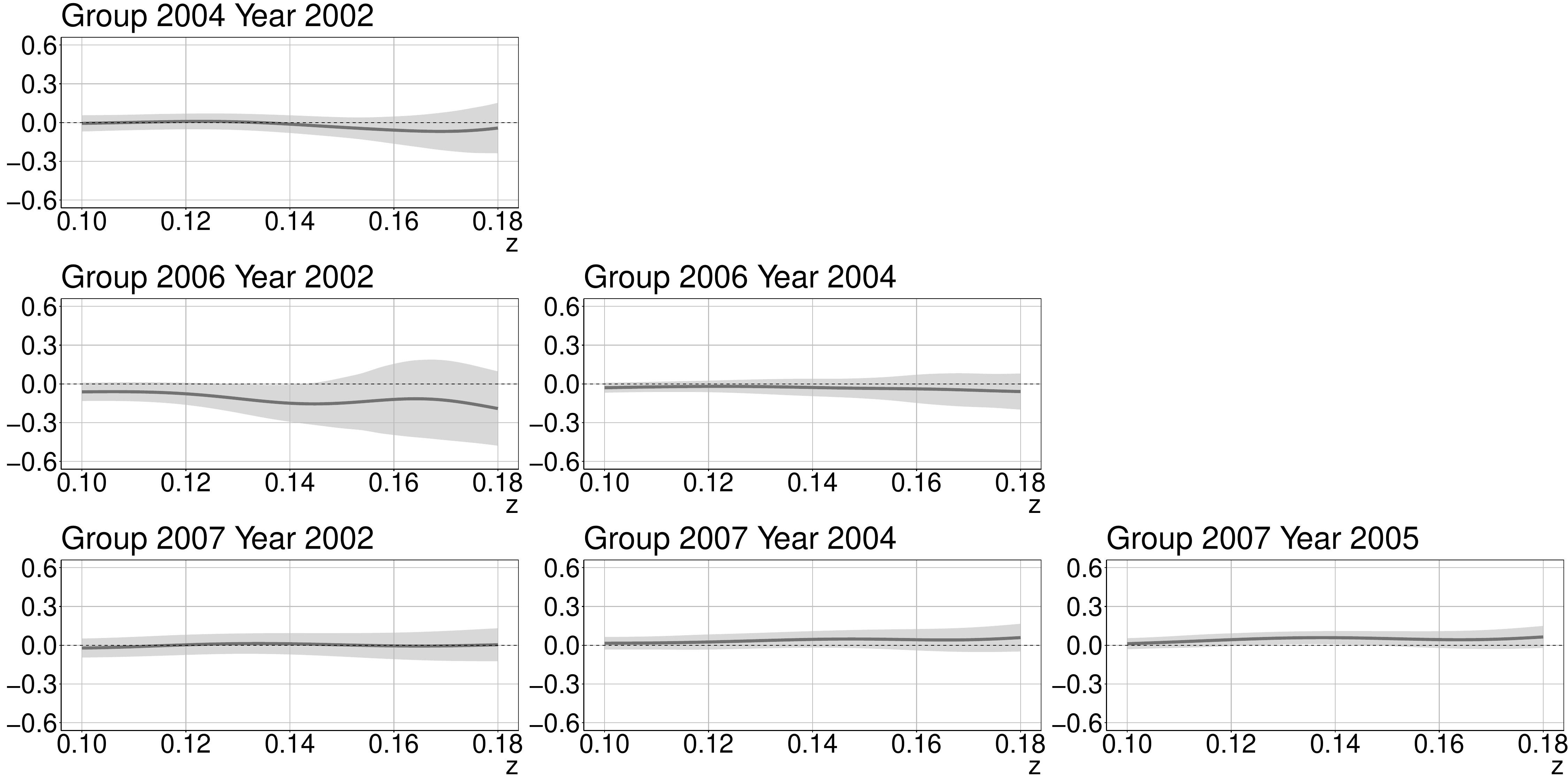}
		\subcaption{CATT}
	\end{minipage}
	
	\vspace{-5em} 
	
	\begin{minipage}[t]{\textwidth}
		\centering
		\includegraphics[width=\linewidth, bb=0 0 2160 1080]{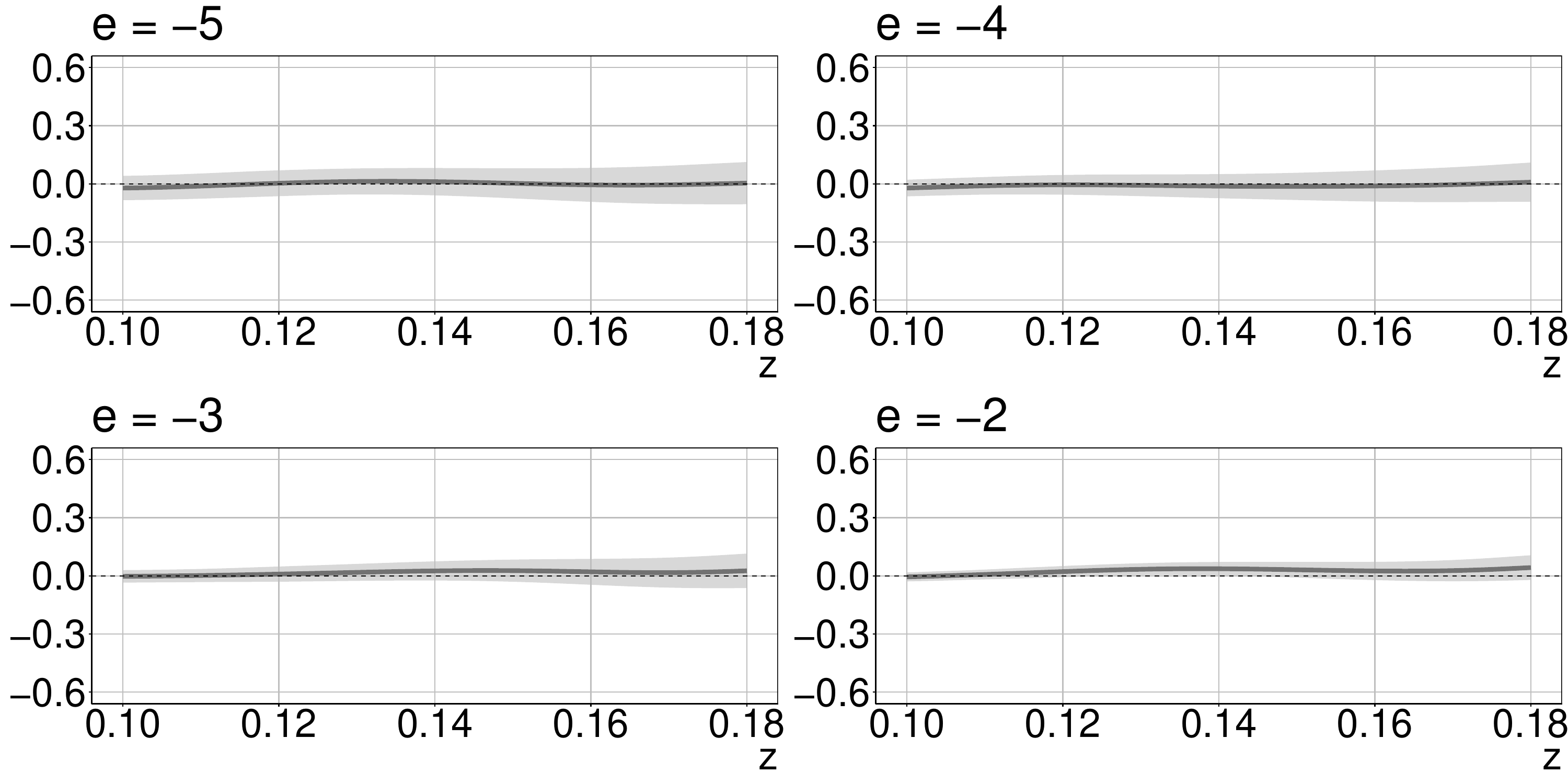}
		\subcaption{The event-study-type conditional average treatment effect}
	\end{minipage}
	
	\caption{The LQR estimates and 95\% uniform confidence bands constructed with weighted/multiplier bootstrapping in the pre-treatment periods.}
	\label{fig:boot-pre}
\end{figure}

\end{document}